\documentclass[12pt]{article}

\usepackage{sdhcmds}
\usepackage{sdhpapers}
\AtEveryBibitem{%
    \ifentrytype{article}
        {\iffieldundef{journaltitle}
            {}
            {\clearfield{doi}%
             \clearfield{url}%
             \clearfield{eprint}}}
        {}%
}
\usepackage{amsmath}
\usepackage[ruled,vlined]{algorithm2e}  

\begin{document}

\title{Experimental Design for Policy Choice}
\author{Samuel D. Higbee}
\institute{
    University of North Carolina at Chapel Hill
}
\date{September 9, 2026}

\abstract{
We show how to optimally design experiments when the
resulting data will be used to choose a welfare-maximizing policy subject to constraints.
A decision maker seeks to maximize Bayes expected welfare by choosing a policy
whose effects depend on an unknown finite-dimensional parameter.
The decision maker has access to a first wave of experimental data with a fixed design
but may choose the design of a second wave that will be collected before choosing the policy.
The resulting experimental design--policy choice problem 
is a very high-dimensional dynamic program that is generally intractable
in finite samples.
We propose a tractable approximation based on the limit experiment
and show it is asymptotically optimal using a new asymptotic representation
theorem for adaptive experiments with continuous treatments.
We apply the method to a conditional cash transfer experiment and demonstrate the potential for large
gains from tailoring the experiment to the policy choice.
}

\acknowledgements{
    I thank St\'ephane Bonhomme, Max Tabord-Meehan, Guillaume Pouliot, 
    Arun Chandrasekhar, Alex Torgovitsky, Lars Peter Hansen,
    Azeem Shaikh, Giovanni Compiani, Ali Horta\c{c}su,
    and numerous seminar participants for their helpful comments.
    Author's email: \texttt{sdhigbee@unc.edu}
}

\maketitle

\newpage

\section{Introduction}
\label{section:introduction}

Economists often use experiments to evaluate, compare, and choose policies.
Given data from an experiment and an econometric model,
decision makers can evaluate counterfactual policies and choose the one
that maximizes some measure of welfare, subject to budgetary or other constraints.
Experimental design can play an important role in this process,
since some experimental designs will be more informative about the optimal policy than others.
Such experiments should be preferred when the ultimate goal of the experiment is to 
inform policy choice.

Despite the fact that choosing policies is a ubiquitous objective for decision 
makers, 
very little is known about how to optimally design experiments for this purpose.
Instead, much of the guidance on experimental design assumes the objective of 
maximizing the precision of parameter or treatment effect estimates
(see \textcite{atheyEconometricsRandomizedExperiments2017},
as well as Section \ref{section:literature} below for a review).
However, not all parameters are equally important for choosing policies.
For example, estimating the average effect of a welfare benefits program may not
be helpful for choosing how benefits should vary based on a recipient's income or
number of dependents.
Instead, the experiment should focus on learning about the effects of
counterfactual policies on welfare,
prioritizing policies that are more likely to be welfare-enhancing.

This paper provides a method for designing experiments specifically for the 
purpose of informing policy choice.
Given an econometric model, a welfare objective, and a set of feasible policies,
our proposed method uses a first wave of experimental data to characterize which 
model parameters are most valuable to learn about for the purpose of policy 
choice, and then designs a second wave to optimally learn about these
policy-relevant parameters.
The resulting experiment delivers data which is maximally informative
for choosing the best policy.

We consider a decision maker who must choose a policy to maximize 
expected welfare subject to budgetary or other constraints.
The effects of counterfactual policies are described by a structural
econometric model governed by an unknown finite-dimensional parameter.
The decision maker has a first wave of experimental data
and can design and run a second wave before choosing the
policy.
The resulting joint experimental design--policy choice problem is
a very high-dimensional dynamic decision problem.
Since the policy depends on the realized data,
a naive backward-induction approach requires the decision maker to specify which policy 
they would choose for every possible dataset that the experiment could 
generate.
As a practical matter, solving this problem directly is not possible.

We propose a tractable method for finding the optimal experiment using
two approximations to the decision problem.
First, we replace the high-dimensional data with a Gaussian estimate of the model parameters from each wave.
Selecting a design then reduces to choosing the covariance matrix of the second-wave estimate.
Second, we approximate the nonlinear welfare objective with a quadratic function.
Under these two approximations, the policy choice problem
depends only on a Gaussian estimate of a particular linear function of the model parameters---
namely, the marginal effect of making small adjustments to the policy.
These approximations yield a low-dimensional dynamic program that is simple to solve in practice.

We show that our proposed method is asymptotically optimal
in the sense that a decision maker who
(i) uses our method to select the experimental design, and
(ii) uses the resulting data to choose a policy,
will have the highest possible limiting welfare
across all feasible experimental designs.
We show this by validating the two approximations underlying our method
in an asymptotic regime in which both waves of the experiment grow large.
First, we establish a new asymptotic representation theorem for
adaptive experiments with continuous treatments, which justifies the Gaussian information structure.
Second, we show that the constrained policy choice problem is locally approximated by a quadratic program.
We use these results to derive an asymptotic lower bound on the welfare achievable by any experimental design--policy choice procedure
and show that our proposed method achieves this bound.

We demonstrate the method in an application to the Progresa conditional cash
transfer experiment.
We consider a decision maker who must choose a school attendance subsidy policy to maximize the graduation rate of students.
The decision maker has a first wave of data from the Progresa experiment,
which she uses to design a second wave to learn about the optimal subsidy policy.
We estimate that our experimental design
reduces the regret of the downstream policy choice by as much as $70\%$ compared to the original experimental design.
This translates to graduation rates that are about $1.5$ percentage points higher than 
what could be attained by policies estimated with the original Progresa experiment.

Our approach applies to a wide range of models and objectives.
While we focus on parametric models, we do not impose further restrictions on the model beyond standard regularity conditions.
Parametric structural models are widely used in economics to evaluate counterfactual policies
not directly observed in the data 
(see \textcite{toddAssessingImpactSchool2006} and \textcite{attanasioEducationChoicesMexico2012} for examples in Progresa).
We require the welfare functions to be strongly concave in the policy in a neighborhood of the optimal policy,
but do not assume a specific functional form for the welfare function.
This allows the decision maker to target objectives that may not be simple reduced-form functions of observed outcomes,
such as agents' subjective utility or long-run profits.
The methods of this paper therefore enable researchers to take advantage
of the ``best of both worlds'' described by
\textcite{toddBestBothWorlds2023} to leverage experimental variation and
economic structure to efficiently learn about optimal policies.

\subsection{Related Literature}

\label{section:literature}

This paper intersects with two complementary literatures in econometrics: experimental 
design and policy choice.  
The questions of how to design experiments, and of how to choose policies 
based on the results of experiments, are often treated separately in econometrics.
The bandit literature, which does treat these questions jointly,
does not generally achieve the asymptotic optimality results that we achieve here
in our constrained, nonlinear, continuous-treatment setting.

The optimal design of experiments constitutes a vast literature
in many fields.
Classic approaches include
\textcite{silveyOptimalDesignIntroduction2013},
\textcite{pukelsheimOptimalDesignExperiments2006},
and
{%
  \DeclareDelimFormat[textcite]{finalnamedelim}
    {\addspace\bibstring{and}\linebreak\space}%
  \textcite{chalonerBayesianExperimentalDesign1995}%
}.
More recent works in econometrics have focused on experimental
design with the aim of efficiently estimating treatment effects
with binary or discrete treatments in semiparametric settings.
This literature is reviewed in 
\textcite{atheyEconometricsRandomizedExperiments2017}.
Examples include
\textcite{hahnAdaptiveExperimentalDesign2011},
\textcite{baiOptimalityMatchedPairDesigns2022},
\textcite{vivianoExperimentalDesignNetwork2020},
\textcite{tabord-meehanStratificationTreesAdaptive2023},
and \textcite{cytrynbaumFineStratificationSurvey2021}.
The methods proposed by these papers, like the one proposed here,
rely on large first-wave samples to inform the design of a second wave.
When this assumption fails, such methods can have poor finite-sample
properties (\cite{caiPerformanceNeymanAllocation2024}).
\textcite{armstrongAsymptoticEfficiencyBounds2022a} provides
asymptotic efficiency bounds for ATE estimation across experimental designs.
Our setting is parametric, but we view the contribution of
an asymptotically optimal experiment for policy choice in parametric models as both a novel and
necessary first step towards designing experiments for policy choice in more flexible settings.

Another large literature in econometrics focuses on policy choice,
given experimental or observational data
(\cite{manskiStatisticalTreatmentRules2004},
\cite{stoyeMinimaxRegretTreatment2009,stoyeMinimaxRegretTreatment2012a},
\cite{bhattacharyaInferringWelfareMaximizing2012},
\cite{kitagawaWhoShouldBe2018},
\cite{atheyPolicyLearningObservational2021},
\cite{mbakopModelSelectionTreatment2021}).
Counterfactual evaluation and policy choice is also a central
goal of structural econometric models;
we motivate our setting with the examples of
\textcite{toddAssessingImpactSchool2006} and
\textcite{attanasioEducationChoicesMexico2012},
which study the Progresa program that we use in our application.
This literature takes the data-generating process 
to be outside the control of the decision maker.
We complement this approach by showing how to choose the most favorable data-generating process
for the policy choice problem from a given class of experiments.
To do this, we extend the asymptotic analysis of treatment choice
rules as in \textcite{hiranoAsymptoticsStatisticalTreatment2009},
\textcite{hiranoAsymptoticAnalysisStatistical2020a},
and \textcite{xuAsymptoticAnalysisPoint2026}
to general nonlinear constrained decision problems
and optimize the value of this problem across experimental designs.

Adaptive experiments are commonly used to optimize a measure of welfare,
with the multi-armed bandit literature providing many algorithms
for both in-sample and out-of-sample welfare maximization
(\textcite{lattimoreBanditAlgorithms2020} and
\textcite{russoTutorialThompsonSampling2018} provide surveys,
while
\textcite{vivianoPolicyDesignExperiments2020},
\textcite{katoAdaptiveExperimentalDesign2025},
and \textcite{cesa-bianchiAdaptiveMaximizationSocial2025}
provide applications to policy choice in economics).
Much of the bandit literature proposes heuristic algorithms and proves
rate-optimality results for expected regret.
We take a different approach, related to
\textcite{hiranoAsymptoticRepresentationsSequential2025} and
\textcite{adusumilliRiskOptimalPolicies2025,adusumilliHowSampleWhen2025,adusumilliContinuousTimeAsymptotic2026a},
in which we characterize the limiting behavior of adaptive experiments and derive designs and policies which are
optimal in the limit experiment,
rather than rate-optimal in finite samples.
Our asymptotic representation theorem differs from prior work by allowing for continuous treatments.
Our decision problem is also distinct,
focusing on strongly concave welfare objectives and nonlinear constraints,
as opposed to a linear objective such as treatment choice with finitely many arms.

\subsection{Outline}

In Section \ref{section:decision}, we describe the general framework
and the decision problem faced by the decision maker,
and find that an exact solution is infeasible.
In Section \ref{section:method}, we describe a tractable solution method
in which the difficult decision problem is replaced by a simpler
Gaussian experiment with quadratic loss.
In Section \ref{section:limit}, we show the proposed method is asymptotically 
optimal, using a limit-of-experiments framework.
In Section \ref{section:progresa}, we apply the method to the Progresa conditional cash transfer experiment.

\section{Environment and decision problem}
\label{section:decision}
The decision maker must choose two things: how to design the second
wave of the experiment, and what policy to choose after observing the results of
the experiment.

\subsection{Welfare and expected regret}

The decision maker has a welfare function
\begin{align*}
    W(\pi, \theta)
\end{align*}
where $\pi\in\R^k$ is a policy which the decision maker can choose and
$\theta\in\R^\ell$ is a finite-dimensional unknown parameter.
The welfare function is not assumed to have a specific structure
beyond smoothness and local concavity assumptions specified in Section~\ref{section:limit}.
The mapping from policy and parameter to welfare is context-specific,
depending both on the model being used to predict the effects of counterfactual policies
and the preferences of the decision maker.
The set of feasible policies is described by nonlinear constraints
\begin{align*}
    \Pi = \{\pi \in \R^k : g(\pi) \leq 0 \}
    ,
\end{align*}
which may include budget constraints, fairness constraints,
or other exogenous constraints on the policy,
such as a nonnegativity constraint on a subsidy or price.

Since the true value of $\theta$ is not known,
the decision maker evaluates policies based on their expected regret.
Regret is the welfare loss from choosing a policy $\pi$
and is given by
\begin{align*}
  R(\pi, \theta) = \max_{\tilde \pi \in \Pi} W(\tilde \pi, \theta) - W(\pi,\theta)
\end{align*}
where the first term is the maximum welfare achieveable if $\theta$ were known.
Expected regret is calculated by averaging regret across possible
values of $\theta$ using a prior density $q(\theta)$,
resulting in the Bayes regret objective
\begin{align*}
    \E\left[R(\pi, \theta)\right] = 
    \int R(\pi, \theta) q(\theta) d\theta
    .
\end{align*}
Minimizing Bayes regret is equivalent to maximizing Bayes welfare
because the recentering term in regret does not affect the choice of policy.
However, recentering the welfare function in this way is useful for
asymptotic analysis.

While we  use Bayes regret as the objective,
the method proposed in this paper does not depend on the prior,
and hence a user of our method does not need to be able to articulate their prior.
This is because the decision maker uses the first wave to inform future decisions,
and as the size of the first wave grows the influence of the prior vanishes
under the conditions given in Section~\ref{section:limit}.

\subsection{The data-generating process}

The decision maker has access to an experimental sample that arrives in two waves.
Of the $n$ total units, $n_1$ arrive in the first wave $\mathcal{I}_{1,n} = \{1, \dots, n_1\}$,
and $n_2 = n - n_1$ arrive in the second wave $\mathcal{I}_{2,n} = \{n_1 + 1, \dots, n\}$.
Each unit is associated with a vector of covariates $x_i$.
Our analysis is conditional on covariates, and we treat them as fixed
throughout.
During the experiment, each unit is assigned a treatment $z_i$,
after which an outcome $y_i$ is observed.

In the first wave, the treatment $z_i$ is assigned according to a fixed experimental
design which is not chosen by the decision maker.
It is drawn from the density
\begin{align*}
    p_{z \mid x}(z_i \mid x_i; \delta_1),
    \quad
    i \in \mathcal{I}_{1,n}
\end{align*}
where $\delta_1$ is a fixed parameter describing the design of the first wave.
The dominating measure for the density $p_{z \mid x}$ may be Lebesgue measure or another measure,
so that treatment may be discrete or continuous, and potentially multidimensional.

After treatment is assigned in the first wave, the outcome $y_i$ is observed for these units.
Outcomes are generated according to the density
\begin{align*}
    p_{y \mid z, x}(y_i \mid z_i, x_i; \theta),
    \quad
    i \in \mathcal{I}_{1,n}
\end{align*}
where $\theta \in \Theta$ is an unknown finite-dimensional parameter
and $\Theta$ is an open subset of $\R^\ell$.
The outcome may not directly measure welfare,
but observations of $y_i$ are informative about $\theta$,
which determines the welfare effects of policies.

While we refer to the observations $i \in \mathcal{I}_{1,n}$ as the first wave
of the experiment, this data can take many forms in practice.
It may arise from an experiment conducted by the decision maker or by another entity.
It may be an observational dataset if
$\theta$ is identified from the observational data.
It can even be replaced by an asymptotically informative prior as in 
\textcite{adusumilliRiskOptimalPolicies2025}.
The first wave represents existing data available to the decision maker
that may be informative for policy decisions,
but does not reveal the optimal policy with perfect certainty.

In the second wave,
the decision maker may choose the design $\delta_2$
governing the conditional distribution of treatment
for the units in the second wave.
It is restricted to lie in a compact set
\begin{align*}
  \delta_2 \in \Delta
\end{align*}
which describes constraints on the experimental design,
such as budget constraints, ethical constraints, etc.
Once $\delta_2$ has been decided, the treatment for units in the second wave is drawn
according to
\begin{align*}
    p_{z \mid x}(z_i \mid x_i; \delta_2),
    \quad
    i \in \mathcal{I}_{2,n}
    .
\end{align*}
Thus, $\delta_2$ can specify how the probability of treatment depends on covariates or
how the value of a continuous treatment (such as a price) depends on covariates.

Once the decision maker assigns treatment, outcomes for the second wave are observed.
These outcomes are drawn according to the same conditional density as in the first wave,
\begin{align*}
    p_{y \mid z, x}(y_i \mid z_i, x_i; \theta),
    \quad
    i \in \mathcal{I}_{2,n},
\end{align*}
so both experimental waves can be combined to learn about $\theta$.

\subsection{Policy choice and the dynamic decision problem}

After observing the outcomes of both waves of the experiment, the decision maker chooses a policy $\pi$
to minimize posterior expected regret.
Letting the first-wave and second-wave random data be denoted by
\begin{align*}
  D_{1,n} = \{y_i, z_i\}_{i \in \mathcal{I}_{1,n}},
  \quad
  \quad
  D_{2,n} = \{y_i, z_i\}_{i \in \mathcal{I}_{2,n}},
\end{align*}
the value of an experimental dataset is given by
\begin{subequations}\label{eq:Vn}
\begin{equation}
  V_n(D_{1,n}, D_{2,n}) =
  \min_\pi
  \quad 
  \E\left[ R(\pi, \theta) \mid D_{1,n}, D_{2,n} \right]
  \quad
  \text{s.t.}
  \quad
  g(\pi) \leq 0.
\end{equation}
By backwards induction, the optimal design of the second wave is the solution to
\begin{equation}
  \min_{\delta_2}
    \quad
    \E_{\delta_2}\left[ 
        V_n(D_{1,n}, D_{2,n}) \mid D_{1,n}
    \right]
    \quad
    \text{s.t.} 
    \quad 
    \delta_2 \in \Delta
\end{equation}
where $\E_{\delta_2}$ denotes the expectation with respect to the law of motion for the second wave of the experiment
when the design $\delta_2$ is chosen,
\begin{equation}
    y_i, z_i \sim 
    p_{y \mid z, x}(y_i \mid z_i, x_i; \theta)
    p_{z \mid x}(z_i \mid x_i; \delta_2),
    \quad
    i \in \mathcal{I}_{2,n}
\end{equation}
in which treatment $z_i$ is first assigned according to $\delta_2$
and then the outcome is realized according to the unknown parameter $\theta$.
\end{subequations}

The dynamic program (\ref{eq:Vn}) characterizes the optimal design for the second wave,
but it is infeasible to solve in practice.
One could imagine trying to solve it
by standard dynamic programming methods,
but since $\pi$ is a function of all experimental data,
the state space of $V_n$ is the set of possible
datasets that the second wave of the experiment could generate.
This is extremely high-dimensional even for moderate-sized experiments
and low-dimensional variables\footnote{
    Alternatively, one could use the posterior on $\theta$ as the state,
    but if the likelihood and prior are not conjugate
    the state space will be the
    set of probability distributions on $\R^{\ell}$.
}.
This problem is a general feature of dynamic experiments,
and features prominently in multi-armed bandit literature.
More broadly, it is a feature of statistical decision problems
where a decision maker must specify a decision rule which may depend
on realizations of the data
(\cite{manskiEconometricsDecisionMaking2021}, \cite{hiranoWaldsStatisticalDecision2026}).

\section{Approximate solution method}
\label{section:method}
Since the finite-sample problem is intractable, we propose an approximation to 
this problem that is low-dimensional and tractable.
This approximation has two components: 
first, replacing the nonlinear regret with a quadratic function,
and second, replacing the data-generating process with a Gaussian observation.
This section presents the method as it is implemented in practice. 
Formal justification of the approximations and the asymptotic optimality of the 
method are given in Section \ref{section:limit}.

\subsection{Quadratic approximation of expected regret}

We replace the nonlinear expected regret loss function with 
a quadratic approximation.
Under this approximation, the optimal policy depends on $\theta$
only through a particular linear function,
which is the marginal welfare effect of changing the policy.
This reduces the state of the dynamic program to the posterior mean of this marginal effect.

Our quadratic approximation uses a second-order Taylor expansion 
of welfare around a reference point.
Suppose that the true parameter $\theta$ is in a neighborhood of 
some $\theta_0$ (to be formalized in Section \ref{section:limit}).
Let the optimal policy under $\theta_0$ be
\begin{align}
    \label{eq:nlp-reference}
    \pi_0 = \arg \min_{\pi} \quad R(\pi, \theta_0)
    \quad
    \text{s.t.} \quad
    g(\pi) \leq 0.
\end{align}
The centering term in regret makes $R(\pi, \theta)$ potentially nondifferentiable in $\theta$.
Our method instead uses derivatives of the welfare function $W(\pi, \theta)$, which is smooth in $\theta$,
and ignores the policy-irrelevant centering term in regret.
Define
\begin{align*}
  D = - \nabla_\pi W(\pi_0, \theta_0),
  \quad\quad
  B = - \nabla^2_{\pi\theta} W(\pi_0, \theta_0),
  \quad\quad
  H = - \nabla^2_{\pi\pi} W(\pi_0, \theta_0)
  .
\end{align*}
Since the choice of policy is subject to constraints,
it will not generally be true that $D=0$,
so this first-order term must be included in the approximation.
The quadratic function we minimize is then given by
\begin{align*}
  \tilde R(\pi, \beta) &=
  (\pi - \pi_0)' D
  + (\pi - \pi_0)' (\beta - \beta_0)
  + \frac{1}{2} (\pi - \pi_0)' H (\pi - \pi_0)
\end{align*}
where $\beta = B \theta$ and $\beta_0 = B \theta_0$.

The approximation $\tilde R$ depends on $\theta$ only through the linear functional $\beta$.
The marginal effect of changing the policy in a neighborhood of $(\pi_0, \theta_0)$ is
\begin{align*}
  \nabla_\pi R(\pi_0, \theta) &\approx D + (\beta - \beta_0)
\end{align*}
and so the parameter $\beta$ gives an approximation to the marginal effect of the policy
in this neighborhood.
When choosing a policy, $\theta$ is only relevant through the effects of counterfactual policies.
For policies close to $\pi_0$,
the marginal effect of changing the policy characterizes these counterfactuals.
Thus, $\beta$ serves as a sufficient statistic for policy choice in this neighborhood.

Further tractability comes from the linearity of $\tilde R$ in $\beta$.
Denoting the posterior mean of $\beta$ after both waves of the experiment by $\hat\beta_{2,n}$,
we have
\begin{align*}
  \E\left[ \tilde R(\pi, \beta) \mid D_{1,n}, D_{2,n} \right]
  &=
  \tilde R(\pi, \hat\beta_{2,n})
\end{align*}
and so when choosing a policy to minimize the expectation of $\tilde R(\pi, \beta)$,
the decision maker only needs to consider the posterior mean of $\beta$.
However, the law of motion for $\hat\beta_{2,n}$ is still complicated,
motivating our next approximation.

\subsection{Gaussian approximation of estimates}

Our second approximation is to replace the full data with a Gaussian observation
from each wave, with precision determined by $\delta_2$.
This allows the decision maker to reduce the optimization over treatment distributions
to optimization over posterior variance-covariance matrices for $\hat\beta_{2,n}$.

We approximate the prior distribution and experimental data with a flat prior
and Gaussian observations, respectively.
After the first wave, suppose the decision maker only observes a Gaussian estimate $\hat\theta_{1,n}$ of $\theta$,
which satisfies
\begin{align*}
  \hat \theta_{1,n} \overset{a}{\sim} N\left(\theta, \frac{1}{n_1} J_{1,n}^{-1}\right)
\end{align*}
where
\begin{align*}
  J_{1,n} = \frac{1}{n_1}\sum_{i \in \mathcal{I}_{1,n}} \E\left[
    \nabla_\theta \log p_{y \mid z, x}(y_i \mid z_i, x_i; \theta_0)
    \nabla_\theta \log p_{y \mid z, x}(y_i \mid z_i, x_i; \theta_0)'
  \right]
\end{align*}
is the Fisher information matrix at $\theta_0$ for the first wave of the experiment.
After choosing $\delta_2$ and running the second wave of the experiment,
suppose the decision maker observes a second Gaussian estimate $\hat \theta_{2,n}$
from the second wave, which satisfies
\begin{align*}
  \hat \theta_{2,n} \overset{a}{\sim} N\left(\theta, \frac{1}{n_2} J_{2,n}(\delta_2)^{-1} \right)
\end{align*}
where 
\begin{align*}
  J_{2,n}(\delta_2) = \frac{1}{n_2} \sum_{i \in \mathcal{I}_{2,n}} \E_{\delta_2} \left[
    \nabla_\theta \log p_{y \mid z, x}(y_i \mid z_i, x_i; \theta_0)
    \nabla_\theta \log p_{y \mid z, x}(y_i \mid z_i, x_i; \theta_0)'
  \right]
\end{align*}
is the Fisher information matrix when the design $\delta_2$ is chosen,
and $\E_{\delta_2}$ indicates that $z_i$ is assigned according to the design $\delta_2$.
By varying the experimental design $\delta_2$,
the decision maker varies the variance-covariance matrix of the resulting 
estimate $\hat\theta_{2,n}$.

The information matrix $J_{2,n}(\delta_2)$ affects the welfare of the
policy chosen at the end of the experiment by determining the information
available to the decision maker when choosing the policy.
This is summarized by the posterior distribution of $\beta$
conditional on the Gaussian estimates $\hat{\theta}_{1,n}$ and $\hat{\theta}_{2,n}$.
Under the flat-prior approximation, 
this posterior is also Gaussian and is given by
\begin{align*}
  \hat\beta_{2,n} =
  B
  \left(n_1 J_{1,n} + n_2 J_{2,n}(\delta_2)\right)^{-1}
  \left(n_1 J_{1,n} \hat \theta_{1,n} + n_2 J_{2,n}(\delta_2) \hat \theta_{2,n} \right)
  .
\end{align*}
Equivalently, since we are using a flat prior approximation,
$\hat\beta_{2,n}$ is the pooled maximum likelihood estimate of $\beta$ from the two waves of the experiment.
Therefore, our method does not depend on the choice of prior $q(\theta)$.
\footnote{This may be desirable if the decision maker has difficulty articulating their prior,
but it may be undesirable if the decision maker has a strong prior belief that is not reflected by a flat or weakly informative prior.
The method here can easily extend to informative Gaussian priors.}

Designs which are more informative about particular parameters
will lead to posteriors which are more informative about those parameters.
Whether a particular $\delta_2$ is optimal in terms of its information
depends on the objective and constraints of the decision problem.
In particular, our quadratic approximation implies that $\beta$
is a sufficient statistic for the policy choice.
However, due to constraints on the policy choice and the local curvature of welfare, 
the decision maker may not care about all components of $\beta$ equally.
If some margins described by $\beta$ are infeasible,
or if the first wave suggests that some margins would almost certainly decrease welfare,
then the decision maker may not want to spend their experimental budget to learn about those margins.

Finally, we note that our method still works if $\delta_2$ is chosen such that
$J_{2,n}(\delta_2)$ is singular,
as long as the resulting posterior information $n_1 J_{1,n} + n_2 J_{2,n}(\delta_2)$ is non-singular.
In this case, the MLE $\hat\theta_{2,n}$ may not be well-defined,
but the pooled MLE from both waves, and the posterior mean $\hat\beta_{2,n}$, will still be well-defined.
In fact, choosing a singular $J_{2,n}(\delta_2)$ may be optimal in some cases,
since the decision maker may prefer a design which is informative about $\beta$
but which provides no information about other components of $\theta$.

\subsection{Feasible analog and value function}

The approximation described above is constructed around a reference point 
$(\pi_0, \theta_0)$ which is unknown to the decision maker.
To make this method feasible, we replace the unknown reference point with estimates
from the first wave.
We use these estimates to construct empirical analogues of the Gaussian and quadratic approximations above.

After collecting the first wave of data, the decision maker constructs 
the maximum likelihood estimator $\hat\theta_{1,n}$,
which is approximately normally distributed.
The decision maker then constructs an estimate of the optimal policy from the first wave, given by
\begin{align*}
  \hat \pi_{1,n} = \arg\min_\pi \quad R(\pi, \hat\theta_{1,n}) \quad \text{s.t.} \quad g(\pi) \leq 0
  .
\end{align*}
The first-wave estimates $(\hat\pi_{1,n}, \hat\theta_{1,n})$ are then used as plug-in estimates for the unknown reference point $(\pi_0, \theta_0)$.
The quadratic approximation to regret is estimated by
\begin{align*}
  \hat R_n(\pi, \beta) &=
  (\pi - \hat \pi_{1,n})' \hat D_n
  + (\pi - \hat \pi_{1,n})' (\beta - \hat\beta_{1,n})
  + \frac{1}{2} (\pi - \hat\pi_{1,n})' \hat H_n (\pi - \hat\pi_{1,n})
\end{align*}
where $\hat D_n$, $\hat \beta_{1,n} = \hat B_n \hat\theta_{1,n}$, and $\hat H_n$
are all given by plug-in estimates.
Likewise, the Gaussian approximation is estimated by replacing the unknown 
Fisher information matrices $J_{1,n}$ and $J_{2,n}(\delta_2)$ with plug-in estimates.
These are given by
\begin{align*}
  \hat J_{1,n} &= \frac{1}{n_1} \sum_{i \in \mathcal{I}_{1,n}} \nabla_\theta \log p_{y \mid z, x}(y_i \mid z_i, x_i; \hat\theta_{1,n})
  \nabla_\theta \log p_{y \mid z, x}(y_i \mid z_i, x_i; \hat\theta_{1,n})'
  \\
  \hat J_{2,n}(\delta_2) &= \frac{1}{n_2} \sum_{i \in \mathcal{I}_{2,n}} \E_{\delta_2} \left[
    \nabla_\theta \log p_{y \mid z, x}(y_i \mid z_i, x_i; \hat\theta_{1,n})
    \nabla_\theta \log p_{y \mid z, x}(y_i \mid z_i, x_i; \hat\theta_{1,n})'
  \right]
\end{align*}
where the expectation in the second line is taken over the distribution of $z_i$ when the design is $\delta_2$
and the distribution of $y_i$ when the parameter is $\hat\theta_{1,n}$,
both of which are known to the decision maker after the first wave of the experiment.

Thus, the feasible approximation to the value function
which the decision maker can solve in practice is
\begin{subequations}\label{eq:hatVn}
\begin{equation}
  \hat V_n(\hat \beta_{2,n}) = \min_\pi 
  \quad \hat R_n(\pi, \hat \beta_{2,n}) 
  \quad \text{s.t.} \quad g(\pi) \leq 0
\end{equation}
after which the decision maker chooses the design of the second wave to minimize the expected value function:
\begin{equation}
  \min_{\delta_2} \quad
  \E\left[ \hat V_n(\hat\beta_{2,n}) \mid \hat \theta_{1,n} \right]
  \quad \text{s.t.} \quad \delta_2 \in \Delta
\end{equation}
where the expectation integrates over the second-wave distribution of $\hat \beta_{2,n}$ using
a plug-in estimate of the Gaussian Bayesian updating formula for the law of motion:
\begin{equation}
  \hat\beta_{2,n} \sim N\left(
    \hat\beta_{1,n},
    \hat B_n
    \left( \frac{1}{n_1} \hat J_{1,n}^{-1} - (n_1 \hat J_{1,n} + n_2 \hat J_{2,n}(\delta_2))^{-1} \right)
    \hat B_n'
  \right)
  \mid \hat \theta_{1,n}
  .
\end{equation}
\end{subequations}
We denote the solution to this problem by $\hat\delta_{2,n}$.

This value function is much lower-dimensional than the original value function $V_n$.
The state is simply $\hat\beta_{2,n}$, rather than the full experimental data $D_{1,n}$ and $D_{2,n}$.
In contrast to $V_n$, the state of $\hat V_n$ does not grow with the sample size.
Moreover, $\hat\beta_{2,n}$ is of the same dimension as the policy $\pi$, which is often low-dimensional for practical reasons.
Solving the value function $\hat V_n$ in practice is straightforward.
While algorithmic details are deferred to Appendix \ref{appendix:progresa},
the main steps are (i) evaluate $\hat V_n(\beta)$ for many values of $\beta$,
(ii) interpolate to construct a smooth surrogate for $\hat V_n$,
(iii) compute $\E_\delta[\hat V_n(\hat\beta_{2,n}) \mid \hat\theta_{1,n}]$ 
as a function of $\delta$ by Monte Carlo integration,
and (iv) optimize over $\delta_2$ using a nonlinear optimization routine.
Solving $\hat{V}_n$ is especially tractable when the constraints $g(\pi) \leq 0$ are linear,
as is the case in the Progresa application in Section \ref{section:progresa}.
Then $\hat{V}_n(\hat \beta)$ is a quadratic program for any value of $\hat \beta$,
making construction of $\hat V_n$ very fast.

\section{Asymptotic optimality}
\label{section:limit}
The solution method presented in the previous section is asymptotically optimal.
We establish this by showing that the
approximations made in the previous section have negligible error as
the sizes of the two waves of the experiment grow large.
These results imply that the design $\hat \delta_{2,n}$ that comes from
solving (\ref{eq:hatVn}) leads to the lowest possible regret asymptotically.

\subsection{Local asymptotic framework}
\label{subsection:localasymptotics}

The Gaussian and quadratic approximations are justified in an
asymptotic decision environment called the limit experiment.
The limit experiment describes the possible asymptotic behaviors of any
design and policy in the finite-sample environment.
Solving for the best design in the limit experiment
motivates finite-sample procedures that achieve asymptotic optimality.

The limit experiment is constructed in a local asymptotic framework.
Local asymptotics are commonly used to study
optimality of estimators and tests,
among other econometric and statistical problems
(\cite{le1972limits}, \cite{vandervaartAsymptoticStatistics2000}).
In local asymptotics, we analyze the performance of our method under
alternative parameter values that are difficult to distinguish from each other
even in large samples.
In particular, we model
\begin{align*}
    \theta = \theta_0 + h/\sqrt{n}
    &&
    \pi = \pi_0 + c/\sqrt{n}
\end{align*}
where $(h,c)$ is a local parameter and policy and
$(\theta_0, \pi_0)$ is a reference parameter and policy.
The $1/\sqrt{n}$ scaling means that the difference between $\theta$ and
$\theta_0$ is of the same order of magnitude
as the standard error of typical estimators of $\theta$.
Regret can be written in terms of local alternatives as
\begin{align*}
  R_n(c, h) = n R(\pi_0 + c/\sqrt{n}, \theta_0 + h/\sqrt{n})
\end{align*}
where the $n$ scaling ensures a nontrivial limit as $n \to \infty$.

Local asymptotics are essential for our framework to capture
finite-sample uncertainty about which constraints are binding.
The set of $\theta$ for which a particular constraint holds with weak complementary slackness
generally has Lebesgue measure zero.
This means that for almost every $\theta$,
the decision problem behaves asymptotically as if the set of binding constraints is known.
In this case, the decision problem in the limit experiment reduces to 
minimizing a quadratic function over a linear subspace,
and hence is equivalent to point estimation.
In contrast, our approach of localizing around a reference parameter $\theta_0$ 
where weak complementary slackness holds preserves nontrivial constraints in the limit
and ensures that even asymptotically, the decision maker does not know
which constraints bind at the optimum.
This uncertainty is essential for capturing the distinctive aspects
of experimental design for policy choice.

This localization has an intuitive interpretation in the cash-transfer example of Section \ref{section:progresa}.
The decision maker is most likely to make a mistake when it is difficult to detect
differences in the marginal effects of the cash transfer across groups.
When this is the case, noisy experimental data may lead the decision maker to allocate extra money to boys even
when girls would benefit more.
In this scenario, $\theta_0$ may be a point where the marginal effects of the cash transfer are positive and equal across groups,
and $\pi_0$ is a vertex of the feasible set.
Then at $\pi_0$, the cash is allocated to only one of the four groups,
but local perturbations $h \neq 0$ may cause the optimal policy $c$ to move away from the vertex and allocate some cash to the other groups.
Our asymptotic framework does not require that $\theta_0$ be specified, but it is intended
to capture this statistically ``hard'' instance of the decision problem.
This is analogous to localizing around an average treatment effect of zero in a binary
treatment choice problem (\cite{hiranoAsymptoticsStatisticalTreatment2009}).

\subsection{Limit experiment}

The bound we establish on the limiting regret of any sequence of designs
is given by the value of an optimization problem in the limit experiment.
This limiting decision problem is a quadratic program with linear constraints.

The constraints in the limit experiment are given by linear approximations 
to the constraints in the finite-sample problem.
Define $\lambda_0$ to be the unique Lagrange multiplier (guaranteed to exist under the assumptions given below)
satisfying
\begin{align*}
  D + \lambda_0' \nabla g(\pi_0) = 0
  &&
  \lambda_{0j} \geq 0
  &&
  \lambda_{0j} g_j(\pi_0) = 0
\end{align*}
for $j=1, \dots, J$.
Let the set of indices $j$ for which
strict/weak complementary slackness holds at $\pi_0$ be
\begin{align*}
    \mathcal{J}_s = \{
        j : g_j(\pi_0) = 0, \; \lambda_{0j} > 0
    \}
    &&
    \mathcal{J}_w = \{
        j : g_j(\pi_0) = \lambda_{0j} = 0
    \},
\end{align*}
respectively.
The feasible set in the limit experiment is the convex cone
\begin{align*}
    \mathcal{G} = \{
        c : \;
        \nabla g_j(\pi_0)' c = 0, j \in \mathcal{J}_s, \;
        \nabla g_j(\pi_0)' c \leq 0, j \in \mathcal{J}_w, \;
    \}
    .
\end{align*}

The quadratic objective arises from the quadratic approximation to regret,
but also incorporates the curvature of the constraints through the Hessian of the Lagrangian.
Define the Hessian of the Lagrangian at $(\pi_0, \lambda_0)$ as
\begin{align*}
  \bar H = H + \lambda_0' \nabla^2 g(\pi_0)
\end{align*}
and define the quadratic objective
\begin{align*}
  r(c, b) = c' b + \frac{1}{2} c' \bar H c
\end{align*}
where $b = B h$ is the local alternative for $\beta$.
This quadratic objective, like $\tilde R$ in Section \ref{section:method}, omits the policy-irrelevant centering term in regret,
which simplifies the dynamic program in the limit experiment.
The asymptotic lower bound, however, is given by the regret in the limit experiment,
which is
\begin{align*}
  R_\infty(c, b) = r(c, b) - \min_{c \in \mathcal{G}} r(c, b).
\end{align*}

We now describe the information structure in the limit experiment.
We begin with the prior.
Fix $K > 0$ and define the localized finite-sample prior as
\begin{align*}
  q_{n,K}(h) \propto q(\theta_0 + h/\sqrt{n}) \1\{ \|h\| \leq K \}
\end{align*}
normalized to integrate to one.
The corresponding density in the limit experiment is
\begin{align*}
  q_{\infty,K}(h) \propto \1\{ \|h\| \leq K \}
\end{align*}
which is uniform on the ball of radius $K$.
This localization of the prior serves two purposes.
First, it ensures that the limit of expected regret is finite,
and hence the asymptotic lower bound is well-defined.
Second, it asymptotically concentrates around $\theta_0$,
a point of possible weak complementary slackness, and hence
reflects the finite-sample property that some constraints may be on the margin of binding or nonbinding.
In the results that follow, we will show our method, which is motivated by a flat-prior approximation,
is asymptotically optimal as $K$ grows large, and hence as our local prior becomes weakly informative.

The observed data in the limit experiment comes from a two-wave adaptive experiment with Gaussian signals.
Let $\rho$ denote the asymptotic fraction of the sample in the second wave of the experiment.
In the first wave of the limit experiment, the decision maker observes a random variable
\begin{align*}
  A_1 \sim N((1-\rho) J_1 h, (1-\rho) J_1)
  .
\end{align*}
Note that when $J_1$ is nonsingular,
observing $A_1$ is equivalent to observing a Gaussian signal
$\hat h_1 = (1-\rho)^{-1} J_1^{-1} A_1 \sim N(h, (1-\rho)^{-1}J_1^{-1})$,
which is the limiting distribution of the first-wave maximum likelihood estimator.
After observing $A_1$, the decision maker chooses a design $\delta_2$ and observes
\begin{align*}
  A_2 \sim N(\rho J_2(\delta_2) h, \rho J_2(\delta_2))
  .
\end{align*}
As with the first wave, when $J_2(\delta_2)$ is nonsingular, observing $A_2$ is equivalent to observing a Gaussian signal
$\hat h_2 = \rho^{-1} J_2(\delta_2)^{-1} A_2 \sim N(h, \rho^{-1} J_2(\delta_2)^{-1})$.

After observing $A_1$ and $A_2$, the decision maker chooses a policy $c$
to minimize the posterior expected regret.
Let $\hat b_{2,K}$ be the posterior mean of $b$ conditional on $A_1$ and $A_2$ under the prior $q_{\infty,K}$.
The value function in the limit experiment (omitting the centering term) is
\begin{align*}
  V(\hat b_{2,K}) =
  \min_c \quad r(c, \hat b_{2,K})
  \quad \text{s.t.} \quad
  c \in \mathcal{G}
\end{align*}
and the optimal design in the limit experiment is given by
\begin{align*}
  \min_{\delta_2} \quad \E_{\delta_2,K}\big[
    V(\hat b_{2,K})
    \mid A_1
  \big]
  \quad \text{s.t.} \quad
  \delta_2 \in \Delta
  .
\end{align*}
We denote the optimal design and policy in the limit experiment by $(\delta_{2,K}^*,c^*_K)$.

We evaluate the asymptotic performance of a sequence of designs and policies
by their Bayes regret under the truncated prior $q_{n,K}$.
We consider experimental designs which are described by
a fixed design rule $\phi_n$ so that $\delta_{2,n} = \phi_n(D_{1,n}, U)$ where $U$ is auxiliary randomness independent of the data.
The ex-ante Bayes regret under $q_{n,K}$ is given by
\begin{align*}
  \mathcal{R}_{n,K}(\delta_{2,n}, c_n)
  = \int \E_h^{\phi_n}\left[ R_n(c_n, h) \right] q_{n,K}(h) dh
\end{align*}
where $\E_h^{\phi_n}$ denotes the expectation over the data when the local parameter is $h$ and the design rule is $\phi_n$.
When $c_n = c_{n,K}^*$ is the Bayes optimal policy for $\delta_{2,n}$, we abbreviate
to $\mathcal{R}_{n,K}^*(\delta_{2,n})$, and when $\delta_{2,n} = \delta_{2,n,K}^*$ is the Bayes optimal design, 
we abbreviate to $\mathcal{R}_{n,K}^{**}$.
The fixed-$K$ regret of a design $\delta_2$ and policy $c$ in the limit experiment is
\begin{align*}
  \mathcal{R}_{\infty,K}(\delta_2, c)
  = \int \E_h^\phi\left[ R_\infty(c, B h) \right] q_{\infty,K}(h) dh
\end{align*}
where $\phi$ is the limit experiment design rule governing $\delta_2 = \phi(A_1, U)$.
The quantities $\mathcal{R}_{\infty,K}^*(\delta_2)$ and $\mathcal{R}_{\infty,K}^{**}$ are defined analogously.

\subsection{Asymptotic representation theorem}

Our first result is an asymptotic representation theorem which justifies the Gaussian approximation
to the information environment.
It states that the statistical behavior of any sequence of designs and policies in the finite-sample experiment
is asymptotically approximated by a design and policy in the Gaussian limit experiment.

To apply asymptotic approximations to both waves of the experiment,
we require that both waves of the experiment grow large.
\begin{assumption}
  \label{assumption:largewaves}
  Wave sizes satisfy $\lim_{n\to\infty} n_2/n = \rho$ and $\lim_{n\to\infty} n_1/n = 1 - \rho$ for some $\rho \in (0,1)$.
\end{assumption}
This assumption implies that the first wave cannot be an uninformative pilot that is vanishingly small
compared to the second wave--- it must be large enough to yield a $\sqrt{n}$-consistent estimate of $\theta$.

A key (and standard) condition for local asymptotic analysis is that the model is smooth in $\theta$ in the sense of differentiability in quadratic mean (DQM),
allowing for local approximations to the likelihood.

\begin{assumption}
  \label{assumption:dqm}
  There exists a function $\psi(y \mid z, x)$,
  called the score of $p_{y \mid z, x}$ at $\theta_0$,
  such that for all $h$,
  \begin{align*}
    \sup_{z, x} \int &\left[
      \sqrt{p_{y \mid z, x}(y \mid z, x; \theta_0 + h)}
      - \sqrt{p_{y \mid z, x}(y \mid z, x; \theta_0)}
    \right. \\ & \left.
      - \frac{1}{2} h' \psi(y \mid z, x) \sqrt{p_{y \mid z, x}(y \mid z, x; \theta_0)}
    \right]^2 \mu(dy)
    = o(\|h\|^2)
  \end{align*}
  as $h \to 0$, for some measure $\mu$ which dominates $p_{y \mid z, x}$ for all $(z, x)$.
\end{assumption}
At a high level, this means the square root of the density is differentiable in $\theta$ at $\theta_0$.
This condition allows us to use the score to linearly approximate the likelihood at $\theta$ near $\theta_0$.
Typically, the score is the gradient of the log-likelihood of the model
(see Lemma 7.6 of \textcite{vandervaartAsymptoticStatistics2000} for sufficient conditions for this formulation).

To show that the Gaussian approximation holds uniformly across the set of designs,
we need to link the conditional distributions
$p_{z \mid x}(\cdot \mid x; \delta)$ and $p_{y \mid z, x}(\cdot \mid z, x; \theta)$
across different designs through a potential outcomes model.
\begin{assumption}
  \label{assumption:potentialoutcomes}
  There exist functions $z(x, \nu; \delta)$ and $y(z, x, \epsilon; \theta)$ such that
  \begin{align*}
    z_i = z(x_i, \nu_i; \delta), \quad\quad y_i = y(z_i, x_i, \epsilon_i; \theta)
  \end{align*}
  where $\nu_i$ and $\epsilon_i$ are independent of each other and independent across $i$.
\end{assumption}
The structural error $\epsilon_i$ captures the unobserved determinants of the outcome $y_i$,
while the randomization device $\nu_i$ is generated by the decision maker to determine the treatment assignment $z_i$.
With this potential outcomes model,
we can define the score of an observation $i$ as a function of the design $\delta$ as
\begin{align*}
  \psi_i(\delta)
  =
  \psi\bigg(
    y\Big(
      z(x_i, \nu_i; \delta), 
      x_i, \epsilon_i; \theta
    \Big), 
    z(x_i, \nu_i; \delta),
    x_i
  \bigg).
\end{align*}
We also define the score process and log-likelihood ratio in the second wave as
\begin{align*}
  A_{2,n}(\cdot) = \frac{1}{\sqrt{n}} \sum_{i \in \mathcal{I}_{2,n}} \psi_i(\cdot)
  , \quad\quad
  \Lambda_{2,n}^h(\cdot) = \sum_{i \in \mathcal{I}_{2,n}} \log \frac{p_{y \mid z, x}(y_i(\cdot) \mid z_i(\cdot), x_i; \theta_0 + h/\sqrt{n})}{p_{y \mid z, x}(y_i(\cdot) \mid z_i(\cdot), x_i; \theta_0)}
  .
\end{align*}
The sample score $A_{1,n}$ and log-likelihood ratio $\Lambda_{1,n}^h$ in the first wave are defined analogously, but with the fixed design $\delta_1$.

The specific smoothness required for uniform approximations of the decision environment
is stochastic equicontinuity of the log-likelihood ratio and the score
as processes in $\Delta$.
\begin{assumption}
  \label{assumption:stocheq}
  $\Delta$ is a compact subset of Euclidean space.
  $A_{2,n}(\cdot)$ is stochastically equicontinuous as a process in $\Delta$ under $\theta_0$.
  For any $h$, $\Lambda_{2,n}^h(\cdot)$ is stochastically equicontinuous as a process in $\Delta$ under $\theta_0$.
\end{assumption}
Stochastic equicontinuity is a high-level condition that ensures that the
score and log-likelihood have relatively smooth sample paths for large enough $n$.
Many lower-level sufficient conditions are available in the literature
(e.g. \textcite{andrewsChapter37Empirical1994}, \textcite{vandervaartWeakConvergenceEmpirical2013}).
While Assumption \ref{assumption:stocheq} is imposed only under $\theta_0$,
we show in Appendix \ref{appendix:art} that the laws of the adaptive experiments
under local alternatives are contiguous, so that approximations that
vanish in probability under $\theta_0$ also vanish in probability under local alternatives.

Assumption \ref{assumption:stocheq} is a joint restriction on both the treatment assignment mechanism
and the outcome model. Both the treatment assignment and the model must be smooth enough that small changes in $\delta$
lead to small changes in the data and associated likelihood.
In Appendix \ref{appendix:stocheq} we give lower-level sufficient conditions which are natural for our setting,
and verify that they are satisfied in the Progresa example in Section \ref{section:progresa}.

Finally, we need the sample information matrices to have well-defined limits.
\begin{assumption}
  \label{assumption:information}
  There exists a matrix $J_1$
  and a covariance kernel $J_2(\cdot, \cdot)$ such that for any $d_1, d_2 \in \Delta$,
  \begin{align*}
    \frac{1}{n_1} \sum_{i \in \mathcal{I}_{1,n}} \E_{\theta_0} \big[ \psi_i(\delta_1) \psi_i(\delta_1)' \big] &\to J_1
    \\
    \frac{1}{n_2} \sum_{i \in \mathcal{I}_{2,n}} \E_{\theta_0} \big[ \psi_i(d_1) \psi_i(d_2)' \big] &\to J_2(d_1, d_2).
  \end{align*}
  $J_1$ is positive definite. 
  When $d_1 = d_2 = d$, we write $J_2(d) := J_2(d, d)$.
  The pooled information matrix is defined as $J(\delta_2) = (1-\rho) J_1 + \rho J_2(\delta_2)$.
\end{assumption}

Our asymptotic representation theorem states that any convergent
sequence of designs and policies in the finite-sample experiment
converges to a design and policy in the limit experiment.
\begin{theorem}[Asymptotic representation theorem]
  \label{theorem:art}
  Suppose \Crefrange{assumption:largewaves}{assumption:information} hold.
  Let $(\delta_{2n}, c_n)$ be a sequence of adaptive designs and policies in the finite-sample experiment which converges in distribution
  under the sequence $\theta_0 + h/\sqrt{n}$ for every $h$.
  Then there exist statistics $(\delta_2, c)$ and random variables
  $A_1, A_2, U$
  such that
  \begin{enumerate}
    \item $A_1 \sim N( (1-\rho) J_1 h, (1-\rho) J_1)$ independently of $U$,
        and $\delta_2 = \delta_2(A_1, U)$
    \item $A_2 | A_1, U \sim N(\rho J_2(\delta_2) h, \rho J_2(\delta_2))$ 
        and $c = c(A_1, A_2, U)$
    \item $(\delta_{2n}, c_n) \overset{h}{\leadsto} (\delta_2, c)$.
  \end{enumerate}
\end{theorem}

Asymptotic representation theorems are often used to provide simple approximations to statistical decision problems.
For the classic, non-adaptive case, see \textcite{le1972limits} and \textcite{vandervaartAsymptoticStatistics2000}.
\textcite{hiranoAsymptoticsStatisticalTreatment2009} use such a theorem to derive asymptotically optimal
rules for binary treatment choice.

\textcite{hiranoAsymptoticRepresentationsSequential2025} establish a similar result for batched multi-armed bandits.
Our result, while stated for a two-wave adaptive experiment, can be extended to multiple waves.
The main distinction between our result and \textcite{hiranoAsymptoticRepresentationsSequential2025} is that we allow for continuous treatments.
Our result builds on this work by extending their argument to these more complex design spaces
in which $\delta$ can determine not just the fraction of units assigned to a treatment,
but also the values of continuous treatments assigned to each unit, which can depend on covariates.
This generality comes at the cost of the additional Assumption \ref{assumption:stocheq},
which is satisfied in the multi-armed bandit setting of \textcite{hiranoAsymptoticRepresentationsSequential2025}
because the score and log-likelihood are partial-sum processes.
Finally, \textcite{hiranoAsymptoticRepresentationsSequential2025} do not solve for optimal policies,
which we consider next.

Theorem \ref{theorem:art} applies to sequences of policies $c_n$ that converge at the parametric rate.
While this can fail in nonparametric empirical welfare maximization problems
(\cite{kitagawaWhoShouldBe2018}, \cite{atheyPolicyLearningObservational2021}),
it is much less restrictive in the parametric models considered here.
In particular, Appendix \ref{appendix:lowerbound} verifies that the Bayes-optimal policy
is uniformly tight.

We note that the usefulness of Theorem \ref{theorem:art} extends beyond the 
specific decision problem considered here.
It does not depend on the objective function or constraints of the decision problem,
nor does it impose a Bayesian framework--- it is stated for fixed local alternatives.
Thus, Theorem \ref{theorem:art} can be used to analyze efficiency of estimators
and tests in adaptive experiments with continuous treatments.
We leave this to future work.
 
\subsection{Asymptotic optimality}

Our main result states that the procedure described in Section \ref{section:method}
is asymptotically optimal.
We first establish a lower bound on the regret of any sequence of designs and policies in the finite-sample experiment.
Then we show that the procedure described in Section \ref{section:method} achieves this lower bound, and is therefore asymptotically optimal.

We first require some regularity on the decision maker's objective and constraints.
\begin{assumption}[Nonlinear program regularity]
  \label{assumption:nlp}
  \begin{enumerate}
    \item $W$ is continuous on $\Pi \times \Theta$ and twice continuously differentiable 
      in a neighborhood of $(\pi_0, \theta_0)$.
      Each $g_j$ is twice continuously differentiable in a neighborhood of $\pi_0$.
    \item $\pi_0$ is the unique global maximizer of $W(\pi, \theta_0)$ subject to $g(\pi) \leq 0$.
      There exists some $\alpha < W(\pi_0, \theta_0)$, a neighborhood $N$ of $\theta_0$, and a compact set $C_\pi$ such that
      $\{ \pi \in \Pi : W(\pi, \theta) \geq \alpha\} \neq \emptyset$ and
      $\{\pi \in \Pi : W(\pi, \theta) \geq \alpha\} \subseteq C_\pi$ for all $\theta \in N$.
      $\Pi$ is closed and convex.
    \item The gradients $\{ \nabla g_j(\pi_0) : g_j(\pi_0) = 0 \}$ are linearly independent.
    \item There is a $\kappa > 0$ such that
      $c' \bar H c \geq \kappa \| c \|^2$
      for all $c \in \text{span}(\mathcal{G})$
      and $c' H c \geq \kappa \| c \|^2$
      for all $c \in \text{span}(\Pi - \Pi)$.
  \end{enumerate}
\end{assumption}

These conditions are standard in the literature on sensitivity analysis for nonlinear programs
under weak complementary slackness (\cite{shapiroSecondOrderSensitivity1985}, \cite{bonnansPerturbationAnalysisOptimization2013}).
They imply that the value of the reference decision problem (\ref{eq:nlp-reference}) is locally approximated by a quadratic program.
Item 3 ensures that the Lagrange multiplier $\lambda_0$ is unique, but not necessarily nonzero.
This means that some constraints in the limit experiment will be inequality constraints,
and the decision maker may not know whether these constraints are binding or nonbinding at the true optimal policy.
Item 4 restricts our attention to strongly concave welfare functions;
this rules out binary treatment choice with a linear welfare function.

We also require smoothness of the prior.
\begin{assumption}
  \label{assumption:prior}
  $q$ is positive and continuous in a neighborhood of $\theta_0$.
\end{assumption}

Under these assumptions, the optimal design and policy $(\delta_{2,K}^*, c_K^*)$ in the limit experiment provides
a lower bound on the regret of any sequence of designs and policies in the finite-sample experiment.
\begin{theorem}[Lower bound]
  \label{theorem:lowerbound}
  Suppose \Crefrange{assumption:largewaves}{assumption:prior} hold.
  Let $(\delta_{2,n}, c_n)$ be any sequence of designs and policies in the finite-sample experiment.
  Then for any fixed $K < \infty$,
  \begin{align*}
    \liminf_{n\to\infty} \mathcal{R}_{n,K}(\delta_{2,n}, c_n) \geq \mathcal{R}_{\infty,K}^{**}
    .
  \end{align*}
\end{theorem}

Theorem \ref{theorem:lowerbound} combines Theorem \ref{theorem:art},
which justifies the Gaussian information environment,
with a second-order expansion of the oracle problem (\ref{eq:nlp-reference}).
The lower bound is more than just a Taylor approximation to welfare;
it is a joint approximation of the objective and constraints of the decision problem.
This approximation is related to existing results on sensitivity analysis of nonlinear programs under weak complementary slackness
(\cite{shapiroSecondOrderSensitivity1985}, \cite{bonnansPerturbationAnalysisOptimization2013}).
In particular, $V(b)$ provides a second-order directional derivative of the oracle problem (\ref{eq:nlp-reference}) 
in the direction $b$, up to policy-invariant centering terms.
Theorem \ref{theorem:lowerbound} then shows that the pointwise directional derivative 
provides an asymptotic lower bound
in the dynamic decision problem in the limit experiment.

We now show that our proposed experimental design $\hat\delta_{2,n}$ 
is asymptotically optimal in the sense that it attains the lower bound in Theorem \ref{theorem:lowerbound}.
To establish asymptotic optimality,
we construct a feasible sequence of terminal policies which attain the lower bound in Theorem \ref{theorem:lowerbound}.
Let $\hat A_{1,n}$ and $\hat A_{2,n}(\delta_2)$ be the feasible analogues of $A_{1,n}$ and $A_{2,n}(\delta_2)$,
where the score is evaluated at $\hat \theta_{1,n}$ rather than $\theta_0$.
Let $\hat J_n(\delta_2) = \frac{n_1}{n} \hat J_{1,n} + \frac{n_2}{n} \hat J_{2,n}(\delta_2)$ be the feasible analogue of $J(\delta_2)$.
Define the one-step estimator\footnote{
  The maximum likelihood estimator could also be used after strengthening the classical assumptions
  (e.g. Section 3.2 of \textcite{vandervaartWeakConvergenceEmpirical2013})
  to ensure uniform $\sqrt{n}$-consistency and asymptotic normality of the MLE across the set of designs.
  Given stochastic equicontinuity of the score process, the one-step estimator yields a simpler proof.
}
$\hat \theta_{2,n}^\text{os} = \hat\theta_{1,n} + \frac{1}{\sqrt n} \hat J_n(\hat\delta_{2,n})^{-1} \hat A_{2,n}(\hat\delta_{2,n})$
and $\hat \beta_{2,n}^\text{os} = \hat \beta_{1,n} + \hat B_n (\hat \theta_{2,n}^\text{os} - \hat \theta_{1,n})$.
The localized version is $\hat h_n^\text{os} = \sqrt{n}(\hat\theta_{2,n}^\text{os} - \theta_0)$.
Define
\begin{align*}
  \hat \pi_n^\text{os} = \arg\min_{\pi \in \Pi} \hat R_n(\pi, \hat\beta_{2,n}^\text{os}) \quad \text{s.t.} \quad g(\pi) \leq 0
\end{align*}
and let $\hat c_n^\text{os} = \sqrt{n}(\hat\pi_n^\text{os} - \pi_0)$ be the corresponding localized policy.
The one-step estimator must be regular in the following sense:

\begin{assumption}[Plug-in estimator]
  \label{assumption:plugin}
  \begin{enumerate}
    \item The first-wave MLE is $\sqrt{n}$ consistent \linebreak and interior to the parameter space.
    \item The score $\psi_\theta$ exists in a neighborhood of $\theta_0$ and satisfies
    \begin{align*}
      \sup_{\|h\|\leq K}
      \left\| \frac{1}{\sqrt{n}} \sum_{i \in \mathcal{I}_{1,n}} \left[\psi_{i,h}(\delta_1) - \psi_i(\delta_1)\right]
      + (1-\rho) J_1 h \right\|
      &= o_{P_0}(1)
      \\
      \sup_{d\in\Delta, \|h\| \leq K}
      \left\| \frac{1}{\sqrt{n}} \sum_{i \in \mathcal{I}_{2,n}} \left[\psi_{i,h}(d) - \psi_i(d) \right]
      + \rho J_2(d) h \right\|
      &= o_{P_0}(1)
      .
    \end{align*}
    \item $\|\hat J_{1,n} - J_1\| = o_{P_0}(1)$ and $\sup_{d\in\Delta} \|\hat J_{2,n}(d) - J_2(d)\| = o_{P_0}(1)$.
    \item The sequence $R_n(\hat c_n^\text{os}, h)$ is uniformly integrable under $q_{\infty,K}$ for any $K < \infty$.
  \end{enumerate}
\end{assumption}
This assumption ensures that the feasible one-step estimator $\hat \beta^\text{os}_{2,n}$ behaves
like the Gaussian limit experiment estimator $\hat b_{2,K}$ for large $K$.
Interiority of $\hat \theta_{1,n}$ implies that the first-wave score  $\hat A_{1,n} = 0$.
Together, items 1-3 imply that the one-step process $\hat A_{2,n}(\cdot)$ is asymptotically equivalent to the Gaussian limit experiment score process $A_2(\cdot)$,
just as Assumption \ref{assumption:stocheq} implies that the true second-wave score process $A_{2,n}(\cdot)$ is asymptotically equivalent to $A_2(\cdot)$.
The uniform score and information conditions are implied by uniform versions of standard
differentiability and dominance conditions on the likelihood.

We now show that a decision maker who uses $\hat\delta_{2,n}$ to design the second wave of the experiment
and subsequently chooses the one-step estimator $\hat \pi_n^\text{os}$ for the final policy 
attains the asymptotic lower bound in Theorem \ref{theorem:lowerbound}.

\begin{theorem}[Asymptotic optimality]
  \label{theorem:optimality}
  Suppose that for any $K$, \Crefrange{assumption:largewaves}{assumption:plugin} hold.
  Then
  \begin{align*}
    \lim_{K\to\infty} \limsup_{n\to\infty} \left\{ \mathcal{R}_{n,K}(\hat\delta_{2,n}, \hat c_n^\text{os}) - \mathcal{R}_{\infty,K}^{**}\right\} = 0
    .
  \end{align*}
\end{theorem}

This is the main result of the paper, which justifies the use of the method proposed in Section \ref{section:method}.
While Theorem \ref{theorem:lowerbound} states that no feasible design and policy in the finite-sample experiment can
achieve lower regret than the optimal design and policy in the limit experiment,
this bound is useful only insofar as the design and policy used in practice can attain it.
Theorem \ref{theorem:optimality} shows that the proposed method does attain this bound\footnote{
Since $\mathcal{R}_{n,K}^*(\hat \delta_{2,n}) \leq \mathcal{R}_{n,K}(\hat\delta_{2,n}, \hat c_n^\text{os})$,
Theorem \ref{theorem:optimality} implies that a decision maker who uses $\hat\delta_{2,n}$ to run 
the experiment and then chooses the fully ex-post Bayes optimal policy $c_{n,K}^*$
will also attain the asymptotic lower bound;
however, the one-step estimator is simpler to compute and is asymptotically equivalent.
}.
Since the limit experiment is defined under a truncated prior,
the flat-prior procedure in Section \ref{section:method} is not asymptotically optimal for any fixed $K$,
since it disregards the information of the truncated prior.
However, as $K$ grows large and the truncation becomes less informative,
the optimality gap between the flat-prior and truncated-prior procedure vanishes.

Previous asymptotic optimality results in adaptive experiments apply to decisions such as point estimation,
hypothesis testing, or discrete treatment choice 
(\cite{hiranoAsymptoticRepresentationsSequential2025,adusumilliRiskOptimalPolicies2025,adusumilliHowSampleWhen2025}).
To our knowledge, Theorem \ref{theorem:optimality} is the first asymptotic optimality result for adaptive experiments with continuous treatments,
nonlinear objectives, and nonlinear constraints.
A limitation of this result is that it is restricted to Bayes regret.
Theorem \ref{theorem:optimality} may be useful in studying asymptotic optimality under minimax regret, but we leave this to future work.

\section{Application to Progresa}
\label{section:progresa}
We apply our method to the Progresa cash transfer experiment in Mexico.
We consider the problem of choosing the size and targeting
of the cash transfer to maximize the secondary school completion rate
among children in the program.
When our proposed method is used to design a second wave of the Progresa
experiment, we find that the expected regret of the policy choice problem 
is reduced by up to 70\% relative to a second wave which uses the original
Progresa experimental design.

\subsection{Model}

We use a dynamic school choice model to learn about the effects
of counterfactual policies that the decision maker may choose.
This approach has been taken by
\textcite{toddAssessingImpactSchool2006} and \textcite{attanasioEducationChoicesMexico2012}.
One reason this model is useful for policy choice is that
the school completion rate is a long-term outcome which is not
immediately observed in the experiment.
Instead, only current schooling enrollment decisions are observed.
The dynamic model allows the decision maker to use the observed effects
of the program to learn about the long-term effects of different policies

In describing the model, we suppress the unit index $i$.
At each age $\tau$ before age $18$, households may choose whether a child in grade $s_\tau$
will attend school, denoted $y_\tau = 1$ or $y_\tau = 0$.
If the child attends school, they receive a subsidy $z_{s_\tau}$
which depends on the grade of enrollment.
The utility a household receives from a child enrolling in school
depends on the subsidy $z$ as well as 
the child's age $\tau$,
current grade level $s_\tau$,
a female sex indicator $f$,
and an indicator for attending secondary school.
The difference at age $\tau$ in the utility from enrolling in school
versus not enrolling is given by
\begin{align*}
    u_{\tau}(s, f, z) =&
    \gamma_0 + \gamma_1 \tau
    +\gamma_2 s
    + \gamma_3 f
    + \gamma_4 \1[s \geq 7] +
    \\
    &
    \gamma_5 (1-f) z
    + \gamma_6 (1-f) \tau z
    + \gamma_7 f z
    + \gamma_8 f \tau z
    + \epsilon_{\tau}
    .
\end{align*}

Households choose a sequence of school attendance $y_\tau$ to maximize long-term
utility.
Each period, if the child was enrolled in school they may advance a grade
with probability $p_\text{pass}(\tau, s_\tau)$ depending on age and grade,
or they may fail and remain at the same grade.
Failure is assumed to be exogenous.
Finally, at age $\tau=18$ households obtain a terminal value
\begin{align*}
    v_{18} = \gamma_9 \1[s_{18} \geq 10]
\end{align*}
which depends on whether the child completed secondary school.
Thus, households solve
\begin{align*}
    \max_{\{y_\tau\}_{\tau=6}^{17}} &\sum_{\tau=6}^{17}
    \eta^\tau \E\left[ y_\tau u_{\tau}(s_\tau, f, z_{s_\tau}) \right]
    + \eta^{18} \E\left[ v_{18} \right]
    \\
    s_{\tau+1} &= s_\tau + y_\tau \text{Bernoulli}(p_\text{pass}(\tau, s_\tau))
\end{align*}
where $\eta$ is a calibrated discount factor and households may observe $\epsilon_\tau$ before choosing $y_\tau$.
We assume that $\epsilon_{\tau}$ has a logistic
distribution, which
enables us to compute the value function for the households by backwards
induction.

\subsection{Decision problem}

The decision problem is to choose
(i) an experimental design and
(ii) a policy that maximizes the fraction of children graduating from
secondary school.
A policy $\pi$ is a four-dimensional vector specifying the subsidy amounts
$z_s$ for boys and girls in primary and secondary school.
The welfare function is
\begin{align*}
    W(\pi, \theta) = 
    P_{\pi, \theta} [s_{18} \geq 10]
\end{align*}
where $P_{\pi,\theta}$ represents the probability distribution 
of $s_{18}$ when $\theta$ is the true parameter value and
subsidy amounts are determined by $\pi$.

We suppose the decision maker has access to a first wave of experimental data
and wants to run a second wave to help choose the optimal policy.
The first wave consists of
a sample of $n_1$ children from the original Progresa experiment.
The design of the second wave can differ from the original Progresa experiment in two ways.
First, the decision maker can change the peso amount of the subsidy
offered to boys and girls in each grade.
Second, the decision maker can let the probability of treatment vary by gender.
Thus, the design $\delta_2$ is a six-dimensional vector describing
the amount of the subsidy offered to each group
and the probability of being treated for boys and girls.
We consider $n_1 = n_2 \in \{250, 500, 1000\}$.

Both the experimental subsidy and the chosen policy subsidy
are subject to constraints.
First, the total amount of the chosen subsidy cannot exceed the total amount
of the original Progresa subsidy, although the decision maker can
allocate the subsidy differently across gender and school level.
The decision maker may also allocate the subsidy differently across the treatment and control groups in the experiment,
but the treated-share-weighted average of the experimental subsidy cannot exceed that of the Progresa experiment.
Second, the subsidy must be nonnegative,
preventing the decision maker from taxing children in some grades
to subsidize children in other grades.
We also constrain the subsidy in each grade to be less than $1000$ pesos.

\subsection{Optimal experiment}

We now present the result of applying our proposed experimental design
to this decision problem.
Following Section \ref{section:limit}, we focus on the performance of the proposed method
in statistically challenging settings where the optimal policy and the active set of constraints
are difficult to learn.
Specifically, we evaluate the expected regret of our method when the true parameter is in 
a neighborhood of a point $\theta_0$ where (i) the optimal policy $\pi_0$ is a vertex of the feasible set and
(ii) the nonnegativity constraints are weakly binding, so the decision maker is unsure
which groups should receive a positive subsidy.
We calibrate these points by estimating $\theta_0$ by maximum likelihood on a large holdout set,
subject to the constraints that the KKT conditions of the policy choice problem
hold with $\lambda_{j0} = 0$ for the nonnegativity constraints $j$.
This yields four points $\theta_0$ corresponding to the four vertices of the feasible set
that lie on the budget constraint.

For each sample size and each $\theta_0$ estimated on the holdout set,
we run a Monte Carlo simulation to compute the expected regret of the proposed method.
We sample $\theta$ from a Uniform distribution on a neighborhood of $\theta_0$ and
generate a first wave of experimental data using the original Progresa design.
This first wave is used to estimate $\hat\delta_{2,n}$ following Section \ref{section:method}.
We then generate a second wave of experimental data using the $\hat\delta_{2,n}$
and use both waves to estimate the optimal policy $\hat \pi_{2,n}$.
For comparison, we also generate a second wave of experimental data using the original Progresa design
and a second wave using an oracle optimal design which uses the true $\theta$ and Fisher information
to choose the optimal design, rather than using first-wave plugins for these quantities.

\begin{figure}
    \centering
    \caption{Expected regret by experimental design}
    \label{figure:regret}
    \includegraphics[width=0.9\textwidth]
    {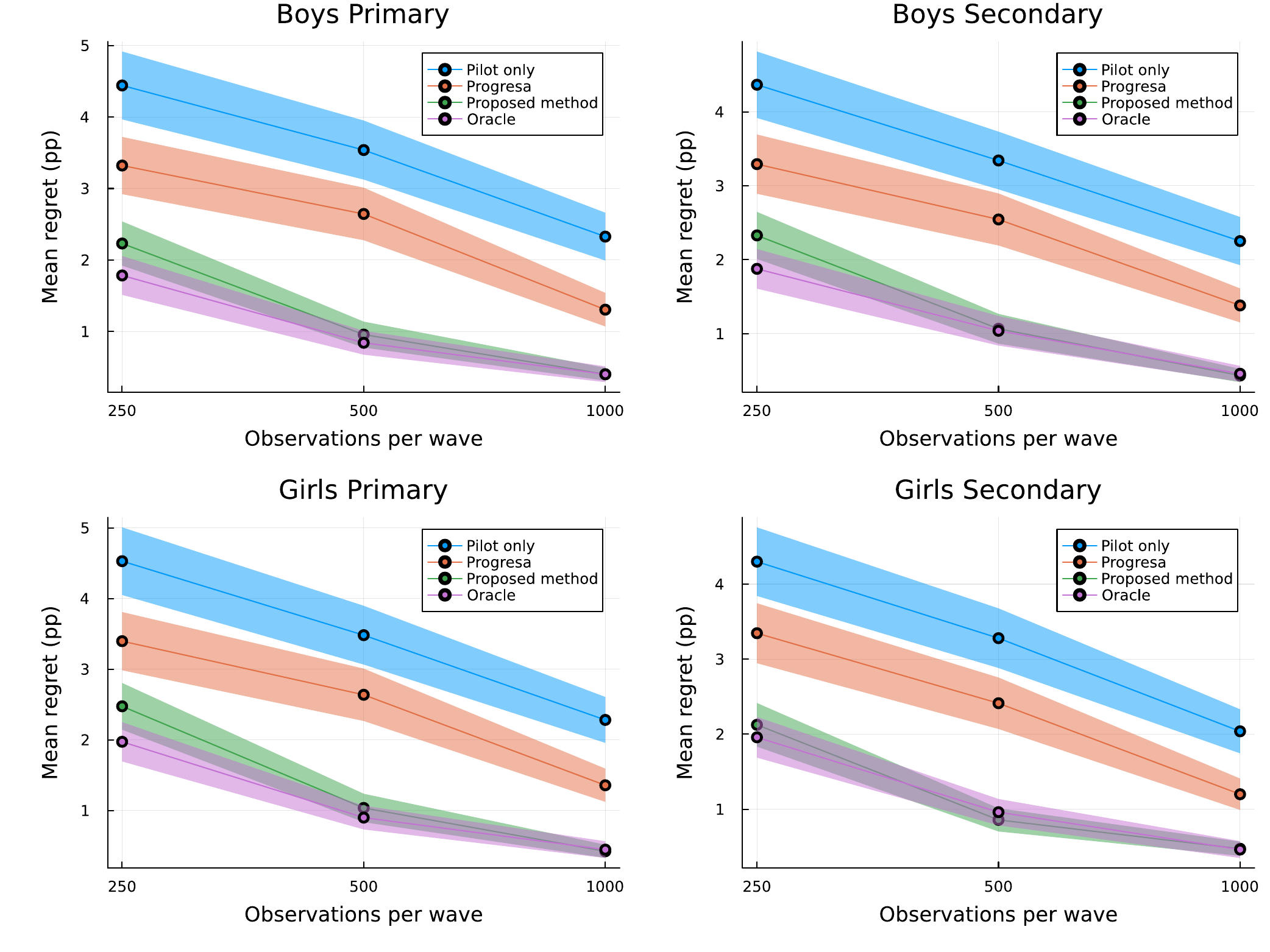}
    \par\smallskip
    \begin{minipage}{1.0\textwidth}
      \footnotesize
      \setstretch{1.0}
      \textit{Note: }
      Each panel shows expected regret and 95\% Monte Carlo confidence intervals
      when $\theta$ is local to one of four reference points $\theta_0$ corresponding to four vertices of the feasible set.
      Each vertex is characterized by a different group of children receiving a positive subsidy.
      ``Pilot only'' disregards the second wave.
      ``Progresa'' generates the second wave using the original Progresa design, just like the first wave.
      ``Proposed method'' generates the second wave using the method proposed in Section \ref{section:method}.
      ``Oracle'' uses the true $\theta$ and Fisher information to construct the approximations.
      Expected regret is computed based on the $>99\%$
      of replications that were successfully solved.
    \end{minipage}
\end{figure}

Figure \ref{figure:regret}
shows the headline expected regret results.
More detailed results are provided in Appendix \ref{appendix:progresa}.
Across all sample sizes and reference points, the proposed method substantially 
reduces expected regret relative to the original Progresa design.
The gain is particularly large for $n_1 = n_2 \in \{500, 1000\}$,
in which regret is reduced by 50-70\%.
The gains at $n_1 = n_2 = 250$ are smaller, suggesting that sampling error in the 
first wave is enough to degrade the quality of the estimated design,
but even this noisily estimated design is more valuable than doubling the sample size of the original Progresa design.
We use the same method to estimate the optimal policy in each experimental design,
so this difference is due to the value of the dataset generated by different designs.

Appendix \ref{appendix:progresa} provides additional details on how the proposed design
differs from the original Progresa design and how these differences affect the quality of the chosen policy.
The optimal design tends to focus exclusively on secondary school and give large experimental subsidies to a small treatment group.
This is despite $\theta_0$ being chosen so that the marginal effect of increasing the subsidy
is equalized across all four groups, so that allocating the subsidy to primary school children is optimal for some values of $\theta$.
This is therefore not an artifact of primary school being unimportant for the policy choice
but rather a reflection of how informative different experimental subsidies are for learning about the optimal policy,
which is jointly determined by the objective, constraints, and information of the model.
In turn, the more informative designs improve the quality of the chosen policy by increasing the 
probability of allocating more subsidy to the group with the highest response.
By focusing the experiment on learning about the decision-relevant
marginal effects of policy changes, the optimal experiment is able to
more effectively inform policy choice.

\begingroup
\setlength{\emergencystretch}{1em}
\printbibliography
\endgroup

\appendix

\section{Proof of Theorem~\ref{theorem:art}}
\label{appendix:art}
Let $\mathcal{D} = \mathcal{Y} \times \mathcal{Z}$ be the sample space for the random data across all units.
Let $\Omega_n = \mathcal{D}^{n_1} \times \mathcal{D}^{n_2} \times [0,1]$
be the sample space of the finite-sample experiment, along with an auxiliary uniform random variable $U$.
The statistics $\delta_{2n}, c_n$ are measurable maps on $\Omega_n$.

Let $\phi_n: \mathcal{D}^{n_1} \times [0,1] \to \Delta$ be a sequence of
adaptive experimental design rules,
so that $\delta_{2,n} = \phi_n(D_{1,n}, U)$.
Given a sequence of adaptive designs $\phi_n$, we construct a sequence
of probability measures $P_{n,h}^{\phi_n}$ on $\Omega_n$ as follows.
Given our potential outcomes model and the law $P_{\epsilon, \nu}$ 
of the structural error term and randomization device,
define the conditional law of the second wave as
\begin{align*}
  &P_{2,n,h}^{\phi_n}(dD_{2,n} \mid D_{1,n}, U)
  \\
  &= \prod_{i \in \mathcal{I}_{2,n}}
  P_{\epsilon, \nu}\left(
    dy\left(z\left(x_i, \nu_i, \phi_n(D_{1,n}, U)\right), x_i, \epsilon_i; \theta_0 + \frac{h}{\sqrt{n}}\right),
    dz\left(x_i, \nu_i; \phi_n(D_{1,n}, U)\right)
  \right)
\end{align*}
and the full law of the finite-sample experiment as
\begin{align*}
  P_{n,h}^{\phi_n}(dU, dD_{1,n}, dD_{2,n}) &= 
  \mu(dU) \times P_{1,n,h}(dD_{1,n}) \times P_{2,n,h}^{\phi_n}(dD_{2,n} \mid D_{1,n}, U)
  .
\end{align*}
The law $P_{n,h}^{\phi_n}$ reflects that the adaptive structure
induces dependence between the first and second wave of the experiment.

In this appendix, we use the superscript $\phi_n$ on probabilities and expectations in the finite-sample experiment
when those quantities depend on the adaptive design rule.
We omit it when $d \in \Delta$ is fixed or the quantity does not otherwise depend on the design rule.
Proofs of supporting lemmas are deferred to Appendix \ref{appendix:auxiliary}.

\begin{lemma}
  \label{lemma:lindeberg}
  Under Assumptions \ref{assumption:largewaves}, \ref{assumption:dqm}, and \ref{assumption:potentialoutcomes},
  the score satisfies the Lindeberg condition
  \begin{align*}
    \frac{1}{n} \sum_{i \in \mathcal{I}_{1,n}}
    \E_0\left[
      \|\psi_i(\delta_1)\|^2
      \1[\|\psi_i(\delta_1)\| > \varepsilon \sqrt{n}]
    \right]
    &\to 0
    \\
    \frac{1}{n} \sum_{i \in \mathcal{I}_{2,n}}
    \E_0\left[
      \|\psi_i(d)\|^2
      \1[\|\psi_i(d)\| > \varepsilon \sqrt{n}]
    \right]
    &\to 0
  \end{align*}
  for every $d \in \Delta$ and $\varepsilon > 0$.
\end{lemma}

The following result extends the standard log-likelihood expansion
to our setting which conditions on covariates and uses the
uniform DQM in Assumption \ref{assumption:dqm}.

\begin{lemma}
    \label{lemma:loglik}
    Suppose Assumptions \ref{assumption:largewaves},
    \ref{assumption:dqm},
    \ref{assumption:potentialoutcomes},
    and \ref{assumption:information}
    hold.
    Then for any fixed $d \in \Delta$ and any $h$,
    \begin{align*}
      \Lambda_{2,n}^h(d)
        =
        h' A_{2,n}(d) - \frac{1}{2} h' \rho J_2(d) h + o_{P_0}(1)
        .
    \end{align*}
\end{lemma}

\begin{proposition}
  \label{proposition:lan}
  Suppose \Crefrange{assumption:largewaves}{assumption:information} hold.
  Then every subsequence has a further subsequence along which
  $(A_{1,n}, A_{2,n}(\cdot)) \overset{0}{\leadsto} (A_{1,\infty}^0, A_{2,\infty}^0(\cdot))$
  where $A_{1,\infty}^0 \sim N(0, (1-\rho) J_1)$
  and $A_{2,\infty}^0$ is a Gaussian process with mean zero and covariance kernel $\rho J_2(\cdot, \cdot)$,
  independent of $A_{1,\infty}^0$ and $\delta_{2,\infty}^0$.
  Along this subsequence,
  \begin{align*}
    \log \frac{dP_{n,h}^{\phi_n}}{dP_{n,0}^{\phi_n}}
    \overset{0}{\leadsto}
    h' A_{1,\infty}^0
    - \frac{1}{2} h' (1-\rho) J_1 h
    + h' A_{2,\infty}^0(\delta_{2,\infty}^0)
    - \frac{1}{2} h' \rho J_2(\delta_{2,\infty}^0) h
    =: \Lambda_h
  \end{align*}
  for any $h$.
  Moreover, $P_{n,h}^{\phi_n}$ is contiguous to $P_{n,0}^{\phi_n}$ along the full sequence.
\end{proposition}

\begin{proof}
  Assumptions \ref{assumption:largewaves}, \ref{assumption:information},
  Lemma \ref{lemma:lindeberg} and the Lindeberg-Feller CLT implies
  finite-dimensional convergence of the score process.
  Assumption \ref{assumption:stocheq} then implies that
  \begin{align*}
    (A_{1,n}, A_{2,n}(\cdot)) \overset{0}{\leadsto} (A_{1,\infty}^0, A_{2, \infty}^0(\cdot))
  \end{align*}
  where $A_{1,\infty}^0 \sim N(0, (1-\rho) J_1)$ and
  $A_{2,\infty}^0(\cdot)$ is a Gaussian process with mean zero and covariance kernel $\rho J_2(\cdot, \cdot)$,
  independent of $A_{1,\infty}^0$.
  In particular, the covariance kernel of $A_{2,\infty}^0(\cdot)$ must be continuous.
  Since $\Delta$ is compact,
  the random elements $(A_{1,n}, A_{2,n}(\cdot), \delta_{2,n})$ are jointly uniformly tight
  and therefore converge in distribution along a subsequence to some limit 
  $(A_{1,\infty}^0, A_{2,\infty}^0(\cdot), \delta_{2,\infty}^0)$.

  The likelihood ratio may be written as
  \begin{align*}
    \frac{P_{n,h}^{\phi_n}}{P_{n,0}^{\phi_n}}(dU, dD_{1,n}, dD_{2,n}) 
    &= \frac{P_{1,n,h}}{P_{1,n,0}}(dD_{1,n}) 
    \times \frac{P_{2,n,h}^{\phi_n}}{P_{2,n,0}^{\phi_n}}(dD_{2,n} \mid D_{1,n}, U)
    ,
  \end{align*}
  where $h$-independent terms cancel in the likelihood ratio. In particular,
  \begin{align*}
    \frac{P_{2,n,h}^{\phi_n}}{P_{2,n,0}^{\phi_n}}(dD_{2,n} \mid D_{1,n}, U)
    &= \prod_{i \in \mathcal{I}_{2,n}} \frac{p_{y \mid z,x }(y_i(\phi_n(D_{1,n}, U)) \mid z_i(\phi_n(D_{1,n}, U)), x_i; \theta_0 + h/\sqrt{n})}{
      p_{y \mid z,x }(y_i(\phi_n(D_{1,n}, U)) \mid z_i(\phi_n(D_{1,n}, U)), x_i; \theta_0)
    }
    .
  \end{align*}
  By Lemma \ref{lemma:loglik}, we have that
  \begin{align*}
    \sum_{i \in \mathcal{I}_{2,n}} \log\frac{p_{y \mid z,x }(y_i(d) \mid z_i(d), x_i; \theta_0 + h/\sqrt{n})}{p_{y \mid z,x }(y_i(d) \mid z_i(d), x_i; \theta_0)}
    &= h' A_{2,n}(d) - \frac{1}{2} h' \rho J_2(d) h + o_{P_0}(1)
  \end{align*}
  pointwise in $d$.
  By Assumption \ref{assumption:stocheq} and continuity of $J_2(\cdot)$, both the left and right-hand sides
  are stochastically equicontinuous in $d$ and therefore the convergence is uniform in $d$.
  A similar expansion holds for the first wave of the experiment, where $\delta_1$ is fixed.

  Since $J_2(\cdot, \cdot)$ is continuous and $A_{2,\infty}(\cdot)$ has almost surely continuous sample paths,
  we can apply the extended continuous mapping theorem to conclude that
  \begin{align*}
    \log \frac{dP_{n,h}^{\phi_n}}{dP_{n,0}^{\phi_n}}
    \overset{0}{\leadsto}
    h' A_{1,\infty}^0 - \frac{1}{2} h' (1-\rho) J_1 h
    + h' A_{2,\infty}^0(\delta_{2,\infty}^0) - \frac{1}{2} h' \rho J_2(\delta_{2,\infty}^0) h
  \end{align*}
  along this subsequence.

  To verify contiguity along this subsequence,
  it suffices to show that $\E_0[\exp(\Lambda_h)] = 1$
  by Le Cam's first lemma (Theorem 6.4 of \textcite{vandervaartAsymptoticStatistics2000}).
  By the law of iterated expectations,
  \begin{align*}
    \E_0[\exp(\Lambda_h)]
    &=
    \E_0\left[
      \exp\left(h' A_{1,\infty}^0 - \frac{1}{2} h' (1-\rho) J_1 h \right)
    \right.
    \\ &\quad
    \left.
    \times \E_0\left[ 
      \exp\left(h' A_{2,\infty}^0(\delta_{2,\infty}^0) - \frac{1}{2} h' \rho J_2(\delta_{2,\infty}^0) h\right)
      \mid A_{1,\infty}^0, \delta_{2,\infty}^0
    \right]
    \right]
    .
  \end{align*}
  Since the process $A_{2,\infty}^0(\cdot)$ is independent of $A_{1,\infty}^0$ and $\delta_{2,\infty}^0$,
  the distribution of the randomly evaluated process
  $A_{2,\infty}^0(\delta_{2,\infty}^0)$  conditional on $\delta_{2,\infty}^0$ is $N(0, \rho J_2(\delta_{2,\infty}^0))$.
  Thus, the inner conditional expectation is $1$ almost surely by properties of log-normal random variables,
  and likewise the outer expectation.
  This establishes contiguity along the subsequence, and since the subsequence was arbitrary, contiguity holds along the full sequence.
\end{proof}

\begin{proof}[Proof of Theorem \ref{theorem:art}]
  By Proposition \ref{proposition:lan}, $P_{n,h}^{\phi_n}$ is contiguous to $P_{n,0}^{\phi_n}$. 
  Since $(\delta_{2n}, c_n)$ is tight under $h=0$ by assumption,
  Proposition \ref{proposition:lan} implies that
  $(\delta_{2n}, c_n, A_{1,n}, A_{2,n}(\cdot), \\ \log (dP_{n,h}^{\phi_n}/dP_{n,0}^{\phi_n}))$
  is jointly tight under $h=0$.
  Thus, there exists a subsequence along which 
  \begin{align}
    \label{eq:jointconvergence}
    \left(
      \delta_{2,n}, c_n, 
      A_{1,n}, A_{2,n}(\cdot),
      \log\frac{dP_{n,h}^{\phi_n}}{dP_{n,0}^{\phi_n}}
    \right)
    \overset{0}{\leadsto}
    \left(
      \delta_{2,\infty}^0, c_\infty^0, A_{1,\infty}^0, A_{2,\infty}^0, \Lambda_h
    \right)
  \end{align}
  for every $h$.
  The log-likelihood ratio in (\ref{eq:jointconvergence}) is between
  two fixed sequences of probability measures that are contiguous.
  Thus, we can apply Le Cam's third lemma
  (\textcite{vandervaartAsymptoticStatistics2000} Theorem 6.6)
  to conclude that the function
  \begin{align*}
    L_h(E) = \E_0\left[ \1[(\delta_{2,\infty}^0, c_\infty^0) \in E] \exp(\Lambda_h) \right]
  \end{align*}
  defines a probability measure and
  $(\delta_{2,n}, c_n) \overset{h}{\leadsto} L_h$
  along the subsequence for which the convergence in (\ref{eq:jointconvergence}) holds.
  Moreover, since the marginal sequence $(\delta_{2,n}, c_n)$
  converges in distribution under $h$ along the full sequence,
  $L_h$ must be the limiting distribution of $(\delta_{2,n}, c_n)$ under $h$ as well.

  Having derived the limiting distribution of $(\delta_{2,n}, c_n)$ under every $h$,
  we now construct a statistic in a Gaussian environment
  which has this distribution for every $h$.
  The construction is similar to the proof of Theorem 3 in
  \textcite{hiranoAsymptoticRepresentationsSequential2025}
  and uses conditional vector quantile functions
  (\cite{carlierVectorQuantileRegression2016}).

  Let $A_1$ be a Gaussian random vector with mean $(1-\rho) J_1 h$ and covariance matrix $(1-\rho) J_1$,
  $A_2(\cdot)$ be a Gaussian process independent of $A_1$ 
  with mean $\rho J_2(\cdot) h$ and covariance function $\rho J_2(\cdot, \cdot)$,
  and $U$ be a uniform random variable on $[0,1]$ independent of $(A_1, A_2(\cdot))$.
  Let $(\delta_{2,\infty}^h, c_\infty^h)$ be the weak limit of $(\delta_{2,n}, c_n)$ under $h$,
  which exist by assumption.
  Our goal is to construct statistics $(\delta_2,  c)$
  where
  $\delta_2 = \phi(A_1, U)$ and
  $c =  C(A_1, A_2(\phi(A_1, U)), U)$
  such that for any $h$,
  $(\delta_2, c) \overset{h}{\sim} (\delta_{2,\infty}^h,  c_\infty^h)$.
  Let $Q_{\delta}(u \mid a_1)$ be the conditional vector quantile
  function of $\delta_{2,\infty}^0$ given $A_{1,\infty}^0 = a_1$.
  Let $Q_{c}(u \mid a_1, d, a_2)$
  be the conditional vector quantile function of $c_\infty^0$
  given $A_{1,\infty}^0 = a_1$, $\delta_{2\infty}^0 = d$, and $A_{2,\infty}^0(d) = a_2$.
  Note that $(A_{1,\infty}^0, A_{2,\infty}^0(\cdot), U)$
  has the same distribution as $(A_1, A_2(\cdot), U)$
  when $h=0$.

  From $U$, construct uniform random vectors $U_1$ and $U_2$
  on the unit cube that are
  independent of each other and of $(A_1, A_2(\cdot))$.
  Define
  \begin{align*}
      \delta_2 = Q_{\delta}(U_1 \mid A_1)
      &&
      c = Q_{c}(U_2 \mid A_1, \delta_2, A_2(\delta_2))
      .
  \end{align*}
  By construction,
  $(\delta_2,  c) \overset{0}{\sim} (\delta_{2,\infty}^0,  c_\infty^0)$
  and therefore
  $(\delta_{2,n},  c_n) \overset{0}{\leadsto} (\delta_2,  c)$.

  We now verify that the statistic $(\delta_2,  c)$ has 
  the desired distribution under $h \neq 0$ as well.
  For some values of $d$ in $\Delta$,
  $J_2(d)$ may be singular, and therefore $A_2(d)$
  may not admit a density with respect to Lebesgue measure on $\R^\ell$.
  However, the likelihood ratio of $A_2(d)$ under $h$ with respect to $h=0$ is well-defined
  on the support of $A_2(d)$, which is $\text{range}(J_2(d))$.
  This ratio is given by
  \begin{align*}
    \frac{dN(\rho J_2(d) h, \rho J_2(d))}{dN(0, \rho J_2(d))}(a)
    = \exp\left(
      h' a - \frac{1}{2} h' \rho J_2(d) h
    \right).
  \end{align*}

  Let $E_1, E_2$ be Borel sets in $\Delta$ and $\R^k$ respectively.
  Then
  \begin{align*}
      P_h[\delta_2 \in E_1, c \in E_2]
      &= \E_h [ \E_h[ \1[\delta_2 \in E_1, c \in E_2] \mid U_1, U_2] ]
      \\
      &= \E_0[ \E_h[ \1[\delta_2 \in E_1, c \in E_2] \mid U_1, U_2] ]
  \end{align*}
  by the law of iterated expectations and the fact
  that the distribution of $U_1$ and $U_2$ does not depend on $h$.
  Using the likelihood ratio of $A_1$ and $A_2(d)$ under $h$ with respect to $h=0$,
  this is equal to 
  \begin{align*}
    &\E_0\bigg[
      \1\left[
        Q_{\delta}(U_1 \mid A_1) \in E_1,
        Q_{c}(U_2 \mid A_1, Q_{\delta}(U_1 \mid A_1), A_2(Q_{\delta}(U_1 \mid A_1))) \in E_2
      \right]
      \\ & \quad \times
      \exp\left(
        h' A_1 - \frac{1}{2} h' (1-\rho) J_1 h
        +
        h' A_2(Q_{\delta}(U_1 \mid A_1)) - \frac{1}{2} h' \rho J_2(Q_{\delta}(U_1 \mid A_1)) h
      \right)
    \bigg]
    \\
    &= \E_0\left[ 
      \1[ (\delta_{2,\infty}^0, c_\infty^0) \in E_1 \times E_2]
      \exp(\Lambda_h)
    \right]
    = L_h(E_1 \times E_2)
  \end{align*}
  and therefore $(\delta_{2,n}, c_n) \overset{h}{\leadsto} (\delta_2, c)$.
\end{proof}

\begin{corollary}
  \label{corollary:art}
  Suppose \Crefrange{assumption:largewaves}{assumption:information} hold.
  If $(\delta_{2,n}, c_n)$ is tight under $h=0$,
  then every subsequence has a further subsequence along which $(\delta_{2,n}, c_n) \overset{h}{\leadsto} (\delta_2, c)$ for any $h$
  and $(\delta_2, c)$ is a design/policy in the limit experiment.
\end{corollary}

\begin{proof}
  Extract a subsequence along which $(\delta_{2n}, c_n, A_{1,n}, A_{2,n}(\cdot))$ converges weakly under $h=0$.
  Along this subsequence, the convergence in (\ref{eq:jointconvergence})
  holds for any $h$ by the argument in the proof of Theorem \ref{theorem:art},
  since this argument does not use the weak convergence of $(\delta_{2n}, c_n)$ under $h$.
  Then we can use Le Cam's third lemma to conclude that $(\delta_{2n}, c_n) \overset{h}{\leadsto} L_h$
  along this subsequence for any $h$.
  Hence, the conditions of Theorem \ref{theorem:art} hold along this subsequence and the conclusion follows.
\end{proof}

\section{Proof of Theorem~\ref{theorem:lowerbound}}
\label{appendix:lowerbound}
For the remainder of the proofs,
we suppress the superscript $\phi_n$ on $\E_h$ and $P_h$ indicating the dependence
of the data on the adaptive design rule.
Probabilities and expectations in finite samples use the law induced by the specified
adaptive rule $\phi_n$ generating the design $\delta_{2,n}$ given in the statements
of Theorems \ref{theorem:lowerbound} and \ref{theorem:optimality}.
Probabilities and expectations in the limit experiment are likewise taken
over the distribution of a random policy $c$ induced by the associated limiting design $\delta_2$.

\begin{lemma}
  \label{lemma:qplipschitz}
  Under Assumption \ref{assumption:nlp},
  $\Gamma(b) = \argmin_{c \in \mathcal{G}} r(c,b)$ exists and is unique for all $b \in \R^k$.
  Moreover,
  \begin{enumerate}
    \item For any $b_1, b_2 \in \R^k$, $\|\Gamma(b_1) - \Gamma(b_2)\| \leq \kappa^{-1} \|b_1 - b_2\|$
    \item $V$ is continuously differentiable with $\nabla V(b) = \Gamma(b)$
    \item For any $b_1, b_2 \in \R^k$ we have
      $0 \leq r(\Gamma(b_1), b_2) - r(\Gamma(b_2), b_2) \leq \kappa^{-1} \| b_1 - b_2 \|^2$
  \end{enumerate}
\end{lemma}

For $\theta$ in a neighborhood of $\theta_0$, let
$\pi^*(\theta) \in \arg\max_{\pi \in \Pi} W(\pi,\theta)$
be an oracle policy that the decision maker would like to choose if $\theta$ were known.
We decompose regret into the smooth and oracle components
$R_n(c, h) = R_n^s(c, h) + R_n^o(h)$ where
\begin{align*}
  R_n^s(c, h) &= n [ W(\pi_0, \theta_0 + h/\sqrt{n}) - W(\pi_0 + c/\sqrt{n}, \theta_0 + h/\sqrt{n}) ] \\
  R_n^o(h) &= n [ W(\pi^*(\theta_0 + h/\sqrt{n}), \theta_0 + h/\sqrt{n}) - W(\pi_0, \theta_0 + h/\sqrt{n}) ].
\end{align*}

\begin{proposition}
  \label{proposition:epiconvergence}
  Suppose Assumption \ref{assumption:nlp} holds.
  If $m_n \to m$ and $c_n \to c$ are deterministic sequences such that $\pi_0 + c_n/\sqrt{n} \in \Pi$ for all $n$, then
  \begin{align*}
    \liminf_{n\to\infty} R_n^s(c_n, m_n) \geq \begin{cases}
      r(c, B m), & \text{if } c \in \mathcal{G}, \\
      \infty, & \text{otherwise.}
    \end{cases}
  \end{align*}
  Moreover, for every $c \in \mathcal{G}$ and $m \in \R^\ell$ there exists
  a deterministic sequence $c_n \to c$ such that $\pi_0 + c_n/\sqrt{n} \in \Pi$ for all $n$ and
  $R_n^s(c_n, m_n) \to r(c, B m)$ whenever $m_n \to m$.
\end{proposition}

\begin{proof}
  For this proof, use $a_j = \nabla g_j(\pi_0)$ and $G_j = \nabla^2 g_j(\pi_0)$.
  Taylor expansions of centered regret and an active constraint $j$ yield
  \begin{align*}
    R_n^s(c_n, m_n) &= \sqrt{n} D' c_n + c_n' B m_n + \frac{1}{2} c_n' H c_n + o(1)
    \\
    0 &\geq n^{-1/2} a_j' c_n + n^{-1} \frac{1}{2} c_n' G_j c_n + o(n^{-1}).
  \end{align*}
  The latter implies that $a_j' c \leq 0$ for all active constraints $j$.
  Stationarity implies $D = - \sum_{j \in \mathcal{J}_s} \lambda_{0j} a_j$.
  Thus, if $a_j' c < 0$ for some strictly active constraint $j \in \mathcal{J}_s$, 
  then $R_n^s(c_n, m_n) \to \infty$ and the result holds.
  On the other hand, if $c \in \mathcal{G}$, then $a_j' c = 0$ for all $j \in \mathcal{J}_s$ and
  substituting for $D$ in the Taylor expansion of $R_n^s$ yields
  \begin{align*}
    R_n^s(c_n, m_n) &= - \sqrt{n} \sum_{j \in \mathcal{J}_s} \lambda_{0j} a_j' c_n + c_n' B m_n + \frac{1}{2} c_n' H c_n + o(1)
    \\
    &\geq c_n' B m_n + \frac{1}{2} c_n' H c_n + \frac{1}{2} \sum_{j \in \mathcal{J}_s} \lambda_{0j} c_n' G_j c_n + o(1)
  \end{align*}
  and hence $\liminf_{n\to\infty} R_n^s(c_n, m_n) \geq r(c, Bm)$ as desired.
  
  We now turn to constructing the recovery sequence.
  Fix $c \in \mathcal{G}$ and define $\mathcal{J}_c = \{ j \in \mathcal{J}_w \cup \mathcal{J}_s : a_j' c = 0 \}$.
  Let $\nu_n \downarrow 0$ be a sequence dominating the Taylor remainder in the expansion of $g_j$ for all $j \in \mathcal{J}_w \cup \mathcal{J}_s$.
  Linear independence of the constraints implies that there exists a right inverse of the gradient matrix $A = [a_j']_{j \in \mathcal{J}_c}$,
  which we can use to solve for a vector $\tilde c_n$ such that
  $a_j' \tilde c_n = -\frac{1}{2} c' G_j c - \nu_n$ for all $j \in \mathcal{J}_s$
  and $a_j' \tilde c_n = - \frac{1}{2} c' G_j c - 1$ for all $j \in \mathcal{J}_c \cap \mathcal{J}_w$.
  Then the sequence
  $c_n = c + n^{-1/2} \tilde c_n$
  is feasible for all sufficiently large $n$.
  Substituting $c_n$ into the Taylor expansion of $R_n^s$ and taking limits yields
  $R_n^s(c_n, m_n) \to r(c, B m)$ as desired.
\end{proof}

\begin{lemma}
  \label{lemma:valueapproximation}
  Suppose Assumption \ref{assumption:nlp} holds.
  For every compact set $M \subset \R^\ell$,
  \begin{align*}
    \sup_{m\in M} \| \sqrt{n}(\pi^*(\theta_0 + m/\sqrt{n}) - \pi_0) - \Gamma(B m) \| \to 0
    \\
    \sup_{m\in M} | R_n^o(m) + V(B m) | \to 0
  \end{align*}
\end{lemma}

\begin{lemma}
  \label{lemma:tightness}
  Suppose Assumption \ref{assumption:nlp} holds.
  For every $K < \infty$ and data $(D_{1,n}, D_{2,n})$,
  there exists $n_K < \infty$ and $C_K < \infty$ such that 
  for all $n \geq n_K$, the optimal policy $\pi_{n,K}^*$ exists and is unique,
  and $\sqrt{n}\| \pi_{n,K}^* - \pi_0 \| \leq C_K$.
\end{lemma}

\begin{proof}[Proof of Theorem \ref{theorem:lowerbound}]
  Since $\mathcal{R}_{n,K}(\delta_{2,n}, c_n) \geq \mathcal{R}^*_{n,K}(\delta_{2,n})$ for any $c_n$, 
  it suffices to restrict attention to the optimal policy $c_{n,K}^* = \sqrt{n}(\pi^*_{n,K} - \pi_0)$.

  We first take a subsequence along which $\mathcal{R}_{n,K}^*(\delta_{2,n}) \to \liminf_{n\to\infty} \mathcal{R}_{n,K}^*(\delta_{2,n})$.
  If the $\liminf$ is infinite, the result holds trivially, and so we proceed under the assumption that it is finite.
  Since $c_{n,K}^*$ is tight by Lemma \ref{lemma:tightness},
  Theorem \ref{theorem:art} and Corollary \ref{corollary:art} imply that every subsequence has a further subsequence along which 
  $(\delta_{2,n}, c_{n,K}^*) \overset{h}{\leadsto} (\delta_2, c)$ for any $h$,
  where $(\delta_2, c)$ is a design/policy in the limit experiment.

  Since the sequence $c_{n,K}^*$ converges in distribution to $c$ along the subsequence for every $h$,
  we can extract a Skorokhod representation for each $h$ such that $c_{n,K}^* \to c$ almost surely.
  Proposition \ref{proposition:epiconvergence} implies that $\liminf_{n\to\infty} R_n^s(c_{n,K}^*, h) \geq r(c, B h)$ 
  for every $h$ if $c \in \mathcal{G}$ and $\liminf_{n\to\infty} R_n^s(c_{n,K}^*, h) = \infty$ otherwise.
  Lemma \ref{lemma:valueapproximation} implies that $R_n^o(h) \to - V(B h)$ for every $h$.
  Thus, $\liminf_{n\to\infty} R_n(c_{n,K}^*, h) \geq R_\infty(c, B h)$ for every $h$.

  Thus, by Fatou's Lemma and Assumption \ref{assumption:prior},
  \begin{align*}
    \liminf_{n\to\infty} \mathcal{R}_{n,K}^*(\delta_{2,n})
    &=
    \liminf_{n\to\infty} \int \E_h [ R_n(c_{n,K}^*, h) ] q_{n,K}(h) dh
    \\
    &\geq
    \int \liminf_{n\to\infty} \E_h \left[ R_n(c_{n,K}^*, h) \right] q_{\infty,K}(h) dh
    \\
    & \geq
    \int \E_h [ R_\infty(c, B h) ] q_{\infty,K}(h) dh
  \end{align*}
  and hence $\liminf_{n\to\infty} \mathcal{R}_{n,K}^*(\delta_{2,n}) \geq \mathcal{R}_{\infty,K}^{**}$.
\end{proof}

\section{Proof of Theorem \ref{theorem:optimality}}
\label{appendix:optimality}
\subsection{Fixed-\texorpdfstring{$K$}{K} upper bound}

Define the local first-wave estimators
$\hat h_{1,n} = \sqrt{n}(\hat \theta_{1,n} - \theta_0)$ and
$\hat c_{1,n} = \sqrt{n}(\hat \pi_{1,n} - \pi_0)$.
We also define 
\begin{align*}
  \hat S_n(c,b) &= n \hat R_n\left(
    \pi_0 + \frac{c}{\sqrt{n}}, \hat B_n \theta_0 + \frac{b}{\sqrt{n}}
  \right)
  \\
  \hat \Gamma_n(b) &= \argmin_c \quad \hat S_n(c,b)
  \quad \text{s.t.} \quad
  g(\pi_0 + c/\sqrt{n}) \leq 0
\end{align*}
and note that $\hat c_n^\text{os} = \hat \Gamma_n(\hat B_n \hat h_n^\text{os})$.
We also define $\hat \beta_n(m) = \hat B_n (\theta_0 + m/\sqrt{n})$.

We use $\E_K$ to denote expectations with respect to the $K$-truncated prior.
Let $\xi \sim N(0, I)$ independently of the first-wave data.
The simulated posterior mean of $h$ 
and corresponding finite-sample loss minimized by $\hat\delta_{2,n}$ are
\begin{align*}
  \tilde h_n(d, \xi) &= \hat h_{1,n} + \left( \frac{n}{n_1} \hat J_{1,n}^{-1} - \hat J_n(d)^{-1} \right)^{1/2} \xi
  \\
  \hat L_n(d) &= \E_\xi \left[ n \hat V_n(\hat \beta_n(\tilde h_n(d, \xi))) \right]
  .
\end{align*}
The expectation $\E_\xi$ integrates only over the randomness in $\xi$ and not the first-wave data.
The posterior mean of $h$ in the limit experiment after the first wave
and the flat-prior and $K$-truncated losses in the limit experiment are
\begin{align*}
  \hat h_{1,\infty} &= (1-\rho)^{-1} J_1^{-1} A_1
  \\
  L_\infty(d) &= \E_\xi \left[ V( B(\hat h_{1,\infty} + ((1-\rho)^{-1} J_1^{-1} - J(d)^{-1})^{1/2} \xi)) \right]
  \\
  L_\infty^K(d) &= \E_K\left[V(\hat b_{2,K}(d)) \mid A_1, U\right]
  .
\end{align*}
We also define the centered loss criteria
\begin{align*}
  \hat L_n^c(d) = \hat L_n(d) - \hat L_n(d_0)
  &&
  L_\infty^c(d) = L_\infty(d) - L_\infty(d_0)
\end{align*}
for a fixed $d_0 \in \Delta$.
This centering does not affect the minimizer.
We make the dependence of $\hat h_n^\text{os}$ on the design $d \in \Delta$ explicit by
writing it as a process $\hat h_n^\text{os}(d)$ where necessary.
The limit analog
and corresponding flat-prior policy are
\begin{align*}
  \hat h^\text{os}(d) = J(d)^{-1} (A_1 + A_2(d))
  &&
  c_\infty^\text{flat} = \Gamma( B \hat h^\text{os}(d))
  .
\end{align*}
To compare the estimated and limiting value functions,
we define the centered value $D_V(b_1, b_2) = V(b_2) - V(b_1) - \Gamma(b_1)' (b_2 - b_1)$.
The next few results establish convergence of various optimizers and values to their limiting counterparts.

\begin{lemma}
  \label{lemma:estimatedqpstability}
  Suppose Assumptions \ref{assumption:nlp} and \ref{assumption:plugin} hold.
  Then
  \begin{enumerate}
    \item for every compact set $M \subset \R^\ell$,
      $\sup_{m\in M} \| \hat \Gamma_n(\hat B_n m) - \hat c_{1,n} \| = O_P(1)$
    \item With probability approaching one, the optimization problem defining $\hat \Gamma_n$ has a unique solution, and
      $0 \leq - n \hat V_n(\hat \beta_n(m)) \leq \kappa^{-1} \| \hat B_n \|^2 \| m - \hat h_{1,n} \|^2$ for every $m \in \R^\ell$.
  \end{enumerate}
\end{lemma}

\begin{lemma}
  \label{lemma:surrogateconvergence}
  Suppose \Crefrange{assumption:largewaves}{assumption:plugin} hold.
  Then for every fixed $h$ and every compact $M \subset \R^\ell$,
  \begin{enumerate}
    \item $\sup_{m \in M} | n \hat V_n(\hat \beta_n(m)) - D_V(B \hat h_{1,n}, Bm) | = o_{P_h}(1)$
    \item $\sup_{m \in M} \| \hat \Gamma_n(\hat B_n m) - \Gamma(Bm) \| = o_{P_h}(1)$
    \item $\sup_{m \in M} | R_n^s(\hat \Gamma_n(\hat B_n m), m) - V(B m) | = o_{P_h}(1)$.
  \end{enumerate}
\end{lemma}

\begin{lemma}
  \label{lemma:onestepconvergence}
  Suppose \Crefrange{assumption:largewaves}{assumption:plugin} hold.
  Then
  \begin{enumerate}
    \item $\sup_{d\in\Delta} \| \hat h_n^\text{os}(d) - J(d)^{-1}( A_{1,n} + A_{2,n}(d)) \| = o_{P_0}(1)$
      and $\hat h_n^\text{os}(\cdot) \overset{0}{\leadsto} \hat h^\text{os}(\cdot)$ in $\ell^\infty(\Delta)$
    \item $\hat L_n^c(\cdot) \overset{h}{\leadsto} L_\infty^c(\cdot)$ in $\ell^\infty(\Delta)$ for any $h$.
  \end{enumerate}
\end{lemma}

\begin{proposition}
  \label{proposition:upperbound}
  Fix $K < \infty$.
  Suppose \Crefrange{assumption:largewaves}{assumption:plugin} hold.
  Then there exists some $\delta_{2,\infty,K}^\text{flat}$ in the limit experiment such that
  $\delta_{2,\infty,K}^\text{flat}$ minimizes $L_\infty(\cdot)$ and
  \begin{align*}
    \limsup_{n\to\infty} \mathcal{R}_{n,K}(\hat\delta_{2,n}, \hat c_n^\text{os}) 
    \leq \mathcal{R}_{\infty,K}(\delta_{2,\infty,K}^\text{flat}, c_\infty^\text{flat})
  \end{align*}
\end{proposition}

\begin{proof}
  Pass to a subsequence for which the $\limsup$ is attained.
  By Lemmas \ref{lemma:surrogateconvergence} and \ref{lemma:onestepconvergence},
  $\hat c_n^\text{os}$ is tight.
  Corollary \ref{corollary:art} then implies that
  there exists a further subsequence and $\delta_{2,\infty,K}^\text{flat}$ in the limit experiment
  such that for every $h$,
  \begin{align*}
    \left(
      \hat \delta_{2,n}, \hat c_n^\text{os}, A_{1,n}, A_{2,n}(\hat\delta_{2,n})
    \right) \overset{h}{\leadsto}
    \left(
      \delta_{2,\infty,K}^\text{flat}, c_\infty, A_1, A_2(\delta_{2,\infty,K}^\text{flat})
    \right).
  \end{align*}
  First, we show that $\delta_{2,\infty,K}^\text{flat}$ minimizes $L_\infty(\cdot)$.
  Define $f(l, d) = l(d) - \inf_{d\in \Delta} l(d)$.
  This is a continuous function of $(l,d)$ in $C(\Delta) \times \Delta$,
  and the limit process has almost surely continuous sample paths.
  Item 2 of Lemma \ref{lemma:onestepconvergence}, along with joint convergence of
  $\hat\delta_{2,n}$ and $A_{1,n}$, implies that
  \begin{align*}
    \left( \hat \delta_{2,n}, \hat L_n^c(\cdot) \right) \overset{h}{\leadsto} \left( \delta_{2,\infty,K}^\text{flat}, L_\infty^c(\cdot) \right).
  \end{align*}
  Since $\hat \delta_{2,n}$ minimizes $\hat L_n^c(\cdot)$,
  the extended continuous mapping theorem implies that
  $f(L_\infty^c, \delta_{2,\infty,K}^\text{flat}) = 0$,
  and so $\delta_{2,\infty,K}^\text{flat}$ minimizes $L_\infty(\cdot)$.

  Item 1 of Lemma \ref{lemma:onestepconvergence} implies that
  \begin{align*}
    \hat h_n^\text{os}(\hat \delta_{2,n}) - J(\hat \delta_{2,n})^{-1} (A_{1,n} + A_{2,n}(\hat \delta_{2,n})) = o_{P_0}(1)
  \end{align*}
  Contiguity as established in Proposition \ref{proposition:lan} then implies that this is $o_{P_h}(1)$ for every $h$.
  Thus,
  \begin{align*}
    \hat h_n^\text{os}(\hat\delta_{2,n}) \overset{h}{\leadsto} J(\delta_{2,\infty,K}^\text{flat})^{-1} (A_1 + A_2(\delta_{2,\infty,K}^\text{flat})) = \hat h^\text{os}(\delta_{2,\infty,K}^\text{flat}).
  \end{align*}

  Items 2 and 3 of Lemma \ref{lemma:surrogateconvergence} imply that
  \begin{align*}
    \| \hat c_n^\text{os} - \Gamma(B \hat h_n^\text{os}(\hat \delta_{2,n})) \| &= o_{P_h}(1) \\
    R_n^s( \hat c_n^\text{os}, \hat h_n^\text{os}(\hat\delta_{2,n})) &= V(B \hat h_n^\text{os}(\hat\delta_{2,n})) + o_{P_h}(1)
  \end{align*}
  which implies that 
  $(\hat h_n^\text{os}(\hat \delta_{2,n}), \hat c_n^\text{os}) \overset{h}{\leadsto} (\hat h^\text{os}(\delta_{2,\infty,K}^\text{flat}), \Gamma(B \hat h^\text{os}(\delta_{2,\infty,K}^\text{flat})))$
  which identifies the limiting policy as
  $c_\infty = c_\infty^\text{flat} = \Gamma(B \hat h^\text{os}(\delta_{2,\infty,K}^\text{flat}))$ almost surely.

  Twice continuous differentiability of $W$ implies that
  \begin{align*}
    \sup_{c \in C} \sup_{m_1, m_2 \in M} | R_n^s(c, m_1) - R_n^s(c,m_2) - c'B(m_1 - m_2) | \to 0
  \end{align*}
  for every compact $C$ and $M$.
  Since $(\hat c_n^\text{os}, \hat h_n^\text{os}(\hat\delta_{2,n}))$ is tight, this implies that
  \begin{align*}
    R_n^s(\hat c_n^\text{os}, h) &= R_n^s(\hat c_n^\text{os}, \hat h_n^\text{os}(\hat\delta_{2,n})) + \hat c_n^{\text{os}\prime} B (h - \hat h_n^\text{os}(\hat\delta_{2,n})) + o_{P_h}(1) \\
    &= V(B \hat h_n^\text{os}(\hat\delta_{2,n})) + \hat c_n^{\text{os}\prime} B (h - \hat h_n^\text{os}(\hat\delta_{2,n})) + o_{P_h}(1) \\
    &\overset{h}{\leadsto} r(c_\infty^\text{flat}, Bh)
  \end{align*}
  by the continuous mapping theorem.
  Since $R_n^o(h) \to - V(Bh)$ by Lemma \ref{lemma:valueapproximation}, we have
  $R_n(\hat c_n^\text{os}, h) \overset{h}{\leadsto} R_\infty(c_\infty^\text{flat}, B h)$ for every $h$.
  The uniform integrability condition of Assumption \ref{assumption:plugin} then implies that
  \begin{align*}
    \lim_{n\to\infty} \mathcal{R}_{n,K}(\hat\delta_{2,n}, \hat c_n^\text{os}) = \mathcal{R}_{\infty,K}(\delta_{2,\infty,K}^\text{flat}, c_\infty^\text{flat})
  \end{align*}
  along our subsequence, establishing the $\limsup$ along the full sequence.
 \end{proof}

 \subsection{Taking the large-\texorpdfstring{$K$}{K} limit}

It remains to close the gap between the upper bound given in Proposition \ref{proposition:upperbound} 
and the lower bound given in Theorem \ref{theorem:lowerbound}.
Define
\begin{align*}
  \hat b_{1,K} = \E_K[Bh \mid A_1],
  &&
  \hat b_{1,\infty} = B \hat h_{1,\infty},
  &&
  \hat b_{2,\infty}(d) = B \hat h^\text{os}(d).
\end{align*}
The next results show that the posterior means and losses
of the $K$-truncated limit problem
converge to their flat-prior counterparts as $K \to \infty$.

\begin{lemma}
  \label{lemma:limitmean}
  Suppose Assumption \ref{assumption:information} holds.
  Then for any sequence of designs $\delta_K$ in the limit experiment,
  \begin{align*}
    \E_K \left[\| \hat b_{1,K} - \hat b_{1,\infty}\|^2 \right] \to 0
    \\
    \E_K \left[\| \hat b_{2,K}(\delta_K) - \hat b_{2,\infty}(\delta_K) \|^2 \right] \to 0
    .
  \end{align*}
\end{lemma}

\begin{lemma}
  \label{lemma:limitloss}
  Suppose Assumptions \ref{assumption:information} and \ref{assumption:nlp} hold.
  Then for every sequence of designs $\delta_K$ in the limit experiment,
  \begin{align*}
    \E_K\left| 
      L_\infty^K(\delta_K) - V(\hat b_{1,K}) - L_\infty(\delta_K) + V(\hat b_{1,\infty})
    \right| \to 0.
  \end{align*}
\end{lemma}

\begin{proposition}
  \label{proposition:limitgap}
  Suppose Assumptions \ref{assumption:information} and \ref{assumption:nlp} hold.
  Let $\delta_{2,\infty,K}^\text{flat}$ be a sequence of feasible designs,
  possibly depending on $K$,
  each of which minimizes $L_\infty(\cdot)$ in the limit experiment.
  Then $\mathcal{R}_{\infty,K}^{**}$ is finite and
  \begin{align*}
    \lim_{K\to\infty} \left\{ 
      \mathcal{R}_{\infty,K}(\delta_{2,\infty,K}^\text{flat}, c_\infty^\text{flat}) - \mathcal{R}_{\infty,K}^{**}
    \right\}
    =0
    .
  \end{align*}
\end{proposition}

\begin{proof}
  We first establish that $\mathcal{R}_{\infty,K}^{**}$ is finite.
  For any $c \in \mathcal{G}$, we have $R_\infty(c,B h) = r(c,B h) - V(B h) \geq 0$.
  The feasible policy $c=0$ has regret
  $R_\infty(0, B h) = - V(Bh) \leq \frac{1}{2\kappa} \|Bh\|^2$
  by Assumption \ref{assumption:nlp}.
  Since $q_{\infty,K}$ has compact support, it has finite second moments, and therefore
  $\mathcal{R}_{\infty,K}^{**} \leq \int \frac{1}{2\kappa} \|Bh\|^2 q_{\infty,K}(h) dh < \infty$.

  Choose an admissible design sequence $\delta_{2,\infty,K}$ satisfying $\mathcal{R}_{\infty,K}^*(\delta_{2,\infty,K}) \leq \mathcal{R}_{\infty,K}^{**} + \varepsilon_K$
  for some $\varepsilon_K \to 0$.
  By this choice of $\delta_{2,\infty,K}$ and Lemma \ref{lemma:limitloss},
  \begin{align*}
    0 &\leq \E_K[ L_\infty^K(\delta_{2,\infty,K}^\text{flat}) - L_\infty^K(\delta_{2,\infty,K}) ] + \varepsilon_K \\
    &\leq \E_K[ L_\infty(\delta_{2,\infty,K}^\text{flat}) - L_\infty(\delta_{2,\infty,K}) ] + o(1) + \varepsilon_K 
    \to 0
  \end{align*}
  which shows that $0 \leq \mathcal{R}_{\infty,K}^*(\delta_{2,\infty,K}^\text{flat}) - \mathcal{R}_{\infty,K}^{**} \to 0$ as $K \to \infty$.
  Lemma \ref{lemma:qplipschitz} then implies that
  \begin{align*}
    0 &\leq \mathcal{R}_{\infty,K}(\delta_{2,\infty,K}^\text{flat}, c_\infty^\text{flat}) - \mathcal{R}_{\infty,K}^{*}(\delta_{2,\infty,K}^\text{flat}) \\
    &\leq \kappa^{-1} \E_K[\| \hat b_{2,\infty}(\delta_{2,\infty,K}^\text{flat}) - \hat b_{2,K}(\delta_{2,\infty,K}^\text{flat}) \|^2]
  \end{align*}
  which converges to zero by Lemma \ref{lemma:limitmean}.
  Combining the previous two displays gives the desired result.
\end{proof}

\begin{proof}[Proof of Theorem \ref{theorem:optimality}]
  By Theorem \ref{theorem:lowerbound},
  \begin{align*}
    \limsup_{n\to\infty} \mathcal{R}_{n,K}(\hat\delta_{2,n}, \hat c_n^\text{os}) \geq
    \liminf_{n\to\infty} \mathcal{R}_{n,K}(\hat\delta_{2,n}, \hat c_n^\text{os}) \geq \mathcal{R}_{\infty,K}^{**}.
  \end{align*}
  By Proposition \ref{proposition:upperbound},
  the left-hand side is bounded from above by $\mathcal{R}_{\infty,K}(\delta_{2,\infty,K}^\text{flat}, c_\infty^\text{flat})$
  where the limiting design $\delta_{2,\infty,K}^\text{flat}$ may depend on $K$.
  Thus,
  \begin{align*}
    0 \leq \limsup_{n\to\infty} \mathcal{R}_{n,K}(\hat\delta_{2,n}, \hat c_n^\text{os}) - \mathcal{R}_{\infty,K}^{**}
    \leq \mathcal{R}_{\infty,K}(\delta_{2,\infty,K}^\text{flat}, c_\infty^\text{flat}) - \mathcal{R}_{\infty,K}^{**}
  \end{align*}
  which converges to zero as $K \to \infty$ by Proposition \ref{proposition:limitgap}.
\end{proof}

\clearpage
\newrefsection
\makeatletter
\begin{center}
    {\normalfont\normalsize
        \MakeUppercase{Supplementary Appendix to}\\
        \MakeUppercase{``\@title''}\par
    }
    \vspace{0.5cm}
    {\normalfont\normalsize\scshape
        \@author\par
    }
    {\normalfont\footnotesize
        \@institute\par
        \@date\par
    }
\end{center}
\makeatother

\section{Proofs of Auxiliary Results}
\label{appendix:auxiliary}
\begin{proof}[Proof of Lemma \ref{lemma:lindeberg}]
  We prove the result for the second wave; the first is identical with a fixed design.
  Fix a direction $v \in \R^\ell$ and define $S = v' \psi(y \mid z, x)$.
  Define the DQM remainder
  \begin{align*}
    \omega_{z,x}(h) &= \sqrt{p_{y\mid z,x}(y \mid z,x; \theta_0 + h)}
    - \sqrt{p_{y\mid z,x}(y \mid z,x; \theta_0)} 
    \\ &\quad - \frac{1}{2} h' \psi(y \mid z, x) \sqrt{p_{y\mid z,x}(y \mid z,x; \theta_0)}
    .
  \end{align*}
  Fix $t>0$. On the event $S > 4/t$, $(t/2) S - 1 \geq (t/4) S$.
  Nonnegativity of densities implies that the DQM remainder satisfies
  \begin{align*}
    \omega_{z,x}(- t v) &=
    \sqrt{p_{y\mid z,x}(y \mid z,x; \theta_0 - tv)} + \left( -1 + \frac{t}{2} S \right) \sqrt{p_{y\mid z,x}(y \mid z,x; \theta_0)}
    \\ &\geq \frac{t}{4}|S|\sqrt{p_{y\mid z, x}(y \mid z, x; \theta_0)}
    .
  \end{align*}
  Thus,
  \begin{align*}
    \| \omega_{z,x}(- t v) \|^2 \geq \frac{t^2}{16} \E_0\left[
      \|S\|^2 \1[S > 4/t]
    \right]
    .
  \end{align*}
  Likewise, by considering the event $S < -4/t$, we have
  \begin{align*}
    \| \omega_{z,x}(tv) \|^2 \geq \frac{t^2}{16} \E_0\left[
      \|S\|^2 \1[S < -4/t]
    \right]
  \end{align*}
  and therefore
  \begin{align*}
    \sup_{z,x} \E_0\left[
      \|S\|^2 \1[\|S\| > 4/t]
    \right] \leq \frac{16}{t^2} \left( \sup_{z,x} \| \omega_{z,x}(- t v) \|^2 + \sup_{z,x} \| \omega_{z,x}(t v) \|^2 \right)
  \end{align*}
  which is $o(1)$ by Assumption \ref{assumption:dqm}.
  Since the score is finite-dimensional, this directional bound implies that for any $d\in\Delta$,
  $\E_0\left[\| \psi_i(d) \|^2 \1[\| \psi_i(d) \| > \varepsilon \sqrt{n}]\right] \to 0$
  uniformly in $z,x$ and the result follows.
\end{proof}

\begin{proof}[Proof of Lemma \ref{lemma:loglik}]
    In this proof, we will use the shorthand
    $p_n(\cdot \mid \cdot, \cdot) :=
    p_{y\mid z, x}(\cdot \mid \cdot, \cdot; \theta_0 + h/\sqrt{n})$ 
    and
    $p(\cdot \mid \cdot, \cdot) 
    := p_{y\mid z, x}(\cdot \mid \cdot, \cdot; \theta_0) $.
    We will also suppress the dependence of 
    $y_i$ and $z_i$ on $d$ to aid readability.

    Define the random variable 
    $S_{n,i} = 2 (\sqrt{p_n/p}(y_i \mid z_i, x_i) - 1 )$
    which is well-defined $P_0$-almost everywhere.
    We write a Taylor expansion of the log-likelihood ratio in terms of $S_{ni}$:
    \begin{align}
        \label{eq:taylor}
        \sum_{i \in \mathcal{I}_{2,n}}
        \log \frac{p_n}{p}(y_i \mid z_i, x_i )
        = \sum_{i \in \mathcal{I}_{2,n}} \left(
            S_{ni}
            - \frac{1}{4} S_{ni}^2 
            + \frac{1}{2} S_{ni}^2 \text{rem}(S_{ni})
        \right)
    \end{align}
    where $\text{rem}(s)\to 0 $ as $s \to 0$.

    First, observe that
    \begin{align*}
      \E\left[
        \sum_{i \in \mathcal{I}_{2,n}} S_{ni}
      \right]
      &=
      - \sum_{i \in \mathcal{I}_{2,n}} \E\left[
        \int (\sqrt{p_n} - \sqrt{p})^2 d\mu
      \right]
      \\
      &=
      - \sum_{i \in \mathcal{I}_{2,n}} \E\left[
        \int \left( \frac{1}{2\sqrt{n}} h' \psi_i(d) \sqrt{p} \right)^2 d\mu
        + s_n(z_i, x_i)
      \right]
      \\
      &\to
      - \frac{1}{4} h' \rho J_2(d) h
    \end{align*}
    where the last line follows from Assumptions \ref{assumption:largewaves}, \ref{assumption:dqm}, and \ref{assumption:information}.
    Next, observe that
    \begin{align*}
      &\Var\left(
        \sum_{i \in \mathcal{I}_{2,n}} \left(
          S_{ni} - \frac{1}{\sqrt{n}} h' \psi_i(d)
        \right)
      \right)
      \\& \leq 
      \sum_{i \in \mathcal{I}_{2,n}} \E\left[
        \left( S_{ni} - \frac{1}{\sqrt{n}} h' \psi_i(d) \right)^2
      \right]
      \\
      &= 4 \sum_{i \in \mathcal{I}_{2,n}} \E\left[
        \int \left( \sqrt{p_n} - \sqrt{p} - \frac{1}{2\sqrt{n}} h' \psi_i(d) \sqrt{p} \right)^2 d\mu
      \right]
      \\
      &\to 0
    \end{align*}
    again by DQM.
    Combining the previous two displays, we have that
    the mean and variance of
    \begin{align*}
      \sum_{i \in \mathcal{I}_{2,n}} S_{ni}
      - \frac{1}{\sqrt{n}} \sum_{i \in \mathcal{I}_{2,n}} h' \psi_i(d)
      + \frac{1}{4} h' \rho J_2(d) h
      \to 0
    \end{align*}
    and hence this expression converges to zero in probability.

    Now define $U_{ni} = \sqrt{n} S_{ni} - h' \psi_i(d)$.
    We have that
    \begin{align*}
      &\E\left| \frac{1}{n} \sum_{i \in \mathcal{I}_{2,n}} \left( 2 h' \psi_i(d) U_{ni} + U_{ni}^2 \right) \right|\\
      &\leq
      2 \left(
        \frac{1}{n} \sum_{i \in \mathcal{I}_{2,n}} \E (h' \psi_i(d))^2
      \right)^{1/2}
      \left(
        \frac{1}{n} \sum_{i \in \mathcal{I}_{2,n}} \E U_{ni}^2
      \right)^{1/2}
      + \frac{1}{n} \sum_{i \in \mathcal{I}_{2,n}} \E U_{ni}^2
      \\
      &= O(1) o(1) + o(1)
    \end{align*}
    by the preceding argument.
    Since $n S_{ni}^2 = (h'\psi_i(d))^2 + 2 h' \psi_i(d) U_{ni} + U_{ni}^2$,
    \begin{align*}
      \sum_{i \in \mathcal{I}_{2,n}} S_{ni}^2
      &= \frac{1}{n} \sum_{i \in \mathcal{I}_{2,n}} (h'\psi_i(d))^2 + o_{P_0}(1)
      \\
      &= \rho h' J_2(d) h + o_{P_0}(1)
    \end{align*}
    by Assumption \ref{assumption:largewaves}, Assumption \ref{assumption:information}, 
    and the Lindeberg law of large numbers, which applies to the scores by
    Lemma \ref{lemma:lindeberg}.

    Combining the preceding calculations, we have that
    \begin{align*}
        \sum_{i \in \mathcal{I}_{2,n}} \left(
            S_{ni}
            - \frac{1}{4} S_{ni}^2 
        \right)
        =
        \frac{1}{\sqrt{n}} \sum_{i \in \mathcal{I}_{2,n}}
        h' \psi_i(d)
        - \frac{1}{2} h' \rho J_2(d) h
        +
        o_{P_0}(1)
        .
    \end{align*}
    To complete the proof, it remains to show
    that the remainder term in (\ref{eq:taylor}) is $o_{P_0}(1)$.
    Observe that 
    \begin{align*}
      \sum_{i \in \mathcal{I}_{2,n}} S_{ni}^2 \text{rem}(S_{ni})
      \leq 
      \max_{i \in \mathcal{I}_{2,n}}
      |\text{rem}(S_{ni})| \sum_{i \in \mathcal{I}_{2,n}} S_{ni}^2
    \end{align*}
    where $\sum_{i \in \mathcal{I}_{2,n}} S_{ni}^2$ is $O_{P_0}(1)$,
    so we want to show that 
    $\max_{i \in \mathcal{I}_{2,n}} |\text{rem}(S_{ni})| = o_{P_0}(1)$.
    For any $\varepsilon > 0$ we have
    \begin{align*}
        P\left(
            \max_{i \in \mathcal{I}_{2,n}} |S_{ni}| > \varepsilon
        \right)
        &\leq \sum_{i \in \mathcal{I}_{2,n}} P(|S_{ni}| > \varepsilon)
        .
    \end{align*}
    Each term in the sum obeys
    \begin{align*}
        P(|S_{ni}| > \varepsilon)
        &\leq P\left((h'\psi_i(d))^2 > \frac{1}{4} n \varepsilon^2 \right)
        + P\left(U_{ni}^2 > \frac{1}{4} n \varepsilon^2\right)
        \\
        &= P\left(
            (h'\psi_i(d))^2 
            \1[(h'\psi_i(d))^2 
            > \frac{1}{4} n \varepsilon^2] 
            > \frac{1}{4} n \varepsilon^2 
        \right)
        + P(U_{ni}^2 > \frac{1}{4} n \varepsilon^2)
        \\
        &\leq 4 n^{-1} \varepsilon^{-2} \E \left[ 
            (h'\psi_i(d))^2 
            \1[(h'\psi_i(d))^2 
            > \frac{1}{4} n \varepsilon^2] 
        \right]
        + 4 n^{-1} \varepsilon^{-2} \E \left[ U_{ni}^2  \right]
        .
    \end{align*}
    By Lemma \ref{lemma:lindeberg},
    \begin{align*}
        \sum_{i \in \mathcal{I}_{2,n}}
        4 n^{-1} \varepsilon^{-2} \E \left[ 
            (h'\psi_i(d))^2  \1[(h'\psi_i(d))^2 > \frac{1}{4} n \varepsilon^2] 
        \right]
        \to 0
    \end{align*}
    and moreover
    \begin{align*}
        \sum_{i \in \mathcal{I}_{2,n}}
        4 n^{-1} \varepsilon^{-2} \E \left[ U_{ni}^2  \right]
        \to 0
        .
    \end{align*}
    This shows that $\max_{i \in \mathcal{I}_{2,n}} |\text{rem}(S_{ni})| = o_{P_0}(1)$
    and the last term in the Taylor expansion
    is $o_{P_0}(1) O_{P_0}(1) = o_{P_0}(1)$.
\end{proof}

\begin{proof}[Proof of Lemma \ref{lemma:qplipschitz}]
  Existence and uniqueness follow from the fact that $\mathcal{G}$ is a closed convex cone
  and $r(\cdot, b)$ is strongly convex for all $b$.
  Optimality of $\Gamma(b_1)$ implies that
  $(\bar H \Gamma(b_1) + b_1)' (\Gamma(b_2) - \Gamma(b_1)) \geq 0$.
  Add this to the corresponding inequality for $\Gamma(b_2)$ and use strong convexity to obtain
  \begin{align*}
    \kappa \| \Gamma(b_1) - \Gamma(b_2) \|^2
    &\leq (b_1 - b_2)' (\Gamma(b_2) - \Gamma(b_1))
    \leq \|b_1 - b_2\| \|\Gamma(b_1) - \Gamma(b_2)\|
  \end{align*}
  which proves item 1.
  Danskin's theorem implies that $\nabla V(b) = \Gamma(b)$ and continuity follows from item 1,
  thus proving item 2.
  For item 3, the first inequality is implied by optimality of $\Gamma(b_2)$
  and the second is because
  \begin{align*}
    r(\Gamma(b_1), b_2) - r(\Gamma(b_2), b_2)
    &= r(\Gamma(b_1), b_1) - r(\Gamma(b_2), b_1) + (b_2 - b_1)' (\Gamma(b_1) - \Gamma(b_2)) \\
    &\leq (b_2 - b_1)' (\Gamma(b_1) - \Gamma(b_2)) \\
    &\leq \kappa^{-1} \| b_1 - b_2 \|^2
  \end{align*}
  where the inequalities are from (i) optimality of $\Gamma(b_1)$ and (ii) item 1.
\end{proof}

\begin{proof}[Proof of Lemma \ref{lemma:valueapproximation}]
  Assumption \ref{assumption:nlp} implies there exists an $a > 0$ and a neighborhood $N_\pi$ of $\pi_0$ such that
  the quadratic growth condition holds:
  \begin{align*}
    W(\pi_0, \theta_0) - W(\pi, \theta_0) \geq a \| \pi - \pi_0 \|^2
  \end{align*}
  for all $\pi \in \Pi \cap N_\pi$.
  The compact level set condition, continuity of $W$, and uniqueness of $\pi_0$ imply that
  any sequence of feasible optimizers $\pi^*(\theta)$ for $\theta \to \theta_0$ eventually lies in $N_\pi$.
  Continuity of $\nabla^2_{\pi,\pi} W$ and positive definiteness of $H$ imply
  that $W(\cdot, \theta)$ is strongly concave on $N_\pi \cap \Pi$ and hence $\pi^*(\theta)$ is
  unique for all $\theta$ sufficiently close to $\theta_0$.
  For such a sequence $\pi^*(\theta)$, the previous display implies
  \begin{align*}
    a \| \pi^*(\theta) - \pi_0 \|^2
    &\leq [W(\pi^*(\theta), \theta) - W(\pi^*(\theta), \theta_0)] 
    - [W(\pi_0, \theta) - W(\pi_0, \theta_0)]
    .
  \end{align*}
  By the mean value theorem and continuity of $\nabla^2 W$, the right side is bounded by 
  $C \| \theta - \theta_0 \| \| \pi^*(\theta) - \pi_0 \|$ for some constant $C$.
  Divide both sides by $\| \pi^*(\theta) - \pi_0 \|$ (the case of equality is trivial) to obtain
  \begin{align*}
    \| \pi^*(\theta) - \pi_0 \| \leq a^{-1} C \| \theta - \theta_0 \|
  \end{align*}
  which is bounded by $a^{-1} C \sup_{m \in M} \| m \| / \sqrt{n}$ for $\theta = \theta_0 + m/\sqrt{n}$
  and sufficiently large $n$.

  Assume to the contrary that the optimizer convergence result does not hold.
  Then there exists $n_j$, $m_j \in M$, and feasible optimizers $\pi_j = \pi^*(\theta_0 + m_j/\sqrt{n_j})$ such that
  the result does not hold.
  Pass to a subsequence along which $m_j \to m$ for some $m \in M$.
  By the previous paragraph, $c_j = \sqrt{n_j}(\pi_j - \pi_0)$ is bounded,
  and so there is a further convergent subsequence $c_j \to c$.
  By Proposition \ref{proposition:epiconvergence}, $\liminf_{j\to\infty} R_{n_j}^s(c_j, m_j) \geq r(c, B m)$.

  Also by Proposition \ref{proposition:epiconvergence},
  for any $\tilde c \in \mathcal{G}$, there is a sequence
  $\tilde c_j \to \tilde c$ which is feasible for all $j$ and $R_{n_j}^s(\tilde c_j, m_j) \to r(\tilde c, B m)$.
  Optimality of $\pi_j$ then implies
  \begin{align*}
    R_{n_j}^s(c_j, m_j) \leq R_{n_j}^s(\tilde c_j, m_j)
  \end{align*}
  and therefore $r(c, B m) \leq r(\tilde c, B m)$ for all $\tilde c \in \mathcal{G}$.
  By Lemma \ref{lemma:qplipschitz}, $c = \Gamma(B m)$ and $\| c_j - \Gamma(B m_j) \| \to 0$,
  a contradiction.

  Finally, we establish the uniform convergence of the value function.
  Take the $\tilde c$ from the previous paragraph to be $\Gamma(B m)$.
  Then $R_{n_j}^s(\tilde c_j, m_j) \to r(\Gamma(B m), B m) = V(B m)$.
  Optimality of $\pi_j$ implies $\limsup_{j\to\infty} R_{n_j}^s(c_j, m_j) \leq V(B m)$,
  and therefore $\lim_{j\to\infty} R_{n_j}^s(c_j, m_j) = V(B m)$.
  Uniformity follows from the same contradiction argument as in the previous paragraph.
\end{proof}

\begin{proof}[Proof of Lemma \ref{lemma:tightness}]
  For $n$ large enough, all parameters $\theta = \theta_0 + h/\sqrt{n}$ 
  with $h$ in the support of the posterior lie in a neighborhood $N_\theta$ on which
  the level set condition, twice continuous differentiability, and strong concavity of $W$ hold.
  The uniqueness argument in the proof of Lemma \ref{lemma:valueapproximation}
  applies uniformly over $N$, so $\pi_{n,K}^*$ exists and is unique.
  Optimality of $\pi_{n,K}^*$ implies that
  \begin{align*}
    a \|\pi_{n,K}^* - \pi_0\|^2 
    &\leq \E_K\left[
      W\left(\pi_{n,K}^*, \theta_0 + \frac{h}{\sqrt{n}}\right) - W(\pi_{n,K}^*, \theta_0)
    \right. \\ & \quad \left.
      - W\left(\pi_0, \theta_0 + \frac{h}{\sqrt{n}}\right) + W(\pi_0, \theta_0)
      \mid D_{1,n}, D_{2,n} 
    \right]
    \\
    &\leq \frac{CK}{\sqrt{n}}\|\pi_{n,K}^* - \pi_0\|
  \end{align*}
  which establishes the desired bound.
\end{proof}

\begin{proof}[Proof of Lemma \ref{lemma:estimatedqpstability}]
  Lemma \ref{lemma:valueapproximation}, consistency of first wave estimates, and continuous differentiability of $W$ imply that
  $\hat H_n \overset{P}{\to} H$ and $\hat B_n \overset{P}{\to} B$.
  Thus, Assumption \ref{assumption:nlp} implies that with probability approaching one,
  $u' \hat H_n u \geq \frac{\kappa}{2} \| u \|^2$ for all $u \in \text{span}(\Pi - \Pi)$.
  Convexity of $\Pi$ and optimality of $\hat \pi_{1,n}$ then imply
  uniqueness of $\hat \Gamma_n(b)$.
  The argument in the proof of Lemma \ref{lemma:qplipschitz} then delivers both results.
\end{proof}

\begin{proof}[Proof of Lemma \ref{lemma:surrogateconvergence}]
  Since $\hat c_{1,n}$ is tight by Lemma \ref{lemma:valueapproximation},
  a Taylor expansion of $R_n^s(\cdot, m)$ around $(\hat \pi_{1,n}, \hat \theta_{1,n})$ gives,
  for every compact $C \subseteq \R^k$,
  \begin{align*}
    \sup_{c \in C} \sup_{m \in M} | \hat S_n(c, \hat B_n m) - R_n^s(c, m) + R_n^s(\hat c_{1,n}, m) | = o_{P_h}(1).
  \end{align*}
  Minimizing each term over $c$ (justified by the localization of Lemmas \ref{lemma:valueapproximation} and \ref{lemma:estimatedqpstability})
  and using Lemma \ref{lemma:valueapproximation} to substitute $\min_c R_n^s(c, m) = V(Bm) + o_{P_h}(1)$ yields
  \begin{align*}
    \sup_{m\in M} |n \hat V_n(\hat \beta_n(m)) - V(Bm) + R_n^s(\hat c_{1,n}, m)| = o_{P_h}(1).
  \end{align*}
  Twice continuous differentiability of $W$ implies that
  \begin{align*}
    \sup_{m \in M} | R_n^s(\hat c_{1,n},m) - R_n^s(\hat c_{1,n},\hat h_{1,n}) - \hat c_{1,n}'B(m - \hat h_{1,n}) | \to 0.
  \end{align*}
  Lemma \ref{lemma:valueapproximation} gives $R_n^s(\hat c_{1,n},\hat h_{1,n}) = V(B\hat h_{1,n}) + o_{P_h}(1)$.
  Lemma \ref{lemma:valueapproximation} also gives $\hat c_{1,n} = \Gamma(B\hat h_{1,n}) + o_{P_h}(1)$.
  Thus,
  \begin{align*}
    \sup_{m\in M} | n \hat V_n(\hat \beta_n(m)) - V(Bm) + V(B \hat h_{1,n}) + \Gamma(B\hat h_{1,n})'B(m - \hat h_{1,n}) | = o_{P_h}(1)
  \end{align*}
  which is item 1 of the lemma.

  For item 2, suppose the conclusion fails.
  Then there exists a subsequence and $m_n$ along which
  $\| \hat \Gamma_n(\hat B_n m_n) - \Gamma(B m_n) \| \geq \epsilon$ with nonvanishing probability.
  By Lemma \ref{lemma:estimatedqpstability}, $\hat \Gamma_n(\hat B_n m_n)$ is tight, and therefore has a weakly convergent subsequence.
  On a Skorokhod representation, $(m_n, \hat \Gamma_n(\hat B_n m_n)) \to (m, c)$ almost surely along the subsequence.
  Then Proposition \ref{proposition:epiconvergence} implies that
  there exists a feasible sequence $c_n \to \Gamma(B m)$ almost surely such that
  \begin{align*}
    \hat S_n(\hat \Gamma_n(\hat B_n m_n), \hat B_n m_n) \leq \hat S_n(c_n, \hat B_n m_n)
    \implies
    r(c, Bm) \leq r(\Gamma(Bm), Bm)
  \end{align*}
  which, by uniqueness of the minimizer, implies $\hat \Gamma_n(\hat B_n m_n) \to \Gamma(Bm)$ in probability, a contradiction.

  For item 3, note that the proof of item 1 implies that
  \begin{align*}
    R_n^s(\hat\Gamma_n(\hat B_n m), m) &= D_V(B\hat h_{1,n}, Bm) + V(B\hat h_{1,n}) + \Gamma(B\hat h_{1,n})'B(m - \hat h_{1,n}) + o_{P_h}(1) \\
    &= V(Bm) + o_{P_h}(1)
  \end{align*}
  uniformly in $m\in M$.
\end{proof}

\begin{proof}[Proof of Lemma \ref{lemma:onestepconvergence}]
  Since $A_{2,n}(\cdot)$ converges to a Gaussian process, its covariance kernel is
  asymptotically equicontinuous.
  Assumptions \ref{assumption:largewaves}, \ref{assumption:information}, 
  and compactness of $\Delta$ imply that $\sup_{d\in\Delta} \| \hat J_n(d) - J(d) \| = o_{P_0}(1)$.
  Stochastic equicontinuity of $A_{2,n}(\cdot)$, uniform consistency of $\hat J_n(\cdot)$, 
  uniform nonsingularity of $J(\cdot)$, and the continuous mapping theorem then imply the first result.

  The first-wave score equation and Assumption \ref{assumption:plugin} imply
  \begin{align*}
    \hat h_{1,n} = ( ( 1-\rho) J_1)^{-1} A_{1,n} + o_{P_0}(1).
  \end{align*}
  Then the second-wave score expansion of Assumption \ref{assumption:plugin} gives the first result of item 1.
  Proposition \ref{proposition:lan} then gives the second result of item 1.

  Proposition \ref{proposition:lan} and contiguity of the first wave also give
  \begin{align*}
    \hat h_{1,n} \overset{h}{\leadsto} \hat h_{1,\infty}
    &&
    \hat J_{1,n}, \hat J_n(\cdot) \overset{P_h}{\to} J_1, J(\cdot)
    .
  \end{align*}
  Tightness of $\hat h_{1,n}$ and uniform boundedness in probability of
  $\frac{n}{n_1}\hat J_{1,n}^{-1} - \hat J_n(\cdot)^{-1}$ imply that $\tilde h_n(d, \xi)$ is tight uniformly in $d \in \Delta$.
  Lemma \ref{lemma:surrogateconvergence} then implies that
  \begin{align*}
    \sup_{d \in \Delta} \left| n \hat V_n(\hat\beta_n(\tilde h_n(d, \xi))) - D_V(B \hat h_{1,n}, B \tilde h_n(d, \xi)) \right| = o_{P_h}(1).
  \end{align*}
  Lemma \ref{lemma:estimatedqpstability} gives
  \begin{align*}
    | n \hat V_n(\hat\beta_n(\tilde h_n(d, \xi))) - D_V(B \hat h_{1,n}, B \tilde h_n(d, \xi)) | \leq O_{P_h}(1) \| \xi \|^2
  \end{align*}
  which establishes uniform integrability of the left-hand side. This implies that
  \begin{align*}
    \sup_{d \in \Delta} \E_\xi \left| n \hat V_n(\hat \beta_n(\tilde h_n(d, \xi))) - D_V(B \hat h_{1,n}, B \tilde h_n(d, \xi)) \right| = o_{P_h}(1)
  \end{align*}
  and therefore
  \begin{align*}
    \sup_{d \in \Delta} | \hat L_n(d) - \E_\xi[D_V(B \hat h_{1,n}, B \tilde h_n(d, \xi))] | = o_{P_h}(1).
  \end{align*}
  The continuous mapping theorem then implies that
  \begin{align*}
    \E_\xi[D_V(B \hat h_{1,n}, B \tilde h_n(\cdot, \xi))] &\overset{h}{\leadsto} \E_\xi[D_V(B \hat h_{1,\infty}, B(\hat h_{1,\infty} + ((1-\rho)^{-1} J_1^{-1} - J(\cdot)^{-1})^{1/2} \xi))] 
    \\
    &= L_\infty(\cdot) - V(B \hat h_{1,\infty})
  \end{align*}
  in $\ell^\infty(\Delta)$.
  Applying the continuous map $f \to f - f(d_0)$ yields $\hat L_n^c(\cdot) \overset{h}{\leadsto} L_\infty^c(\cdot)$ in $\ell^\infty(\Delta)$.
\end{proof}

\begin{proof}[Proof of Lemma \ref{lemma:limitmean}]
  We prove the second result; the first is similar.
  Let $\hat h_{2,\infty,K}$ be the mean of the posterior distribution on $h$
  restricted to the $K$-ball $\mathcal{B}_K$.
  Fix $M > 0$.
  On the event where $\text{dist}(h, \mathcal{B}_K^c) > 2M$ and $\| \hat h_{2,\infty,\infty} - h \| \leq M$,
  we have $\text{dist}(\hat h_{2,\infty,\infty}, \mathcal{B}_K^c) \geq M$.
  Thus,
  \begin{align*}
    \| \hat h_{2,\infty,K} - \hat h_{2,\infty,\infty} \| \leq \frac{\E[\|\xi\|\1[ \| \xi \| \geq M]]}{1 - P(\|\xi\|\geq M)}
  \end{align*}
  for a $N(0, J(d)^{-1})$ random vector $\xi$.
  The right-hand side tends to zero uniformly over $d\in\Delta$ as $M\to\infty$.

  On the other hand, the event where $\text{dist}(h, \mathcal{B}_K^c) \leq 2M$
  has mass $1 - (1 - 2M/K)^\ell \leq 2\ell M/K$.
  The inequality $\E_K \|\hat h_{2,\infty,K} - \hat h_{2,\infty,\infty} \|^4 \leq \E_K \|\hat h_{2,\infty,\infty} - h \|^4$
  establishes a uniform bound on fourth moments over designs,
  and hence uniform integrability,
  so that $\E_K \| \hat h_{2,\infty,K} - \hat h_{2,\infty,\infty} \|^2 \to 0$.
  The result then follows by multiplying by $B$.
\end{proof}

\begin{proof}[Proof of Lemma \ref{lemma:limitloss}]
  We begin by showing that the $K$-truncated and flat posteriors in the limit experiment converge in total variation as $K$ grows.
  Let $Q_{1,K} \propto N(\hat h_{1,\infty}, (1-\rho)^{-1} J_1^{-1}) \1[\|h\| \leq K]$ 
  and $Q_{1,\infty} = N(\hat h_{1,\infty}, (1-\rho)^{-1}J_1^{-1})$ be the first wave posteriors in the limit experiment.
  Then $\| Q_{1,K} - Q_{1,\infty} \|_\text{TV} = Q_{1,\infty}(\|h\| > K)$.
  On the event where $\text{dist}(h, \mathcal{B}_K^c) > 2M$ and $\| \hat h_{1,\infty} - h \| \leq M$,
  we have $Q_{1,\infty}(\|h\| > K) \leq P(\|\xi\| \geq M)$.
  The probability of the complement of this event is bounded by
  $P(\text{dist}(h, \mathcal{B}_K^c) \leq 2M) + P(\| \hat h_{1,\infty} - h \| > M)$.
  Thus,
  \begin{align*}
    \E_K \| Q_{1,K} - Q_{1,\infty} \|_\text{TV} &\leq P(\|\xi\| \geq M) + P(\text{dist}(h, \mathcal{B}_K^c) \leq 2M) + P(\| \hat h_{1,\infty} - h \| > M)
    \\
    &\leq P(\|\xi\| \geq M) + 2\ell M/K + P(\| \hat h_{1,\infty} - h \| > M)
  \end{align*}
  which tends to zero letting $K \to \infty$ and then $M \to \infty$.

  We now show that the posterior predictive distributions of the second wave data converge in total variation.
  Conditional on $(h, A_1, U)$, the observation $A_2$ follows the same Gaussian law in the $K$-truncated and flat limit experiments.
  Since total variation contracts under a Markov kernel, we have
  \begin{align*}
    \left\| \int \Phi(\delta_2) dQ_{1,K} - \int \Phi(\delta_2) dQ_{1,\infty} \right\|_\text{TV} 
    \leq \| Q_{1,K} - Q_{1,\infty} \|_\text{TV}
  \end{align*}
  where $\Phi(d)$ is the Gaussian kernel $h \mapsto N(\rho J_2(d)h, \rho J_2(d)) $.
  Thus, the posterior predictive distributions of $A_2$ converge in total variation as well.

  Finally, we show that the centered expected loss functions converge.
  The posterior means $\hat h_{2,\infty,K}$ and $\hat h_{2,\infty,\infty}$ are martingales
  under their predictive laws, so
  \begin{align*}
    L_\infty^K(\delta_K) - V(\hat b_{1,K}) = \E_K[D_V(\hat b_{1,K}, \hat b_{2,K}(\delta_K)) \mid A_1, U]
    \\
    L_\infty(\delta_K) - V(\hat b_{1,\infty}) = \E_\infty[D_V(\hat b_{1,\infty} , \hat b_{2,\infty}(\delta_K)) \mid A_1, U]
    .
  \end{align*}
  By Lemma \ref{lemma:qplipschitz},
  \begin{align*}
    &| D_V(\hat b_{1,K}, \hat b_{2,K}(\delta_K)) -D_V(\hat b_{1,\infty}, \hat b_{2,\infty}(\delta_K)) | \\
    &\leq 2 \kappa^{-1} ( \| \hat b_{2,K}(\delta_K) - \hat b_{1,K} \| + \| \hat b_{2,\infty}(\delta_K) - \hat b_{1,\infty} \| )
    \\ &\quad \times ( \| \hat b_{1,K} - \hat b_{1,\infty} \| + \| \hat b_{2,K}(\delta_K) - \hat b_{2,\infty}(\delta_K) \| )
  \end{align*}
  Consequently, 
  \begin{align*}
    \E_K[| D_V(\hat b_{1,K}, \hat b_{2,K}(\delta_K)) -D_V(\hat b_{1,\infty} , \hat b_{2,\infty}(\delta_K)) | \mid A_1, U] \to 0
  \end{align*}
  in $L_1$ by Lemma \ref{lemma:limitmean}.

  It remains to show that 
  \begin{align*}
    \E_K[ D_V(\hat b_{1,\infty} , \hat b_{2,\infty}(\delta_K)) \mid A_1, U]
    - \E_\infty[ D_V(\hat b_{1,\infty} , \hat b_{2,\infty}(\delta_K)) \mid A_1, U]
  \end{align*}
  converges to $0$ in $L_1$ under the $K$-truncated prior.
  By concavity of $V$ and Lemma \ref{lemma:qplipschitz},
  \begin{align*}
    -\frac{1}{2\kappa} \|b_2 - b_1\|^2 \leq D_V(b_1, b_2) \leq 0
  \end{align*}
  which establishes uniform second moment bounds on $D_V$ through fourth moment bounds on the Gaussian posterior increment $\hat b_{2,\infty}(\delta_K) - \hat b_{1,\infty}$.
  The inequality
  \begin{align*}
    \left| \int D_V \Phi(\delta_2) Q_{1,K}(dh) - \int D_V \Phi(\delta_2) Q_{1,\infty}(dh) \right| \\
    \leq \left[ 
      2 \| \Phi(\delta_2) Q_{1,K} - \Phi(\delta_2) Q_{1,\infty} \|_\text{TV} 
      \left(\int D_V^2 \Phi(\delta_2) Q_{1,K} + \int D_V^2 \Phi(\delta_2) Q_{1,\infty} \right)
    \right]^{1/2}
  \end{align*}
  combined with the total variation convergence established above implies that
  the difference in conditional expectations converges to zero in $L_1$ under the $K$-truncated prior.
\end{proof}

\section{Verifying stochastic equicontinuity}
\label{appendix:stocheq}

We say a sequence of processes $A_n(\cdot)$ indexed by a compact Euclidean set $\Delta$
is stochastically equicontinuous under $P$ if for every $\varepsilon > 0$
and $\eta > 0$ there exists an $\iota > 0$ such that
\begin{align}
  \label{eq:stocheq}
  \limsup_{n\to\infty} P\left( \sup_{d, d' \in \Delta: \|d - d'\| < \iota} \|A_n(d) - A_n(d')\| > \eta \right) < \varepsilon.
\end{align}
Here we suppose the supremum is measurable.

For a scalar-valued class of functions $\mathcal{F} = \{ f(\cdot, d) : d \in \Delta \}$
where $f(\cdot, d)$ is a function of the data $w_i$ and the treatment assignment $d$,
define the centered empirical process
\begin{align*}
  \mathbb{G}_n f(d) = \frac{1}{\sqrt{n}} \sum_{i \in \mathcal{I}_{2,n}} (f(w_i, d) - \E[f(w_i, d)]).
\end{align*}
Note that we normalize by $\sqrt{n}$ but sum over only the second wave to match the processes in the main text.

Our primitive conditions for stochastic equicontinuity are based on the Type IV condition of \textcite{andrewsChapter37Empirical1994}.
This class is defined as follows.
Let $p \geq 2$.
Suppose there are constants $C_1,C_2,r_0>0$ such that
\begin{align}
  \sup_n\sup_{i\in\mathcal I_{2,n}}
  \left(
    E\sup_{\substack{\tilde d\in\Delta\\
                        \lVert d-\tilde d\rVert\leq r}}
    |f(w_i,d)-f(w_i,\tilde d)|^p
  \right)^{1/p}
  \leq C_1r^{C_2}
  \label{eq:lpsmoothness}
\end{align}
for every $d\in\Delta$ and $0<r\leq r_0$.
Then we say that $\mathcal{F}$ satisfies the Type IV condition with $p$.

We also will impose the following envelope condition.
For a class of functions $\mathcal{F}$, the relevant envelope condition is
that there exists some $F$ such that
\begin{align*}
  |f(\cdot)| \leq F(\cdot), \quad \forall f \in \mathcal{F}
\end{align*}
and 
\begin{align}
  \label{eq:envelope}
  \limsup_{n\to\infty} \frac{1}{n} \sum_{i \in \mathcal{I}_{2,n}} \E |F(w_i)|^{2+\kappa} < \infty
\end{align}
for some $\kappa > 0$.

Our primitive conditions are on the bracketing entropy of the function classes.
For any $\varepsilon > 0$ and $p \in [2, \infty)$,
the bracketing number $N_{[]}(\varepsilon, \mathcal{F}, L_p(P))$ is the smallest number $N$
such that there exist functions $\{a_j, b_j\}_{j=1}^N$ such that
for every $d$, some $j$ satisfies
\begin{align*}
  | f(\cdot, d) - a_j(\cdot) | \leq b_j(\cdot)
\end{align*}
where
\begin{align*}
  \max_{j} \sup_n \sup_{i \in \mathcal{I}_{2,n}} (\E[b_j(w_i)^p])^{1/p} \leq \varepsilon.
\end{align*}
The bracketing entropy is
\begin{align*}
  \mathcal{J}_{[]}(\mathcal{F}, L_p(P)) = \int_0^1 \sqrt{\log N_{[]}(\varepsilon, \mathcal{F}, L_p(P))} d\varepsilon.
\end{align*}

\subsection{Primitive conditions for stochastic equicontinuity}

Let $w_i = (\epsilon_i, \nu_i, x_i, n)$.
Note that $x_i$ is fixed in our setup, and $n$ is included to allow the log-likelihood to depend on
the local parameter $\theta_n(h) = \theta_0 + h/\sqrt{n}$.
Define the following log-likelihood quantities:
\begin{align*}
  \ell_h(w_i, d) &= \sqrt{n}\log \frac{p_{y \mid z, x}(y_i(d) \mid z_i(d), x_i; \theta_n(h))}{p_{y \mid z, x}(y_i(d) \mid z_i(d), x_i; \theta_0)}
  \\
  b_{n,h}(d) &= \frac{1}{\sqrt{n}} \sum_{i \in \mathcal{I}_{2,n}} \E_0[\ell_h(w_i, d)]
  \\
  \Lambda_{2,n}^h(d) &= \frac{1}{\sqrt{n}} \sum_{i \in \mathcal{I}_{2,n}} \ell_h(w_i, d)
  = \mathbb{G}_n \ell_h(d) + b_{n,h}(d).
\end{align*}
In what follows, we refer to the following function classes:
\begin{align*}
  \mathcal F_\psi^k
  =\{\psi_k(\cdot, d) : d\in\Delta\},
    \qquad k=1,\ldots,\ell,
  &&
  \mathcal F_\ell^h
  =\{\ell_h(\cdot, d) : d\in\Delta\}.
\end{align*}

\begin{assumption}[Primitive conditions for equicontinuity]
  \label{assumption:primitive-stocheq}
  \begin{enumerate}
    \item The class $\mathcal{F}_\psi^k$ satisfies
      the Type IV condition of \textcite{andrewsChapter37Empirical1994} with $p=2$
      for each $k$.
      It has an envelope satisfying (\ref{eq:envelope}) for some $\kappa > 0$.
    \item For every fixed $h$, the class $\mathcal{F}_\ell^h$ satisfies
      the Type IV condition of \textcite{andrewsChapter37Empirical1994} with $p=2$.
      It has an envelope satisfying (\ref{eq:envelope}) for some $\kappa > 0$.
    \item For every fixed $h$,
      \begin{align*}
        \sup_{d\in\Delta} \left| b_{n,h}(d) + \frac{\rho}{2} h' J_2(d) h \right| \to 0
      \end{align*}
      where $J_2$ is continuous.
  \end{enumerate}
\end{assumption}

\begin{proposition}
  \label{proposition:primitive-stocheq}
  Suppose Assumptions \ref{assumption:largewaves},
  \ref{assumption:dqm},
  \ref{assumption:potentialoutcomes},
  \ref{assumption:information},
  and \ref{assumption:primitive-stocheq} hold.
  Then $\mathcal{J}_{[]}(\mathcal{F}_\psi^k, L_2(P_0))$ is finite for every $k$,
  and $\mathcal{J}_{[]}(\mathcal{F}_\ell^h, L_2(P_0)) < \infty$ for every $h$.
  Consequently, $A_{2,n}(\cdot)$ and $\Lambda_{2,n}^h(\cdot)$ are stochastically equicontinuous.
\end{proposition}

\begin{proof}
  Theorem 5 of \textcite{andrewsChapter37Empirical1994} implies that 
  $\mathcal{J}_{[]}(\mathcal{F}_\psi^k, L_2(P_0)) < \infty$ for every $k$.
  Combined with the envelope conditions, Theorem 4 of \textcite{andrewsChapter37Empirical1994} implies that
  $\mathbb{G}_n \psi_k(\cdot)$ is stochastically equicontinuous for each $k$
  with respect to the metric
  \begin{align*}
    \varrho(d, \tilde d) = \sup_n \sup_{i \in \mathcal{I}_{2,n}} (\E_0[|\psi_k(w_i, d) - \psi_k(w_i, \tilde d)|^2])^{1/2}.
  \end{align*}
  Since this is dominated by the Euclidean metric in the sense that 
  $\varrho(d, \tilde d) \leq C_1 \|d - \tilde d\|^{C_2}$ for some constants $C_1$  and $C_2$, 
  $\mathbb{G}_n \psi_k(\cdot)$ is stochastically equicontinuous with respect to the Euclidean metric as well.
  Since the score is mean zero by DQM,
  $A_{2,n,k}(\cdot) = \mathbb{G}_n \psi_k(\cdot)$ and we conclude that the vector $A_{2,n}(\cdot)$ is stochastically equicontinuous.

  The same argument applies to the log-likelihood to yield stochastic equicontinuity of
  $\mathbb{G}_n \ell_h(\cdot)$.
  By assumption, the centered drift term $b_{n,h}(\cdot)$ converges uniformly to a continuous function of $d$,
  so $\Lambda_{2,n}^h(\cdot) = \mathbb{G}_n \ell_h(\cdot) + b_{n,h}(\cdot)$ is stochastically equicontinuous as well.
\end{proof}

\subsection{Verifying primitive conditions for the Progresa application}

We first establish some regularity of the model described in Section \ref{section:progresa}.
Let $z_s$ denote the subsidy in grade $s$.
Let $V_{\theta,\tau}(s,f;z)$ denote the ex ante value at age $\tau$.
In the terminal period,
\begin{align*}
  V_{\theta,18}(s,f;z) = \gamma_9\1[s\geq 10].
\end{align*}
Choice-specific values are defined recursively by
\begin{align*}
  V_{\theta,\tau}(s,f;z)
  =
  v^0_{\theta,\tau}(s,f;z)
  +
  \log\{1+\exp(v_{\theta,\tau}(s,f;z))\}
\end{align*}
where
\begin{align*}
  v^0_{\theta,\tau}(s,f;z) &= \eta V_{\theta,\tau+1}(s,f;z),\\
  v^1_{\theta,\tau}(s,f;z) &= \gamma_0+\gamma_1 \tau+\gamma_2 s+\gamma_3 f + \gamma_4\1[s\geq 7]\\
  &\quad+\gamma_5(1-f)z_s+\gamma_6(1-f)\tau z_s+\gamma_7fz_s+\gamma_8f\tau z_s\\
  &\quad+\eta\big[
    p_{\mathrm{pass}}(\tau,s) V_{\theta,\tau+1}(s+1,f;z)
    + \{1-p_{\mathrm{pass}}(\tau,s)\} V_{\theta,\tau+1}(s,f;z)
  \big],\\
  v_{\theta,\tau}(s,f;z) &=v^1_{\theta,\tau}(s,f;z)-v^0_{\theta,\tau}(s,f;z).
\end{align*}

In what follows, let $x$ contain $(\tau,s,f)$ and define
\begin{align*}
  v^y_\theta(z,x)
  &=
  v^y_{\theta,\tau}(s,f;z),\\
  v_\theta(z,x)
  &=
  v^1_\theta(z,x)-v^0_\theta(z,x),\\
  \varsigma_\theta(z,x)
  &=
  p_{y\mid z,x}(1\mid z,x;\theta)
  =
  \frac{\exp\{v_\theta(z,x)\}}
       {1+\exp\{v_\theta(z,x)\}},\\
\end{align*}

\begin{proposition}
  \label{proposition:modelregularity}
  Suppose $\theta_0$ is interior to a compact, convex $\bar\Theta \subseteq \Theta$.
  Suppose the covariate space, treatment space, and household state space are contained in compact and convex sets.
  Define
  \begin{align*}
    M_\theta = \sup_{\theta, z, x} \|\nabla_\theta v_\theta(z,x) \|, \qquad
    M_z = \sup_{\theta, z, x} \|\nabla_z v_\theta(z,x) \|, \\
    M_{\theta z} = \sup_{\theta, z, x} \|\nabla^2_{z,\theta} v_\theta(z,x)\|, \qquad
    M_{\theta\theta} = \sup_{\theta,z,x} \| \nabla^2_{\theta,\theta} v_\theta(z,x) \|, \\
    L_\psi = M_{\theta z} + \frac{1}{4}M_\theta M_z
  \end{align*}
  where $\theta$ ranges over $\bar \Theta$.
  Then these constants are finite and the following hold uniformly over $\theta \in \bar\Theta$:
  \begin{enumerate}
    \item Uniformly over $y$, $z$, and $x$, we have $\|\psi_\theta(y, z, x)\| \leq M_\theta$ and 
      $\psi_\theta(y, z, x)$ is Lipschitz in $y$ and $z$ with Lipschitz constants $M_\theta$ and $L_\psi$.
    \item $\varsigma_{\theta_0}$ is Lipschitz in treatment with constant $\frac{1}{4}M_z$
    \item For every $\ell_h \in \mathcal{F}_\ell^h$,
      $|\ell_h(\cdot, d) | \leq \|h\| M_\theta$ and $\ell_h$ is Lipschitz in $y$ and $z$ 
      with Lipschitz constants $\|h\| M_\theta$ and $\|h\| L_\psi$.
    \item For every fixed $h$, the log-likelihood satisfies the expansion
      \begin{align*}
        &\sup_{z,x} \left|
          \E_0 \left[ \log \frac{p_{y\mid z,x}(y \mid z, x; \theta_0 + h/\sqrt{n})}{p_{y\mid z,x}(y \mid z, x; \theta_0)} \mid z, x \right]
          \right.
          \\ &\quad \left. + \frac{1}{2n} h' \E_0\left[ \psi_{\theta_0}(y, z, x) \psi_{\theta_0}(y, z, x)' \mid z, x \right] h
          \right| = o(n^{-1}).
      \end{align*}
  \end{enumerate}
\end{proposition}

\begin{proof}
  First, the terminal value and flow payoffs described in Section \ref{section:progresa}
  are twice continuously differentiable in $\theta$ and $z$.
  Finite-horizon backward induction preserves this property,
  so $v_\theta(z,x)$ is twice continuously differentiable in $\theta$ and $z$.
  Compactness of the parameter, treatment, and covariate space implies that the
  derivatives of $v_\theta(z,x)$ are uniformly bounded, so that $M_\theta$, $M_z$, $M_{\theta z}$, and $M_{\theta,\theta}$ are finite.

  By direct calculation,
  \begin{align*}
    \psi_\theta(y, z, x) = \nabla_\theta \log p_{y\mid z,x}(y \mid z, x; \theta) = (y - \varsigma_\theta(z,x)) \nabla_\theta v_\theta(z,x)
  \end{align*}
  and so the score is bounded uniformly by $M_\theta$.
  Its gradient with respect to $z$ satisfies
  \begin{align*}
    \|\nabla_z \psi_\theta(y, z, x)\| 
    &\leq \|\nabla_z \varsigma_\theta(z,x)\| \|\nabla_\theta v_\theta(z,x)\| 
    + |y - \varsigma_\theta(z,x)| \|\nabla^2_{z,\theta} v_\theta(z,x)\|\\
    &\leq \frac{1}{4} M_\theta M_z + M_{\theta z} = L_\psi
  \end{align*}
  and the change resulting from a change in $y$ is
  \begin{align*}
    \|\psi_\theta(\tilde y, z, x) - \psi_\theta(y, z, x)\| = \|(\tilde y - y) \nabla_\theta v_\theta(z,x)\|
    \leq M_\theta |\tilde y - y |.
  \end{align*}
  The gradient of $\varsigma_\theta$ with respect to $z$ is
  \begin{align*}
    \nabla_z \varsigma_\theta(z,x) = \varsigma_\theta(z,x) (1 - \varsigma_\theta(z,x)) \nabla_z v_\theta(z,x)
  \end{align*}
  which is bounded by $\frac{1}{4}M_z$.
  For any $\ell_h \in \mathcal{F}_\ell^h$, the fundamental theorem of calculus implies that
  \begin{align*}
    \ell_h(w_i, d) = \int_0^1 h' \psi_{\theta_0 + t h/\sqrt{n}}(y_i(d), z_i(d), x_i) dt
  \end{align*}
  and so the envelope and increment bounds for $\psi$ imply the same bounds for $\ell_h$ with constants multiplied by $\|h\|$.

  We now verify the log-likelihood expansion.
  First, we have
  \begin{align*}
    \nabla_\theta \psi_\theta(y, z, x) 
    = (y - \varsigma_\theta(z,x)) \nabla^2_{\theta,\theta} v_\theta(z,x) 
    - \varsigma_\theta(z,x) (1 - \varsigma_\theta(z,x))\nabla_\theta v_\theta(z,x) \nabla_\theta v_\theta(z,x)' 
  \end{align*}
  and therefore $\sup_{\theta, y, z, x} \|\nabla_\theta \psi_\theta(y, z, x)\| \leq M_{\theta\theta} + \frac{1}{4} M_\theta^2$
  and $\nabla_\theta \psi_\theta$ is integrable.
  Also, continuity of $v_\theta$ and compactness ensure that $\varsigma_\theta(z,x)$ is bounded away from $0$ and $1$,
  so the log-likelihood ratio is integrable as well.
  Direct calculation yields the mean-zero property of the score
  \begin{align*}
    \E_0[\psi_{\theta_0}(y, z, x) \mid z, x] 
    &= \E_0[y - \varsigma_{\theta_0}(z,x) \mid z, x] \nabla_\theta v_{\theta_0}(z,x)
    = 0
  \end{align*}
  and the information identity
  \begin{align*}
    \E_0[\nabla_\theta \psi_{\theta_0}(y, z, x) \mid z, x] 
    &= - \varsigma_{\theta_0}(z,x) (1 - \varsigma_{\theta_0}(z,x))\nabla_\theta v_{\theta_0}(z,x) \nabla_\theta v_{\theta_0}(z,x)' \\
    &= - \E_0[\psi_{\theta_0}(y, z, x) \psi_{\theta_0}(y, z, x)' \mid z, x].
  \end{align*}
  A Taylor expansion of the log-likelihood ratio then gives
  \begin{align*}
    &\log p_{y\mid z,x}(y \mid z, x; \theta_0 + h/\sqrt{n})
    - \log p_{y\mid z,x}(y \mid z, x; \theta_0)\\
    &= \frac{1}{\sqrt{n}} h' \psi_{\theta_0}(y, z, x)
    + \frac{1}{2n} h' \nabla_{\theta} \psi_{\theta_0}(y, z, x) h
    + o(n^{-1}).
  \end{align*}
  Since $\theta, z, x$ are all in compact sets and $y$ is binary, 
  $\nabla_\theta \psi_\theta(y, z, x)$ is uniformly continuous in $\theta$,
  so the Taylor remainder above is uniformly $o(n^{-1})$ in $z$ and $x$.
  Taking expectations conditional on $z$ and $x$ and using the information identity and the zero mean property of the score gives
  \begin{align*}
    \sup_{z,x} \bigg|
    &\E_0\left[
      \log \frac{p_{y\mid z,x}(y \mid z, x; \theta_0 + h/\sqrt{n})}{p_{y\mid z,x}(y \mid z, x; \theta_0)}
      \mid z, x
    \right]
    \\ &\quad
    +\frac{1}{2n} h' \E_0[\psi_{\theta_0}(y, z, x) \psi_{\theta_0}(y, z, x)' \mid z, x] h
    \bigg|
    = o(n^{-1})
  \end{align*}
  as desired.
\end{proof}

We now verify that our model satisfies Assumption \ref{assumption:primitive-stocheq}.

\begin{proposition}
  Maintain the assumptions of Proposition \ref{proposition:modelregularity}.
  Suppose the treatment is $J$-dimensional and is assigned by
  \begin{align*}
    z_i(d) = \Phi_0 ' x_i + \Phi_1' x_i \1[\nu_i \leq G(\phi_2' x_i)], \quad d = (\text{vec}(\Phi_0), \text{vec}(\Phi_1), \phi_2) \in \Delta
  \end{align*}
  where $G: \R \mapsto [0,1]$ is Lipschitz with constant $K$
  and $\nu_i \sim U[0,1]$ conditional on $x_i$.
  Suppose $\epsilon_i$ and $\nu_i$ are independent conditional on $x_i$ and $\Delta$ is compact.

  Then for every finite $p \geq 2$ and every fixed $h$,
  $\mathcal{F}_\psi^k$ and $\mathcal{F}_\ell^h$ satisfy the Type IV condition (\ref{eq:lpsmoothness})
  of \textcite{andrewsChapter37Empirical1994} for every $k = 1, \dots, \ell$.
  Both classes also have uniformly bounded envelopes
  satisfying (\ref{eq:envelope}) for some $\kappa > 0$.

  If, in addition, Assumptions \ref{assumption:largewaves},
  \ref{assumption:dqm}, \ref{assumption:potentialoutcomes}, and \ref{assumption:information} hold,
  then for every fixed $h$,
  Assumption \ref{assumption:primitive-stocheq}(3) holds
  and consequently $A_{2,n}(\cdot)$ and $\Lambda_{2,n}^h(\cdot)$ are stochastically equicontinuous.
\end{proposition}

\begin{proof}
  Let $I_i(d) = \1[\nu_i \leq G(\phi_2' x_i)]$.
  If $\|\tilde d - d\| \leq r$, then 
  \begin{align*}
    \|z_i(\tilde d) - z_i(d)\|
    &\leq \| \tilde \Phi_0' x_i - \Phi_0' x_i \| + \| (\tilde \Phi_1' - \Phi_1') x_i \| + \| \Phi_1' x_i \| \| I_i(\tilde d) - I_i(d) \|. 
  \end{align*}
  The first two terms are each bounded by $M_x r$ where $M_x = \sup_i \|x_i\|$.
  $\|\Phi_1' x_i\|$ is bounded by a constant multiple of $M_x$ as well by compactness of $\Delta$.

  To bound the last term, note that Lipschitz continuity of $G$ implies that
  $| G(\tilde \phi_2' x_i) - G(\phi_2' x_i) | \leq K M_x r$ so that
  $I_i(\tilde d) \neq I_i(d)$ only if $\nu_i$ is within $K M_x r$ of $G(\phi_2' x_i)$.
  Since $\nu_i$ is uniform, we have
  \begin{align*}
    P_0\left(
      \sup_{\tilde d: \|\tilde d - d\| \leq r}
      |I_i(\tilde d) - I_i(d)| > 0
      \mid x_i 
    \right)
    &\leq 2 K M_x r.
  \end{align*}
  Since the choice probability $\varsigma_{\theta_0}$ is Lipschitz in $z$,
  on the event 
  \begin{align*}
    S_i(d,r) = \{ \sup_{\| \tilde d - d\| \leq r} |I_i(\tilde d) - I_i(d)| = 0 \},
  \end{align*}
  we have
  $\sup_{\tilde d: \|\tilde d - d\| \leq r} \| z_i(\tilde d) - z_i(d) \| \leq C r$ for some constant $C$.
  Since the logistic distribution has a bounded density, this implies that there exists a $C$ such that
  \begin{align*}
    P_0\left(
      \sup_{\tilde d: \|\tilde d - d\| \leq r} |y_i(\tilde d) - y_i(d)| > 0
      \cap S_i(d,r)
      \mid x_i
    \right) \leq C r.
  \end{align*}
  Since $P_0(S_i(d,r)^c \mid x_i) \leq 2 K M_x r$, we have
  \begin{align*}
    P_0\left(
      \sup_{\tilde d: \|\tilde d - d\| \leq r}
      | I_i(\tilde d) - I_i(d) | + | y_i(\tilde d) - y_i(d) | > 0
      \mid x_i
    \right)
    \leq C r
  \end{align*}
  for some constant $C$.

  We now use this fact to verify the Type IV condition for $\mathcal{F}_\psi^k$.
  We have
  \begin{align*}
    \psi(y_i(\tilde d), z_i(\tilde d), x_i) - \psi(y_i(d), z_i(d), x_i)
    &= [\psi(y_i(\tilde d), z_i(\tilde d), x_i) - \psi(y_i(d), z_i(\tilde d), x_i)]\\
    &+ [\psi(y_i(d), z_i(\tilde d), x_i) - \psi(y_i(d), z_i(d), x_i)].
  \end{align*}
  The first term is zero unless the outcome switches and is otherwise uniformly bounded by Proposition \ref{proposition:modelregularity}.
  The second is Lipschitz in $z_i$ by Proposition \ref{proposition:modelregularity}.
  Thus,
  \begin{align*}
    \E_0 \sup_{\tilde d: \|\tilde d - d\| \leq r} \|\psi(y_i(\tilde d), z_i(\tilde d), x_i) - \psi(y_i(d), z_i(d), x_i)\|^p
    \leq C_p r^p + C_p r \leq C_p r
  \end{align*}
  and \ref{eq:lpsmoothness} holds for $\mathcal{F}_\psi^k$.

  To verify the Type IV condition for $\mathcal{F}_\ell^h$, we apply the same decomposition to the log-likelihood.
  For any $\ell_h \in \mathcal{F}_\ell^h$, Proposition \ref{proposition:modelregularity} implies that
  \begin{align*}
    |\ell_h(w_i, \tilde d) - \ell_h(w_i, d)| 
    \leq \| h \| M_\theta |y_i(\tilde d) - y_i(d)| + \| h \| L_\psi \| z_i(\tilde d) - z_i(d) \|.
  \end{align*}
  The first term is zero unless the outcome switches and is otherwise uniformly bounded since $y_i$ is binary.
  On the event that the assignment indicator does not switch,
  \begin{align*}
    \sup_{\tilde d: \|\tilde d - d\| \leq r} \| z_i(\tilde d) - z_i(d) \| \leq C r
  \end{align*}
  and is otherwise uniformly bounded since $z_i$ is bounded.
  Thus,
  \begin{align*}
    \E_0 \sup_{\tilde d: \|\tilde d - d\| \leq r} |\ell_h(w_i, \tilde d) - \ell_h(w_i, d)|^p
    \leq C_p r^p + C_p r \leq C_p r
  \end{align*}
  and \ref{eq:lpsmoothness} holds for $\mathcal{F}_\ell^h$.

  We now verify convergence of the drift term.
  Proposition \ref{proposition:modelregularity} implies that
  \begin{align*}
    &\sup_{z,x} \left|
    \E_0 \left[ \log \frac{p_{y\mid z,x}(y \mid z, x; \theta_0 + h/\sqrt{n})}{p_{y\mid z,x}(y \mid z, x; \theta_0)} \mid z, x \right] \right.
    \\&\quad\left.+ \frac{1}{2n} h' \E_0\left[ \psi_{\theta_0}(y, z, x) \psi_{\theta_0}(y, z, x)' \mid z, x \right] h
      \right| = o(n^{-1}).
  \end{align*}
  Evaluate at $(z,x) = (z_i(d), x_i)$ and sum this across $i \in \mathcal{I}_{2,n}$ to obtain
  \begin{align*}
    \sup_{d\in\Delta} \left| b_{n,h}(d) + \frac{1}{2} h' \frac{n_2}{n} J_{2,n}(d) h \right| = o(1)
  \end{align*}
  where
  \begin{align*}
    J_{2,n}(d) = \frac{1}{n_2} \sum_{i \in \mathcal{I}_{2,n}} \E_0\left[ \psi_{\theta_0}(y, z_i(d), x_i) \psi_{\theta_0}(y, z_i(d), x_i)' \right].
  \end{align*}
  The Type-IV condition for $\mathcal{F}_\psi^k$ and the uniform bound on the score
  imply that $J_{2,n}(d)$ is uniformly equicontinuous.
  Assumptions \ref{assumption:largewaves} and \ref{assumption:information} imply that the pointwise limit is $J_2(d)$.
  Since $\Delta$ is compact, we conclude that this convergence is uniform in $d$.
  Hence
  \begin{align*}
    \sup_{d\in\Delta} \left| b_{n,h}(d) + \frac{\rho}{2} h' J_2(d) h \right| = o(1)
  \end{align*}
  and Assumption \ref{assumption:primitive-stocheq}(3) holds.
  The envelope conditions are established in Proposition \ref{proposition:modelregularity}, so the result follows
  from Proposition \ref{proposition:primitive-stocheq}.
\end{proof}

\section{Details of the Progresa application}
\label{appendix:progresa}
Progresa was a conditional cash transfer program administered by the Mexican
government beginning in 1998.
The program gave cash to households for every child enrolled in school.
The size of the transfer depended on children's
gender and grade of enrollment, and the magnitude was economically significant---
the ninth grade subsidy constituted about $40\%$ of an adult male's wage
and about $66\%$ of a child's wage (\cite{schultzSchoolSubsidiesPoor2004}).
The government conducted an experimental evaluation of the program before a
larger rollout in 2000.
In the original experiment,
$62\%$ of the experimental sample was treated.
After observing the results of the experiment,
the government enacted a policy that gave the exact same subsidy
that was used in the experiment to a wider set of villages.

\begin{table}[htbp]
  \centering
  \begin{threeparttable}
    \caption{Descriptive statistics for first-wave Progresa sample}
    \label{table:summary}
    \footnotesize
    \begingroup
\setlength{\tabcolsep}{4pt}
\begin{tabular}{@{}l S[table-format=2.3] @{\hspace{8pt}} S[table-format=2.3] S[table-format=2.3] @{\hspace{10pt}} S[table-format=2.3] S[table-format=2.3]@{}}
\toprule
    & \multicolumn{1}{c}{All} & \multicolumn{2}{c}{Boys} & \multicolumn{2}{c}{Girls} \\
\cmidrule(lr){3-4}\cmidrule(l){5-6}
    & & \multicolumn{1}{c}{Educ $<6$} & \multicolumn{1}{c}{Educ $\geq 6$} & \multicolumn{1}{c}{Educ $<6$} & \multicolumn{1}{c}{Educ $\geq 6$} \\
\midrule
    Observations & \multicolumn{1}{c}{\num{1000}} & \multicolumn{1}{c}{\num{280}} & \multicolumn{1}{c}{\num{235}} & \multicolumn{1}{c}{\num{276}} & \multicolumn{1}{c}{\num{209}} \\
\addlinespace
    Age & 12.56 & 11.70 & 13.78 & 11.38 & 13.89 \\
     & \multicolumn{1}{c}{(1.89)} & \multicolumn{1}{c}{(1.73)} & \multicolumn{1}{c}{(1.33)} & \multicolumn{1}{c}{(1.50)} & \multicolumn{1}{c}{(1.39)} \\
    Educ & 4.99 & 3.78 & 6.51 & 3.84 & 6.43 \\
     & \multicolumn{1}{c}{(1.64)} & \multicolumn{1}{c}{(1.13)} & \multicolumn{1}{c}{(0.75)} & \multicolumn{1}{c}{(1.11)} & \multicolumn{1}{c}{(0.72)} \\
\addlinespace
    Enrolled & 0.806 & 0.889 & 0.757 & 0.899 & 0.627 \\
\addlinespace
    Treatment assignment & 0.627 & 0.586 & 0.660 & 0.623 & 0.651 \\
    Beneficiary & 0.616 & 0.579 & 0.647 & 0.612 & 0.636 \\
    Beneficiary $\mid$ assigned & 0.982 & 0.988 & 0.981 & 0.983 & 0.978 \\
\bottomrule
\end{tabular}
\endgroup

    \begin{tablenotes}[flushleft]
      \footnotesize
      \item \textit{Notes:} the first wave sample consists of $1000$ observations drawn uniformly from the
      $12,031$ observations in the analysis sample.
      All variables report the sample mean;
      continuous variables report sample standard deviations in parentheses.
      Treatment indicates initial experimental assignment, and beneficiary
      indicates realized treatment status.
    \end{tablenotes}
  \end{threeparttable}
\end{table}

We construct the analysis sample from the Progresa raw data files contained
in the replication package for \textcite{adhvaryuHelpingChildrenCatch2023}\footnote{
  For many years, the original Progresa data was made publicly available by the Mexican government,
  but the hosting website appears to have been taken down.
}.
We combine the October 1998 and March 1999 surveys and define enrollment
in the 1998--99 school year as reported attendance in either wave.
We merge the 1997 survey data to recover sex and other household characteristics.
We retain children ages 10--17 in households classified as poor
and with fewer than nine years of completed schooling.
These restrictions yield 12,031 children, from which we draw 1,000 children uniformly
as the first-wave estimation sample.
Table \ref{table:summary} reports descriptive statistics for this sample.

\begin{table}[htbp]
  \centering
  \begin{threeparttable}
    \caption{Estimates of $\theta$ at each vertex and unrestricted fit}
    \label{table:estimates}
    \footnotesize
    \begin{tabular}{lrrrrr}
\toprule
Parameter & Unrestricted (SE) & Boys Pri. & Boys Sec. & Girls Pri. & Girls Sec. \\
\midrule
Constant & 3.134 (0.367) & 3.157 & 3.008 & 3.187 & 3.028 \\
Female & -0.384 (0.082) & -0.337 & -0.315 & -0.401 & -0.407 \\
Enrollment Grade & 0.88 (0.051) & 0.838 & 0.831 & 0.839 & 0.829 \\
Age & -0.728 (0.029) & -0.704 & -0.687 & -0.704 & -0.683 \\
Secondary School & -3.536 (0.213) & -3.336 & -3.343 & -3.331 & -3.338 \\
Subsidy Boys & 5.108 (2.59) & 6.652 & 8.852 & 2.926 & 6.261 \\
Subsidy Girls & 2.666 (2.245) & 2.245 & 4.607 & 5.314 & 7.648 \\
Subsidy Age Boys & -0.359 (0.264) & -0.649 & -0.833 & -0.281 & -0.6 \\
Subsidy Age Girls & -0.096 (0.238) & -0.219 & -0.45 & -0.529 & -0.735 \\
Terminal Completion & 2.358 (0.417) & 2.476 & 2.48 & 2.451 & 2.496 \\
\midrule
$\lambda_{\min}(H)$ & -- & 0.0025 & 0.0105 & 0.0051 & 0.0245 \\
\midrule
LR Statistic & -- & 22.22 & 17.13 & 22.85 & 16.68 \\
LR P-Value & -- & $<0.0001$ & 0.0007 & $<0.0001$ & 0.0008 \\
\bottomrule
\end{tabular}

    \begin{tablenotes}[flushleft]
      \footnotesize
      \item \textit{Notes:} All estimates are obtained from the holdout sample of $8,031$ observations.
        Standard errors for the unrestricted model are reported in parentheses.
        $\lambda_{\min}(H)$ is the minimum eigenvalue of the negative Hessian of welfare at each vertex;
        positive values indicate strong concavity of welfare at the vertex.
        The likelihood ratio tests the weak complementarity constraints used to define each vertex.
        Subsidy amounts are in thousands of pesos, and age enters the model centered at 5.
      \end{tablenotes}
  \end{threeparttable}
\end{table}

\begin{figure}[htbp]
    \centering
    \caption{Enrollment-rate model fit}
    \label{figure:enrollment}
    \includegraphics[width=0.9\textwidth]
    {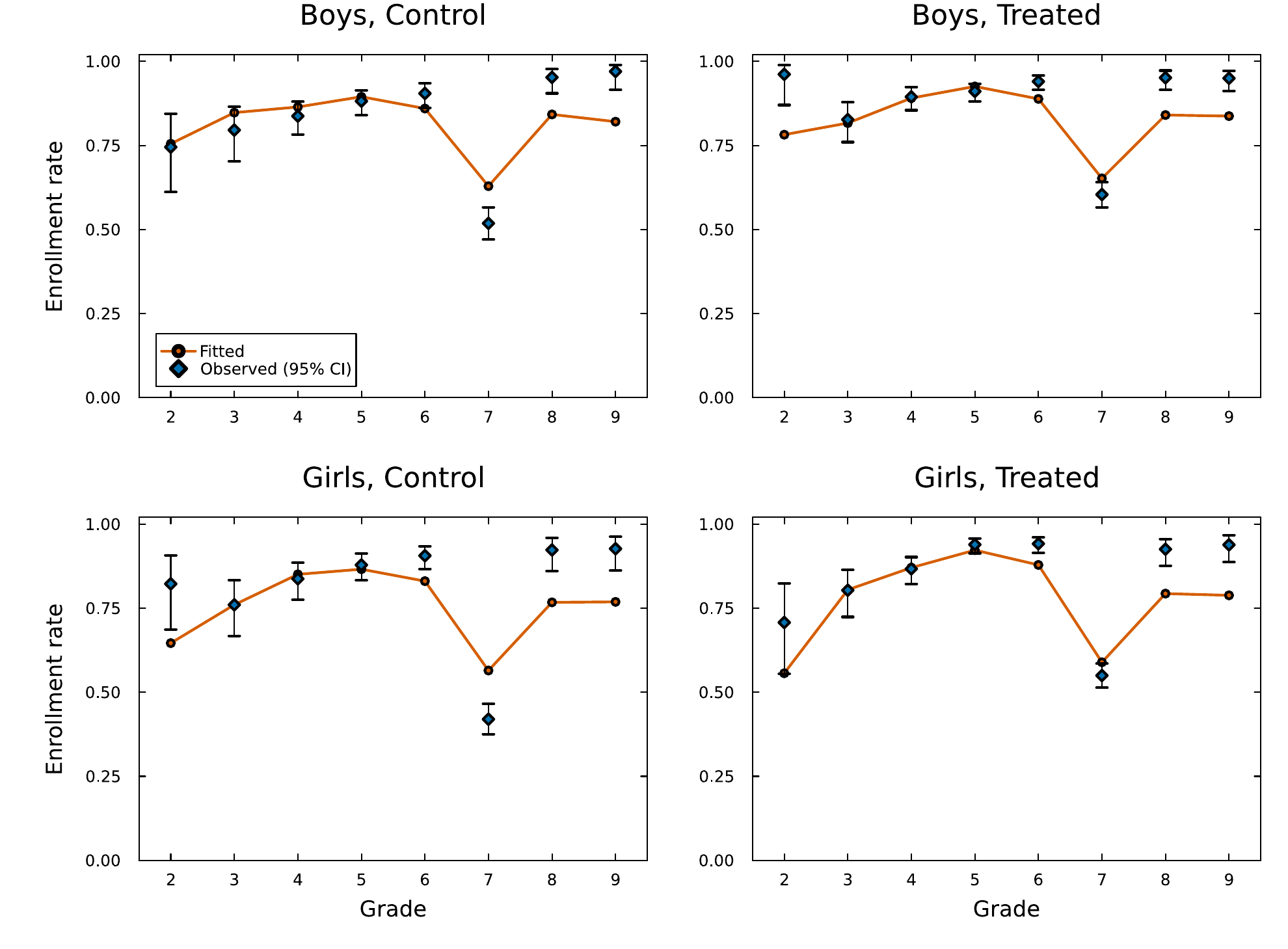}
    \par\smallskip
    \begin{minipage}{0.9\textwidth}
      \footnotesize
      \setstretch{1.0}
      \textit{Note: } Each panel shows predicted and observed enrollment rates
      for one of the four groups.
      The fitted enrollment rates use the estimates at the ``Girls, Secondary''
      vertex of Table \ref{table:estimates}.
      Observed enrollment rates are plotted along with 95\% Wilson confidence intervals.
      Both series are computed on the holdout sample of $8,031$ observations.
    \end{minipage}
\end{figure}

To calibrate the values of $\theta_0$ used in our regret calculations,
we estimate the model on a holdout set from the Progresa experiment.
To capture uncertainty in the active set of constraints,
we choose four values of $\theta_0$ corresponding to the vertices of the feasible set
which lie on the budget constraint.
For each vertex, we choose $\theta_0$ by estimating the model on the holdout set
subject to the constraint that the optimal policy implied by $\theta$ is given by the vertex,
and that the marginal effect of changing $\pi$ at this vertex is equalized across the four groups.
This means that at each vertex, the nonnegativity constraints bind weakly,
so changing the subsidy in budget-neutral directions has no first-order effect on welfare.

Table \ref{table:estimates} shows the unrestricted maximum-likelihood estimates of $\theta$ on the holdout set,
as well as the estimates of $\theta$ at each vertex.
We also report the minimum eigenvalue of $H$ at each vertex to verify strong concavity of welfare at these points.
We also conduct a likelihood ratio test of the weak complementarity constraints used to define each vertex.
At the holdout sample size, all vertices can be rejected, but likelihood ratio tests fail to reject these points at
the sample sizes considered in the Monte Carlo exercise.
Figure \ref{figure:enrollment} shows the predicted enrollment rates at the ``Girls, Secondary'' vertex
compared to the observed enrollment rates in the holdout set.

For the Monte Carlo exercise, we sample from a Uniform ball of size $K$
around each vertex $\theta_0$ to generate draws of $\theta$.
$K$ is chosen so that $\theta$ is difficult to distinguish from $\theta_0$ given
the experimental data; it is 3.396 in Mahalanobis units.
This has the interpretation that any $\theta$ in the ball
lies in a $68\%$ joint Wald confidence region around $\theta_0$
with a sample size of $1000$.

The method for computing the optimal design for the second wave of the experiment is as follows.
For each pilot estimate $\hat\theta_{1,n}$, we compute the quadratic approximation $\hat R_n$.
This allows us to compute the feasible value function $\hat V_n(\beta)$ for any $\beta$ by solving a quadratic program.
To optimize the expected value as $\delta_2$ varies, we construct a differentiable
neural network surrogate $\hat V_n^{\mathrm{NN}}$ for $\hat V_n$.
Given this differentiable surrogate, we compute the expected value by Monte Carlo integration over the 
Gaussian distribution of $\hat \beta_2$ conditional on $\hat{\theta}_1$.
The expected value remains a differentiable function of $\delta_2$,
and so we solve for the optimal $\delta_2$ again using Ipopt.

\begin{table}[htbp]
    \centering
    \begin{threeparttable}
      \caption{Expected regret by experimental design}
      \label{table:regret}
      \footnotesize
      \begin{tabular}{lrrrr}
\toprule
Sample Size & Pilot Only & Progresa & Proposed method & Oracle \\
\midrule
\multicolumn{5}{l}{\textit{Panel A: Boys Primary}} \\
250 & 4.441 (0.242) & 3.321 (0.204) & 2.231 (0.158) & 1.784 (0.139) \\
500 & 3.538 (0.211) & 2.644 (0.187) & 0.957 (0.093) & 0.843 (0.085) \\
1000 & 2.327 (0.17) & 1.305 (0.12) & 0.403 (0.046) & 0.402 (0.056) \\
\addlinespace
\multicolumn{5}{l}{\textit{Panel B: Boys Secondary}} \\
250 & 4.366 (0.229) & 3.292 (0.205) & 2.329 (0.162) & 1.877 (0.137) \\
500 & 3.342 (0.199) & 2.544 (0.18) & 1.066 (0.103) & 1.039 (0.102) \\
1000 & 2.252 (0.166) & 1.382 (0.117) & 0.434 (0.046) & 0.458 (0.054) \\
\addlinespace
\multicolumn{5}{l}{\textit{Panel C: Girls Primary}} \\
250 & 4.528 (0.244) & 3.396 (0.21) & 2.474 (0.169) & 1.972 (0.142) \\
500 & 3.481 (0.212) & 2.637 (0.189) & 1.036 (0.103) & 0.9 (0.086) \\
1000 & 2.282 (0.166) & 1.357 (0.12) & 0.423 (0.047) & 0.447 (0.06) \\
\addlinespace
\multicolumn{5}{l}{\textit{Panel D: Girls Secondary}} \\
250 & 4.299 (0.233) & 3.345 (0.204) & 2.125 (0.148) & 1.958 (0.138) \\
500 & 3.279 (0.203) & 2.412 (0.175) & 0.857 (0.079) & 0.96 (0.09) \\
1000 & 2.038 (0.15) & 1.199 (0.107) & 0.473 (0.048) & 0.462 (0.057) \\
\bottomrule
\end{tabular}

      \begin{tablenotes}[flushleft]
        \footnotesize
        \item \textit{Notes:} 
          Each panel reports expected regret in percentage points and Monte Carlo standard errors 
          under a different reference point $\theta_0$, corresponding to a different vertex of the feasible set.
          Each vertex is characterized by a different group of children receiving a positive subsidy.
          Not all replications achieved convergence; expected regret is computed based on the $>99\%$
          of replications that converged for each vertex and sample size.
     \end{tablenotes}
    \end{threeparttable}
\end{table}

\begin{figure}[htbp]
    \centering
    \caption{Variation in experimental designs}
    \label{figure:experimental_design_girls_secondary_n500}
    \includegraphics[width=1.0\textwidth]
    {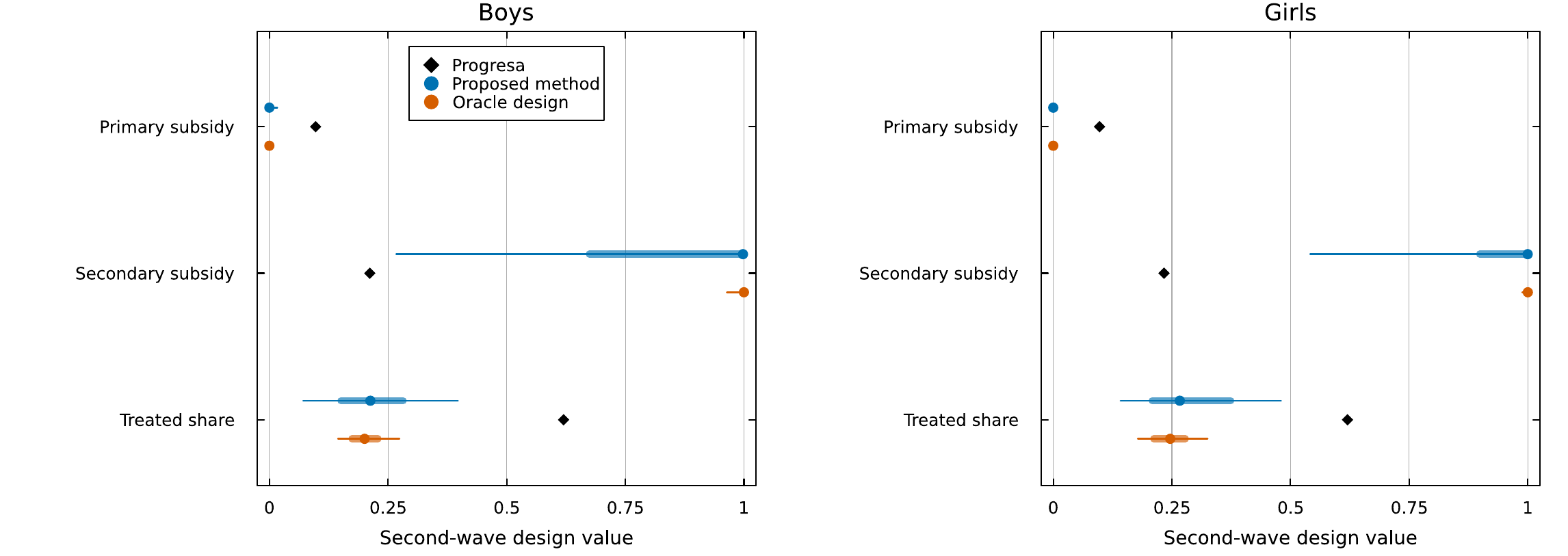}
    \par\smallskip
    \begin{minipage}{1.0\textwidth}
      \footnotesize
      \setstretch{1.0}
      \textit{Note: } The figure describes the second wave designs when $\theta_0$ is at the ``Girls Secondary'' vertex and \(n_1=n_2=500\).
      The figure plots averages of the Progresa design over grades in primary/secondary school,
      but the simulations use the original Progresa design which varies by grade.
      The dot is the median, the thick line is the interquartile range, and the thin line is the 10th-90th percentile range.
      Subsidy amounts are in thousands of pesos.
      This figure is based on 498/500 replications for which the design optimization converged.
    \end{minipage}
\end{figure}

\begin{figure}[htbp]
    \centering
    \caption{Policy errors by experimental design}
    \label{figure:policy_girls_secondary_n500}
    \includegraphics[width=1.0\textwidth]
    {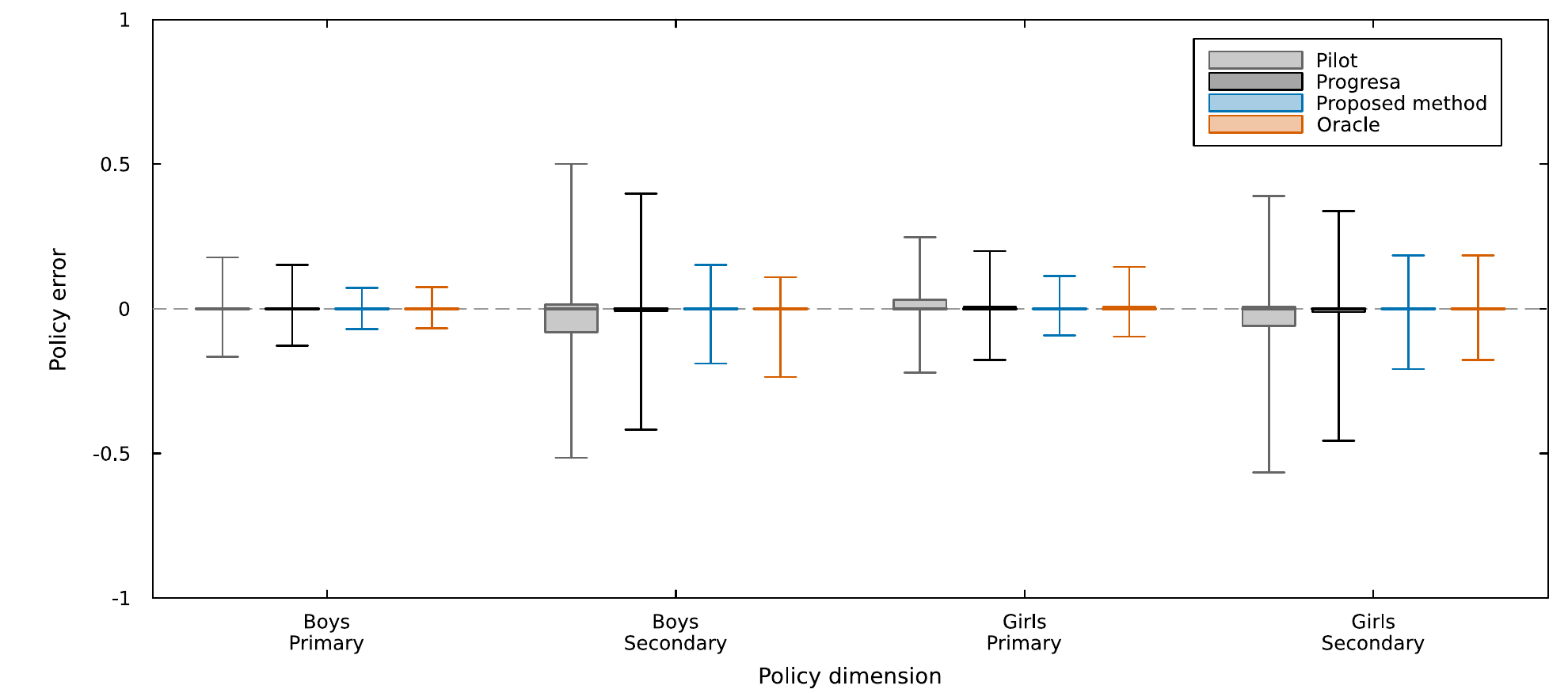}
    \par\smallskip
    \begin{minipage}{1.0\textwidth}
      \footnotesize
      \setstretch{1.0}
      \textit{Note: } The figure reports the distribution of policy errors $\hat \pi_{2,n} - \pi^*(\theta)$
      when $\hat \pi_{2,n}$ is estimated using different experimental designs.
      Results are shown for the case where $\theta_0$ is at the ``Girls Secondary'' vertex and \(n_1=n_2=500\).
      The center line is the median, the box is the interquartile range, and the whiskers are the 10th-90th percentile range.
      Subsidy amounts are in thousands of pesos.
      This figure is based on 498/500 replications for which the design optimization converged.
    \end{minipage}
\end{figure}

Table \ref{table:regret} reports the expected regret of the policy choice under different experimental designs
for every vertex $\theta_0$ and sample size.
Figure \ref{figure:experimental_design_girls_secondary_n500}
illustrates how the optimal experiment differs from the Progresa design,
focusing on the case where $\theta_0$ is at the ``Girls Secondary'' vertex and $n_1=n_2=500$.
Figure \ref{figure:policy_girls_secondary_n500} likewise illustrates how the
policy induced by the optimal design differs from the policy induced by the Progresa design
by reducing the policy error $\hat \pi_{2,n} - \pi^*(\theta)$.

\printbibliography

\end{document}